\documentclass[12pt]{article}

\usepackage{amsmath, amssymb, amsthm,bbm}
\usepackage{natbib} 
\usepackage{xcolor} 

\definecolor{lblue}{RGB}{40, 103, 178}
\definecolor{cred}{RGB}{177, 4, 14}
\definecolor{sgreen}{RGB}{46, 139, 87}

\PassOptionsToPackage{pdftex}{graphicx}
\usepackage{tikz}
\usetikzlibrary{shapes.geometric, arrows}

\usepackage{algorithm,algorithmic,mathtools}

\usepackage{calc} % To reset the counter in the document after title page
\usepackage{enumerate} % Includes lists

\usepackage{parskip}
\usepackage{longtable}
\usepackage{setspace}
\usepackage{array}
\usepackage{multirow}
\usepackage{makecell}
\usepackage{comment}

\usepackage{chngcntr}
\usepackage{fullpage,anyfontsize}
\RequirePackage[colorlinks,linkcolor=blue,citecolor=blue,urlcolor=blue,breaklinks]{hyperref}

\usepackage[all]{nowidow} % Tries to remove widows
\usepackage[protrusion=true,expansion=true]{microtype} % Improves typography, load after fontpackage is selected

\tikzstyle{disc} = [rectangle, rounded corners, 
minimum width=3cm, 
minimum height=1.5cm,
text centered, 
draw=black, 
fill=cred!30]

\tikzstyle{match} = [rectangle, rounded corners, 
minimum width=3cm, 
minimum height=1.5cm,
text centered, 
draw=black, 
fill=lblue!30]

\tikzstyle{weights} = [rectangle, rounded corners, 
minimum width=3cm, 
minimum height=1.5cm,
text centered, 
draw=black, 
fill=orange!30]

\tikzstyle{arrow} = [very thick,->,>=stealth]

\usepackage{blindtext}
\usepackage{titling}
\usepackage{array}
\usepackage{xcolor}
\newtheorem{theorem}{Theorem}%[section]
\newtheorem{lemma}[theorem]{Lemma}
\newtheorem{proposition}[theorem]{Proposition}
\newtheorem{remark}{Remark}
\newtheorem{definition}{Definition}
\newtheorem{corollary}[theorem]{Corollary}

\newtheorem{assumption}{Assumption}

\newtheorem{innercustomasm}{Assumption}
\newenvironment{customasm}[1]
  {\renewcommand\theinnercustomasm{#1}\innercustomasm}
  {\endinnercustomasm}

\newcommand{\R}{\mathbb{R}}

\newcommand{\EE}[1]{\mathbb{E}\left[{#1}\right]}
\newcommand{\EEst}[2]{\mathbb{E}\left[{#1}\  \middle| \ {#2}\right]}
\newcommand{\VV}[1]{\mathrm{Var}\left({#1}\right)}
\newcommand{\VVst}[2]{\mathrm{Var}\left({#1}\  \middle| \ {#2}\right)}
\newcommand{\VPst}[2]{\mathrm{Var}_{{#1}}({#2})}
\newcommand{\Vp}[2]{\mathrm{Var}_{{#1}}\left({#2}\right)}
\newcommand{\Ep}[2]{\mathbb{E}_{{#1}}\left[{#2}\right]}
\newcommand{\Epst}[3]{\mathbb{E}_{{#1}}\left[{#2}\  \middle| \ {#3}\right]}
\newcommand{\PP}[1]{\mathbb{P}\left\{{#1}\right\}}
\newcommand{\PPst}[2]{\mathbb{P}\left\{{#1}\  \middle| \ {#2}\right\}}
\newcommand{\Ppst}[3]{\mathbb{P}_{{#1}}\left\{{#2}\  \middle| \ {#3}\right\}}
\newcommand{\Pp}[2]{\mathbb{P}_{{#1}}\left\{{#2}\right\}}
\newcommand{\eqd}{\stackrel{\textnormal{d}}{=}}
\newcommand{\One}[1]{{\mathbbm{1}}\left\{{#1}\right\}}
\newcommand{\one}[1]{{\mathbbm{1}}_{{#1}}}
\newcommand{\iidsim}{\stackrel{\textnormal{iid}}{\sim}}

  \newcommand\independent{\protect\mathpalette{\protect\independenT}{\perp}}
\def\independenT#1#2{\mathrel{\rlap{$#1#2$}\mkern2mu{#1#2}}}

\DeclareMathOperator*{\argmax}{argmax}
\newcommand{\ISNR}{\mathrm{ISNR}}
\newcommand{\ISO}{\mathrm{ISO}}

\newcommand{\Xcal}{\mathcal{X}}
\newcommand{\Ycal}{\mathcal{Y}}
\newcommand{\Zcal}{\mathcal{Z}}

\newcommand{\Mcal}{\mathcal{M}}

\newcommand{\bs}{\mathbf{s}}

\newcommand{\bx}{\mathbf{x}}
\newcommand{\bX}{\mathbf{X}}

\newcommand{\bY}{\mathbf{Y}}
\newcommand{\bZ}{\mathbf{Z}}
\newcommand{\bw}{\mathbf{w}}
\newcommand{\bv}{\mathbf{v}}

\newcommand{\RR}{\mathbb{R}}

\newcommand{\bzeta}{\mathbf{\eps}}
\def\bzeta{{\boldsymbol{\zeta}}}
\def\beps{{\boldsymbol{\epsilon}}}

\newcommand{\normal}{\mathcal{N}}
\newcommand{\TV}{\textnormal{d}_{\textnormal{TV}}}

\definecolor{orange}{RGB}{255,127,0}

\hypersetup{ 	
pdfsubject = {},
pdftitle = {},
pdfauthor = {}
}

\usepackage[capitalize,noabbrev]{cleveref}

\usepackage[blocks, affil-it]{authblk}

\title{Testing conditional independence under isotonicity}

\author[1]{\fontsize{13.8}{13.8} Rohan Hore}
\author[2]{\fontsize{13.8}{13.8} Jake A. Soloff}
\author[3]{\fontsize{13.8}{13.8} Rina Foygel Barber}
\author[4]{\fontsize{13.8}{13.8} Richard J. Samworth}
\affil[1]{Department of Statistics and Data Science, Carnegie Mellon University}
\affil[2]{Department of Statistics, University of Michigan}
\affil[3]{Department of Statistics, University of Chicago}
\affil[4]{Statistical Laboratory, University of Cambridge}

\date{\today}

\begin{document}

\maketitle

\begin{abstract}
We propose a test of the conditional independence of random variables $X$ and~$Y$ given~$Z$ under the additional assumption that $X$ is stochastically nondecreasing in~$Z$.  The well-documented hardness of testing conditional independence means that some further restriction on the null hypothesis parameter space is required. In contrast to existing approaches based on parametric models, smoothness assumptions, or approximations to the conditional distribution of $X$ given $Z$ and/or $Y$ given $Z$, our test requires only the stochastic monotonicity assumption.
Our procedure, called \textnormal{\texttt{PairSwap-ICI}}, determines the significance of a statistic by randomly swapping the $X$ values within ordered pairs of~$Z$ values. The matched pairs and the test statistic may depend on both $Y$ and $Z$, providing the analyst with significant flexibility in constructing a powerful test. Our test offers finite-sample Type~I error control, and provably achieves high power against a large class of alternatives.  
We validate our theoretical findings through a series of simulations and real data experiments.
 
\end{abstract}

\section{Introduction}
Testing conditional independence (CI) has emerged in recent years as a core contemporary statistical challenge for practitioners and theoreticians alike. On the practical side, it underpins the problems of variable selection, graphical modelling and causal inference, while from a
theoretical perspective, it has recently been understood that there is a strong sense in which the problem is fundamentally hard without further restrictions, and this has led to great interest in reasonable additional conditions under which one can establish non-trivial CI tests.

To formalize our setting, consider the CI null hypothesis
\[
H_0^{\textnormal{CI}}: X \independent Y \mid Z,
\]
where $X$ and $Y$ are variables of interest (such as a treatment $X$ and an outcome $Y$), while~$Z$ represents a confounder.  Our available data consist of independent copies   $(X_1,Y_1,Z_1),\ldots,(X_n,Y_n,Z_n)$ of $(X,Y,Z) \sim P$, where $P$ is an unknown distribution on $\Xcal\times\Ycal\times\Zcal$. Throughout this paper, we write $P_{X\mid Z}$ and $P_{Y\mid Z}$ to denote the conditional distributions of $X$ given $Z$ and of $Y$ given $Z$ respectively.

Given the clear practical relevance of conditional independence testing, as well as the well-established literature on unconditional independence testing, the following hardness result due to \citet{shah2020hardness} initially comes as a big surprise.
\begin{theorem}[\citealp{shah2020hardness}]\label{thm:hardness_CI}
    Let $d \in \mathbb{N}$ and let $\mathcal{P}_{\textnormal{Cont}}$ denote the set of distributions of $(X,Y,Z)$ on $\R^{d}$ that are absolutely continuous with respect to Lebesgue measure.  For any $\alpha\in (0,1)$, any test of the null $H_0^{\textnormal{CI}} \cap \mathcal{P}_{\textnormal{Cont}}$ with Type~I error level $\alpha$ has power no greater than $\alpha$ at every alternative distribution in $\mathcal{P}_{\textnormal{Cont}}\setminus H_0^{\textnormal{CI}}$.
\end{theorem}
In other words, any valid test of the conditional independence null is powerless at \emph{any} alternative distribution.
We also refer the reader to \citet{neykov2021minimax} and \citet{kim2022local} for further discussion and results on the fundamental hardness of constructing distribution-free tests for CI. Naturally, these hardness results have led to growing interest in identifying natural and minimal assumptions under which non-trivial CI tests can be developed. Existing approaches typically aim to guarantee validity (Type~I error control) over restricted classes of null distributions that impose one of the following additional structures:
\begin{enumerate}
    \item[(a)] a parametric model, such as joint Gaussianity of~$(X, Y, Z)$, or a Gaussian linear model for $Y\mid (X,Z)$ \citep[e.g.,][]{kalisch2007estimating,Epskamp04072018,colombo12};
    \item[(b)] a known or well-estimated conditional distribution $P_{X\mid Z}$ \citep{candes2018panning,barber2020robust,berrett2020conditional,niu2024reconciling}; or 
    \item[(c)] smoothness of the conditional distribution $P_{X\mid Z}$ \citep{shah2020hardness,lundborg2022conditional,kim2022local,lundborg2022projected}, e.g.~it belongs to a H\"older class.
\end{enumerate}

A common limitation of these approaches is that these structural assumptions require specifying information about the underlying model or data distribution. In contrast, we will consider a different setting, where we restrict the null by placing a shape constraint on the conditional distribution $P_{X\mid Z}$ (or alternatively $P_{Y\mid Z}$), which requires neither estimating any distributions nor selecting any tuning parameter, and is typically motivated by the application.

\subsection{Our approach}
In this work, we introduce a nonparametric structure under which we can test conditional independence: we assume a shape constraint---specifically, a form of stochastic monotonicity---for the conditional distribution of $X\mid Z$. A common application of conditional independence testing is to determine whether one or more covariates are needed in a regression model for a response after accounting for confounders $Z$.  In our context, and by analogy with the `Model-$X$' framework of \citet{candes2018panning}, it is convenient to regard $X$ as our response and $Y$ as the covariates of interest.  Typically in such practical applications, $X$ is univariate, and our initial exposition will therefore focus on this case, while allowing $Y$ and $Z$ to take values in general measurable spaces $\mathcal{Y}$ and $\mathcal{Z}$ (the space $\mathcal{Z}$ should be equipped with a partial order).  Later, in Section~\ref{sec:extension}, we demonstrate how our framework, methodology and theory can be extended to accommodate multivariate $X$.

Our stochastic monotonicity assumption can be formulated as follows:
 
\begin{customasm}{1}[Monotonicity of the conditional distribution $P_{X\mid Z}$]\label{asm:st}
   Let $\Xcal = \R$, and $\preceq$ be a partial order on $\Zcal$.  We assume $X$ is stochastically nondecreasing in $Z$, meaning that
    \[
    \textnormal{if $z\preceq z'$ then }x
    \Ppst{P}{X\ge x}{Z=z}\le \Ppst{P}{X\ge x}{Z=z'}\textnormal{ for all $x$}.
    \]
\end{customasm}

Assumption~\ref{asm:st} is motivated by several applications, particularly in biomedicine, where, for instance, factors such as smoking intensity may be associated with increased risk of certain diseases or conditions.  We observe similar trends in economics, where a greater education level is often linked to higher salaries, and in credit risk modelling, where better credit scores are associated with higher chances of loan approval.

This assumption does not fall into any of our categories (a)--(c), which summarize existing literature on conditional independence tests, and as illustrated by the examples above, may be naturally motivated by applications. To the best of our knowledge, there are no existing non-trivial CI tests under such nonparametric shape constraints.  

In Section~\ref{sec:method}, we establish that our proposed framework provides a valid CI test for any space~$\mathcal{Z}$ equipped with a partial ordering under which Assumption~\ref{asm:st} holds. For example, when $\mathcal{Z}=\mathbb{R}^{d_Z}$, a common choice is the coordinate-wise order, which is appropriate for instance in settings where increases in any of the risk factors such as age, smoking intensity or exposure to air pollution are associated with higher chances of developing certain diseases. 

Similar to the general framework of (a)--(c), we address the hardness of conditional independence testing via a restricted null hypothesis, namely the isotonic conditional independence (ICI) null,
\begin{align}\label{eq:null}
    H_0^{\textnormal{ICI}} : X\independent Y \mid Z,\textnormal{ and $P_{X\mid Z}$ satisfies Assumption~\ref{asm:st}.}
\end{align}
Naturally, this test should only be applied in settings where the monotonicity condition of Assumption~\ref{asm:st} is well-motivated, so that a rejection of $H_0^{\textnormal{ICI}} $ can reasonably be interpreted as evidence that $X\not\independent Y\mid Z$.  However, we emphasise again that some additional assumption beyond $H_0^{\textnormal{CI}}$ is essential for any valid test to have non-trivial power at any alternative.

The idea behind our framework is to consider carefully-chosen pairs of data points $(X_i,Y_i,Z_i)$ and $(X_j,Y_j,Z_j)$ with $Z_i \preceq Z_j$.  Under $H_0^{\textnormal{ICI}}$, we expect $X_i \leq X_j$ more often than not, and a substantial violation of this provides evidence of the influence of $Y$, i.e.~evidence against $H_0^{\textnormal{ICI}}$.  To calibrate the test appropriately under the null, we employ a particular type of permutation test, where permutations are restricted within matched pairs. To set the stage for the notation and ideas underlying our methodology, we briefly review the simpler framework of testing the null hypothesis of marginal independence, $X \independent Y$, and illustrate how permutations facilitate a valid test in that setting.

\subsection{Background: testing independence}\label{sec: marginal independence}

Given a joint distribution $P$ on $\mathcal{X}\times\mathcal{Y}$, let $(X_1,Y_1),\ldots,(X_n,Y_n) \iidsim P$, and write $\bX = (X_i)_{i=1}^n$ and $\bY = (Y_i)_{i=1}^n$.  We can reframe the problem of testing marginal independence as testing whether the entries of~$\bX$ are independent and identically distributed~given~$\bY$. Specifically,   
permutation tests look for violations of exchangeability of $\bX$ given~$\bY$. The general approach proceeds as follows: based on~$\bY$, the analyst chooses any statistic~$T : \Xcal^n \to \R$, with larger values of $T(\bX)$ indicating evidence against the independence null.\footnote{We emphasise that $T(\bX)$ is allowed to depend on both $\bX$ and $\bY$---for instance we may define the function as $T(\bx) = |\sum_{i=1}^n x_i Y_i|$, though we suppress the dependence on $\bY$ in our notation.}  Write $\mathcal{S}_n$ for the set of permutations of $[n] := \{1,\ldots,n\}$, and for $\sigma \in \mathcal{S}_n$, let $T_\sigma = T(\bX^\sigma)$ denote the value of the statistic when the entries of $\bX$ are permuted according to $\sigma$---that is, $\bX^\sigma = (X_{\sigma(1)},\dots,X_{\sigma(n)})$. Finally, define a $p$-value
\[
p = \frac{1}{n!}\sum_{\sigma\in\mathcal{S}_n}\One{T_\sigma\ge T}.
\]
This construction produces a valid $p$-value for \emph{any} choice of test statistic~$T$, and the statistic~$T$ can be tailored to have power against certain specific alternatives.  Indeed, this strategy has been successfully employed to construct independence tests via nearest neighbor distances and mutual information \citep{berrett2019nonparametric}, moment methods \citep{kim2022minimax}, kernels \citep{pfister2018kernel}, the Hilbert--Schmidt independence criterion \citep{albert2022adaptive} and $U$-statistics \citep{berrett2021usp,berrett2021optimal}.

In our work, since we are interested in conditional (rather than  marginal) independence, we will follow a similar strategy, but will use a restricted class of functions~$T$ and a (data-dependent) subgroup of permutations~$\sigma\in \mathcal{S}_n$, both of which respect the stochastic monotonicity condition specified in Assumption~\ref{asm:st}.

\subsection{Our contributions}
In Section~\ref{sec:method}, we introduce our general methodology for testing $H_0^{\textnormal{ICI}}$, the \textnormal{\texttt{PairSwap-ICI}} procedure. Our framework allows practitioners flexibility in specifying a matching, i.e., a sequence of ordered pairs of observations, and a test statistic for each matched pair; these latter quantities are then aggregated across all matched pairs to form the overall test statistic. We calibrate the test by comparing the observed statistic to its counterpart computed on~$X$ values permuted within matched pairs, and prove that the \textnormal{\texttt{PairSwap-ICI}} procedure guarantees finite-sample Type~I error control over the null $H_0^{\textnormal{ICI}}$, regardless of the matching or test statistic chosen.

Notwithstanding our general validity guarantee, the matching and test statistic for the \textnormal{\texttt{PairSwap-ICI}} procedure must be chosen carefully in order to achieve good power against alternatives of interest. Accordingly, in Section~\ref{sec:design}, we provide specific recommendations for these choices, motivated both by consideration of oracle strategies for maximizing power and by heuristics for how these can be approximated in practice. 
This intuition is then formalized in Section~\ref{sec:power}, where in Theorem~\ref{thm:general-power-asymptotic} we provide an asymptotic expression for the power of a general version of the \textnormal{\texttt{PairSwap-ICI}} test under a broad family of regression models for~$X$ conditional on $(Y,Z)$. We further introduce a quantity called the isotonic signal-to-noise ratio ($\ISNR$), and show that it controls the power of any valid test of $H_0^{\textnormal{ICI}}$.  These results turn out to be closely connected, because  we prove in Theorems~\ref{thm:oracle-matching-asymptote} and~\ref{thm:oracle-matching-estimated-mu} that the empirical analog of $\ISNR$ characterizes the power of the \texttt{PairSwap-ICI} test under both oracle and plug-in matching schemes and under the additional assumption that $d_Z=1$ and $\mathcal{Z}$ is equipped with the usual total order on $\mathbb{R}$. 

After extending our methodology and theory to multivariate $X$ in Section~\ref{sec:extension}, our simulations in Section~\ref{sec:simulations} validate our theory, confirming that our test is not excessively conservative and has power to detect departures from $H_0^{\textnormal{ICI}}$ provided that the signal-to-noise ratio exceeds the critical rate determined by our theory.  As an illustration of an application of the \textnormal{\texttt{PairSwap-ICI}} test on real data, in Section~\ref{sec:experiments} we consider a diabetes dataset from the Pima community near Phoenix, Arizona.  One goal here is to assess which of six variables are individual risk factors for diabetes, after controlling for age.  Finally, we conclude with a discussion in Section~\ref{sec:discussion}. Most proofs are deferred to the Appendix.

\section{Methodology}\label{sec:method}

In this section, we take $\mathcal{Y}$ and $\mathcal{Z}$ to be measurable spaces with $\mathcal{Z}$ having a partial order, and let $\mathcal{X} = \mathbb{R}$. Here we present our general procedure, called the \textnormal{\texttt{PairSwap-ICI}} test, for testing the isotonic conditional independence null~$H_0^{\textnormal{ICI}}$.  This builds on the intuition outlined in the introduction, whereby pairs of data points $(X_i,Y_i,Z_i)$ and $(X_j,Y_j,Z_j)$ with large positive differences $X_i - X_j$ and $Z_i\preceq Z_j$ provide evidence against the null.

More formally, based on~$\bY = (Y_i)_{i=1}^n$ and~$\bZ = (Z_i)_{i=1}^n$, and without access to $\bX = (X_i)_{i=1}^n$, the analyst chooses:
\begin{enumerate}[(i)]
    \item A sequence of ordered pairs 
    \[
    (i_{1},j_{1}),\ldots, (i_{L},j_{L})
    \]
    of indices in $[n]$, where all $2L$ entries are distinct.
    We require the pairs to be ordered in the sense that 
    \begin{align}\label{eq:monotonicity}
        Z_{i_\ell}\preceq Z_{j_\ell}
    \end{align}
    for each $\ell \in [L]$.  Such a choice of ordered pairs (with any $L \leq \lfloor n/2 \rfloor$) is referred to as a \emph{matching}, and we denote the set of all possible matchings in $[n]$ satisfying~\eqref{eq:monotonicity} as~$\mathcal{M}_n(\bZ)$. In Section~\ref{sec:design}, we show how matched pairs can be carefully selected to encode information about $\bY$ and thereby enhance the power to detect specific conditional dependencies.
    
    \item A sequence of functions~$\psi_1,\ldots,\psi_L$, where each $\psi_\ell : \Xcal\times \Xcal \to \R$ satisfies the following \emph{anti-monotonicity} property:
    \begin{definition}
    A function $\psi:\Xcal\times\Xcal\to\R$ is said to satisfy anti-monotonicity if 
        \begin{align}
        \label{defn:monotonicity-of-psi}
        \psi(x + \Delta, x' - \Delta') - \psi(x' - \Delta', x + \Delta) \geq \psi(x,x') - \psi(x',x) 
    \end{align}
    for all $x,x'$ and all $\Delta,\Delta'\geq 0$.
    \end{definition}
     An example of a class of functions $\psi$ that satisfy anti-monotonicity is $\psi(x,x')=f(x-x')$ for any monotone nondecreasing function $f$ that is also \emph{odd} in the sense that~$f(x)= -f(-x)$, for instance, $\psi_\ell(x,x') = w_\ell\cdot (x-x')$ or $\psi_\ell(x,x') = w_\ell\cdot \textnormal{sign}(x-x')$.  As with the choice of matching, in Section~\ref{sec:design}, we show how the functions ${\psi_\ell}$ can be selected advantageously; for instance, when $\psi_\ell(x,x') = w_\ell\cdot (x-x')$, dependence on $(\bY,\bZ)$ is simply captured through the weights $w_\ell$.
\end{enumerate}

With these choices in place, our test statistic $T:\mathcal{X}^n\to\R$ is defined as\footnote{ Again, as for marginal permutation tests in Section~\ref{sec: marginal independence}, here we suppress dependence on $\bY$ and $\bZ$ in the notation $T(\bx)$, even though this statistic does depend on $\bY$ and $\bZ$ through the choices of the matched pairs $(i_\ell,j_\ell)$ and functions $\psi_\ell$.}
    \begin{align}\label{eq:test-stat}
        T(\bx) = \sum_{\ell=1}^L \psi_\ell(x_{i_\ell}, x_{j_\ell}).
    \end{align}
In order to calibrate the test, the analyst compares the observed test statistic $T = T(\bX)$ with versions of $T$ where indices within pairs $(i_\ell, j_\ell)$ are randomly swapped. Specifically, for $\bs = (s_1,\ldots,s_L) \in \{\pm 1\}^L$, define $T_\bs = T(\bX^\bs)$, where $\bX^\bs$ is a swapped version of the data vector $\bX$, with entries
\[
\begin{cases}
    (X_{i_\ell}^\bs, X_{j_\ell}^\bs) = (X_{i_\ell}, X_{j_\ell}) & s_\ell = +1, \\
    (X_{i_\ell}^\bs, X_{j_\ell}^\bs) = (X_{j_\ell}, X_{i_\ell}) & s_\ell = -1.
\end{cases}
\]
Thus, $s_\ell=-1$ indicates that the random variables $X_{i_\ell}$ and $X_{j_\ell}$ are swapped in $(X_{i_\ell}^\bs, X_{j_\ell}^\bs)$ compared with $(X_{i_\ell}, X_{j_\ell})$, while $s_\ell=+1$ indicates no swap.
Informally, the two constraints~\eqref{eq:monotonicity} and~\eqref{defn:monotonicity-of-psi} ensure that, under the null, each $\psi_{\ell}(X_{i_\ell},X_{j_\ell})$ is likely to be no larger than its swapped version, $\psi_{\ell}(X_{j_\ell},X_{i_\ell})$---and thus, the statistic $T=T(\bX)$ is likely to be no larger than its swapped copies, $T_{\bs}=T(\bX^{\bs})$. If instead $T>T_{\bs}$ for many swaps $\bs$, this indicates evidence against the null. Indeed, we may define 
\begin{align}\label{eq:def-test}
p&:= \frac{1}{2^L}\sum_{\bs\in\{\pm 1\}^L}\One{T_\bs \ge T},
\end{align}
and in Theorem~\ref{thm:main} we show that this is a valid $p$-value.  In words, this $p$-value compares the observed value~$T$ of the statistic against all possible permuted statistic values $T_{\bs}$ that we would obtain by swapping indices within matched pairs of our observed data $\bX$.

\paragraph{Example: a linear test statistic.}
Before proceeding, we give a simple example of a test statistic that we might choose to use: consider a linear test statistic,
\[T(\bx) = \sum_{i=1}^n \beta_i x_i\]
for some coefficients $\beta_i\in\R$ (which may depend on $\bY$, $\bZ)$. This function can be used as the test statistic for the \textnormal{\texttt{PairSwap-ICI}} test, as long as the coefficients satisfy
\begin{equation}\label{eqn:linear_condition}
\beta_{i_\ell} \geq \beta_{j_\ell}\textnormal{ for each $\ell\in[L]$}.
\end{equation}
To see why, first note that without loss of generality we can take $\beta_i =0$ for all $i\in[n]\setminus \{i_1,j_1,\dots,i_L,j_L\}$, i.e., all data points not belonging to any of the $L$ pairs. This is because the indicator $\One{T_{\bs}\geq T}$, appearing in the computation of the $p$-value, is invariant to these terms. Next,
define
\[
\psi_\ell(x,x') = \beta_{i_\ell} x + \beta_{j_\ell}x',
\]
which satisfies~\eqref{defn:monotonicity-of-psi} because for any $\Delta,\Delta' \geq 0$ with $x,x',x+\Delta,x'-\Delta' \in \mathcal{X}$, we have 
\[
\psi_\ell(x + \Delta, x' - \Delta') - \psi_\ell(x' - \Delta', x + \Delta) - \psi_\ell(x,x') + \psi_\ell(x',x) = (\beta_{i_\ell} - \beta_{j_\ell})(\Delta + \Delta') \geq 0
\]
by our assumption~\eqref{eqn:linear_condition}. We then have $T(\bx)$ equal to the test statistic defined in~\eqref{eq:test-stat}.

We remark that choosing such a test statistic is by no means implying an assumption that the dependence between $X$ and $Y$ given $Z$ follows a linear model
---it may be the case that the statistic $T(\bx)=\beta^\top\bx$ has good power for distinguishing the null from the alternative even if a linear model is only a coarse approximation to the true  model.

\subsection{Validity}\label{sec:validity}

Our first main result establishes that our method yields a valid test of~$H_0^{\textnormal{ICI}}$.

\begin{theorem}\label{thm:main}
    Under $H_0^{\textnormal{ICI}}$, the conditional Type~I error of the \textnormal{\texttt{PairSwap-ICI}} test satisfies $\PPst{p\le \alpha }{ \bY, \bZ}\le \alpha$ for all~$\alpha \in [0,1]$. In particular, the test enjoys marginal error control: $\PP{p\leq \alpha}\leq \alpha$ for all $\alpha$.
\end{theorem}

\begin{proof}
Our proof is split into three steps. First we derive some deterministic properties of the $p$-value. Next, we compare to the \emph{sharp null}, where the pair $X_{i_\ell},X_{j_\ell}$ are identically distributed (rather than stochastically ordered). Finally, we examine the validity of the test under the sharp null.

\paragraph{Step 1: some deterministic properties of the $p$-value.} First, fix $\alpha \in [0,1]$ and define a function $p:\R^n\to[0,1]$ as
\begin{align}\label{eq:pvalue-as-a-function-of-data}
    p(\bx) = \frac{1}{2^L}\sum_{\bs\in\{\pm 1\}^L}\One{T(\bx^{\bs})\geq T(\bx)},
\end{align}
so that the $p$-value in~\eqref{eq:def-test} can be written as
$p = p(\bX)$. For each $\bs\in\{\pm 1\}^L$, we can observe that the value $p(\bx^{\bs})$ is simply computing the quantile of $T(\bx^{\bs})$ among all possible swapped statistics, $\big(T(\bx^{\bs'}) : \bs'\in\{\pm 1\}^L\big)$. Consequently, it holds deterministically that
\begin{equation}\label{eqn:pvalue_deterministic_quantile}\frac{1}{2^L}\sum_{\bs\in\{\pm 1\}^L} \One{p(\bx^{\bs}) \leq \alpha}\leq \alpha.
\end{equation}
(Lemma \ref{lem:deterministic_control_on_perm_pvalue} in the Appendix verifies this bound, for completeness.)

In addition, we claim that $p(\cdot)$ is monotone in its coordinates, namely, $p(\bx)$ is nonincreasing in each $x_{i_\ell}$, and nondecreasing in each $x_{j_\ell}$. To see why, for each $\bs\in\{\pm 1\}^L$, we can calculate
\begin{align*}
    \One{T(\bx^\bs)\geq T(\bx)}
    &=\One{\sum_{\ell : s_\ell = +1} \psi_\ell(x_{i_\ell},x_{j_\ell}) +\sum_{\ell : s_\ell = -1} \psi_\ell(x_{j_\ell},x_{i_\ell}) \geq \sum_{\ell=1}^L \psi_\ell(x_{i_\ell},x_{j_\ell})}\\
    &=\One{\sum_{\ell : s_\ell = -1} \bigl(\psi_\ell(x_{i_\ell},x_{j_\ell}) - \psi_\ell(x_{j_\ell},x_{i_\ell})\bigr) \leq 0}.
\end{align*}
By the anti-monotonicity condition~\eqref{defn:monotonicity-of-psi} on  $\psi_\ell$, this function is nonincreasing in each $x_{i_\ell}$, and nondecreasing in each $x_{j_\ell}$, and therefore the same is true for $p(\bx)$ as well.

\paragraph{Step 2: compare to the sharp null.} For each $i\in[n]$, let $P_i = P_{X|Z}(\cdot\mid Z_i)$ denote the null distribution of $X_i$ (after conditioning on $\bY,\bZ$). 
By Assumption~\ref{asm:st}, we know that $P_{i_\ell}\preceq_{\rm st}P_{j_\ell}$ for each pair $\ell\in[L]$, where $\preceq_{\rm st}$ denotes the stochastic ordering on distributions. Next, we also define distributions $\bar{P}_\ell$ for each $\ell\in[L]$, given by the mixture
\[\bar{P}_\ell = \frac{1}{2}P_{i_\ell}+ \frac{1}{2}P_{j_\ell}.\] In particular, then, conditional on $\bY, \bZ$,
\begin{equation}\label{eqn:stochastic_order_for_mix}P_{i_\ell}\preceq_{\rm st}\bar{P}_\ell\preceq_{\rm st}P_{j_\ell}, \quad \ell\in[L].\end{equation}
We will now compare the observed data values, whose distribution (conditional on $\bY,\bZ$) is given by
\[\bX = (X_1,\dots,X_n)\ \sim \  P_1\times \dots \times P_n,\]
against a different distribution,
\[\bX_{\sharp} = \bigl((X_\sharp)_1,\dots,(X_\sharp)_n\bigr) \ \sim \  (P_\sharp)_1 \times \dots\times (P_\sharp)_n,\]
where the distributions $(P_\sharp)_i$ are defined by setting
\[(P_\sharp)_{i_\ell} = (P_\sharp)_{j_\ell} = \bar{P}_\ell\]
for each $\ell\in[L]$ (and, for any index $i\in[n]\setminus \{i_1,j_1,\dots,i_L,j_L\}$ that does not belong to any of the $L$ matched pairs, we simply take $(P_\sharp)_i=P_i$).
We can think of this alternative vector of observations as being drawn from a \emph{sharp null}, because for each pair $\ell$, the random variables $(X_\sharp)_{i_\ell},(X_\sharp)_{j_\ell}$ are identically distributed (rather than stochastically ordered, as for $X_{i_\ell},X_{j_\ell}$). In particular, this implies that for any $\bs\in\{\pm 1\}^L$,
\begin{equation}\label{eqn:eqd_under_sharp_null}(\bX_\sharp)^\bs \eqd \bX_\sharp\end{equation}
after conditioning on $\bY,\bZ$. 

In Step 1, we verified that the function $p(\bx)$ is nonincreasing in each $x_{i_\ell}$, and nondecreasing in each $x_{j_\ell}$. In particular, combined with the stochastic ordering~\eqref{eqn:stochastic_order_for_mix}, this means that
$p(\bX_\sharp)\preceq_{\rm st} p(\bX) $
 (conditional on $\bY,\bZ$). We therefore have
\[\PPst{p\leq \alpha}{\bY,\bZ} = \PPst{p(\bX)\leq \alpha}{\bY,\bZ}\leq \PPst{p(\bX_\sharp)\leq \alpha}{ \bY,\bZ}.\]
From this point on, then, we only need to verify the validity of the $p$-value $p(\bX_\sharp)$ computed under the sharp null.

\paragraph{Step 3: validity under the sharp null.} 
For  data $\bX_\sharp$ drawn under a sharp null, we have
\begin{multline*} \PPst{p(\bX_\sharp)\leq \alpha}{\bY,\bZ}= \frac{1}{2^L}\sum_{\bs\in\{\pm 1\}^L}\PPst{p\bigl((\bX_\sharp)^\bs\bigr) \leq \alpha}{\bY,\bZ} \\= \EEst{\frac{1}{2^L}\sum_{\bs\in\{\pm 1\}^L}\One{p\bigl((\bX_\sharp)^\bs\bigr) \leq \alpha} }{\bY,\bZ} \leq \alpha,\end{multline*}
where the first step holds by~\eqref{eqn:eqd_under_sharp_null}, while the last step holds by the deterministic calculation~\eqref{eqn:pvalue_deterministic_quantile} from Step 1. 
\end{proof}

The $p$-value constructed in~\eqref{eq:def-test} requires computing~$T_\bs = T(\bX^\bs)$ for all~$2^L$ values of~$\bs\in \{-1,1\}^L$, which may be computationally prohibitive for moderate or large~$L$. In practice, it is common to use a Monte Carlo approximation to the $p$-value: we sample~$\bs^{(1)}, \ldots, \bs^{(M)}\iidsim \textnormal{Unif}(\{\pm1\}^L)$, and then compute 
\[
\hat{p}_M = \frac{1+\sum_{m=1}^M \One{T_{\bs^{(m)}}\ge T}}{1+M}.
\]
The extra~`$1+$' term appearing in the numerator and denominator is necessary to ensure error control for this Monte Carlo version of our test \citep{davison1997bootstrap,phipson2010permutation}; in particular, this correction ensures we cannot have $\hat{p}_M=0$. The following theorem verifies that this version of the test also controls the Type~I error.

\begin{theorem}\label{thm:mc-validity}
 Fix any $M\in \mathbb{N}$. Under $H_0^{\textnormal{ICI}}$, it holds that $\PPst{\hat{p}_M\le \alpha }{\bY,\bZ}\le \alpha$ for all~$\alpha \in [0,1]$, and consequently, $\PP{\hat{p}_M\leq \alpha}\leq \alpha$.
\end{theorem}

We emphasise that the Type~I error control of the \texttt{PairSwap-ICI} method holds regardless of the dimensions of $\Ycal,\Zcal$; in fact, these may even be infinite-dimensional.  In Section~\ref{sec:extension} we show that the same validity guarantee holds under our extended framework that permits $X$ to be multivariate.  
\section{Designing a powerful \textnormal{\texttt{PairSwap-ICI}} test}\label{sec:design}

 Our \texttt{PairSwap-ICI} methodology from Section~\ref{sec:method} offers great flexibility to the analyst, who may select any pairs $(i_\ell,j_\ell)$ subject to the ordering constraint~\eqref{eq:monotonicity}, and any functions~$\psi_\ell$ satisfying anti-monotonicity~\eqref{defn:monotonicity-of-psi}. In particular, they can decide on these aspects of the test \textit{after} exploring the data~$\textbf{Y}, \textbf{Z}$, choosing to include a pair $(i_\ell,j_\ell)$ if the data observed in $\bY$ indicates that $X_{i_\ell}>X_{j_\ell}$ would be likely under the alternative of interest. However, the quality of the matches~$(i_\ell, j_\ell)$ and functions~$\psi_\ell$ affects the power of our test.  The aim of this section, then, is to construct a principled, powerful implementation of our test, by designing concrete choices for the pairs $(i_\ell,j_\ell)$ and the functions $\psi_\ell$ introduced in Section~\ref{sec:method}. Throughout, we will restrict our attention to test statistics $T(\bx)$ of the form 
\begin{align}\label{def:sum-of-psi}
    T(\bx) = \sum_{\ell=1}^L w_\ell\, \psi(x_{i_\ell},x_{j_\ell})
\end{align}
for some $w_1,\ldots,w_L  \geq 0$.  That is, in the original definition of the test statistic~\eqref{eq:test-stat}, we take $\psi_\ell(\cdot) = w_\ell\psi(\cdot)$ for some vector of non-negative weights $\bw = (w_\ell)_{\ell=1}^L$ and some \emph{fixed} kernel~$\psi$, which is required to satisfy the anti-monotonicity property~\eqref{defn:monotonicity-of-psi}.

With this simplification, designing a test statistic now requires specifying the kernel $\psi$, deciding which pairs~$(i_\ell, j_\ell)$ are matched, and finally, how much weight~$w_\ell$ to assign to each pair, as depicted in this flowchart:

\begin{center}
\begin{tikzpicture}
\node (start) [disc] {Kernel $\psi$};
\node (mid) [match, right of=start, xshift=1.7in] {Matching $M=\{(i_\ell, j_\ell)\}_{\ell=1}^L$};
\node (end) [weights, right of=mid, xshift=1.9in] {Weights $\bw = (w_\ell)_{\ell=1}^L$};
\draw [arrow] (start) -- (mid);
\draw [arrow] (mid) -- (end);
\end{tikzpicture}
\end{center}

As guaranteed by Theorem~\ref{thm:main}, our test controls the Type~I error for any choice of $\psi$ satisfying~\eqref{defn:monotonicity-of-psi}, $M\in\Mcal_n(\bZ)$, and non-negative weight vector $\bw$. However, for the test to be effective, we need to tailor these choices to the specific application of interest. 
Of course, all of these choices interact with each other: what constitutes a good matching depends on how we choose the weights, and vice versa. 

\subsection{Specifying the kernel \texorpdfstring{$\psi$}{psi}}

We begin by considering several simple options for the kernel $\psi$.  As a first example, consider $\psi(x,x') = x-x'$. This choice of $\psi$ means that $\psi(X_{i_\ell},X_{j_\ell})$ is likely to be non-positive under the null (since $Z_{i_\ell}\preceq Z_{j_\ell}$), but under the alternative, may be likely to be positive and large (if the pair $(i_\ell,j_\ell)$ is chosen wisely). Alternatively, we may consider nonlinear test statistics to handle a broader range of settings. If $X$ has heavy tails, then the distribution of a linear statistic~$T$ can be very sensitive to extreme values.  We can ameliorate this sensitivity by using $\psi(x,x') = \textnormal{sign}(x-x')$, or $\psi(x,x') = (-K) \vee (x-x') \wedge K$  for some constant $K > 0$ (i.e., the truncation of $x-x'$ to some bounded interval $[-K,K]$).

Each of the above examples of kernels $\psi$ satisfies the required anti-monotonicity property~\eqref{defn:monotonicity-of-psi}.  In fact, we may take any such anti-monotonic kernel to be also anti-symmetric.  To see this, observe that we can always write
\[
    \psi(x,x') = \frac{1}{2}\bigl\{\psi(x,x') + \psi(x',x)\bigr\} + \frac{1}{2}\bigl\{\psi(x,x') - \psi(x',x)\bigr\} =: \psi_{\mathrm{sym}}(x,x') + \psi_{\mathrm{anti}}(x,x'), 
\]
where by construction, $\psi_{\mathrm{sym}}$ is a symmetric function while $\psi_{\mathrm{anti}}$ is anti-symmetric.  The function $\psi_{\mathrm{sym}}$ does not contribute to the $p$-value---that is, the $p$-value would be unchanged if we replace $\psi$ with $\psi_{\mathrm{anti}}$ (for example, our test with the kernel $\psi(x,x')=x-x'$ coincides with the special case of linear test statistic from Section~\ref{sec:method}). This establishes our claim, and henceforth we therefore assume without loss of generality that our anti-monotonic kernel $\psi$ is also anti-symmetric.

\subsection{Oracle strategies for choosing the weights and matching}\label{sec:oracle}

We now build intuition for how to choose the weight vector $\bw$ and matching $M$ effectively by sketching the asymptotics of our test, assuming some oracle knowledge (or estimates) of the data distribution; see Section~\ref{sec:power} for a more formal analysis.  Let us consider any statistic~$T$ of the form~\eqref{def:sum-of-psi}, where (as discussed above) without loss of generality we can take $\psi$ to be an anti-symmetric function. 
We now wish to choose the weights $\bw = (w_\ell)_{\ell\in [L]}$ and matching $M = \{(i_\ell, j_\ell)\}_{\ell\in [L]}$ to maximize the power of our test. 
Given the data $(\bX, \bY, \bZ)$, the reference statistic $T_\bs$ can be expressed as $T_\bs = \sum_{\ell=1}^L s_\ell \cdot w_\ell\psi(x_{i_\ell},x_{j_\ell})$ by construction of the swapping operation, and is therefore a sum of $L$ independent random variables. 
Under some regularity conditions on the weights~$\bw$ and the function $\psi$, a central limit theorem (CLT) approximation gives, for large $L$, that
\begin{align}\label{eqn:CLT_approx_p_from_sec:oracle}
    p 
    &\approx 1 - \Phi(\hat{T}) 
    \qquad\textnormal{where}\qquad \hat{T} = \hat{T}(\bw,M)
    :=\frac{\sum_{\ell=1}^Lw_\ell\psi(X_{i_\ell},X_{j_\ell})}
    {\sqrt{\sum_{\ell=1}^L w_\ell^2\psi(X_{i_\ell},X_{j_\ell})^2}},
\end{align}
where $\Phi$ is the standard Gaussian distribution function. The above approximation holds for fixed~$\bX,\bY,\bZ$ and relies only on the CLT approximation for a weighted sum of $L$ independent signs $s_1,\dots,s_L\in \{\pm 1\}$.  We should therefore aim to choose weights that \emph{maximize} the approximate probability of rejection, $\mathbb{P}\bigl\{1 - \Phi(\hat{T})\leq \alpha \mid \bY,\bZ\bigr\}$, in order to achieve the best possible power. Under some conditions this conditional power can be further approximated~as
\begin{equation}
\label{eqn:estimate_power}
\Phi\Biggl(\Phi^{-1}(\alpha)+\frac{\sum_{\ell=1}^L w_\ell \EEst{\psi(X_{i_\ell},X_{j_\ell})}{\bY,\bZ}}{\sqrt{\sum_{\ell=1}^L w_\ell^2 \VVst{\psi(X_{i_\ell},X_{j_\ell})}{\bY,\bZ}}}\Biggr)
\end{equation}
(see Theorem \ref{thm:general power analysis} in the Appendix for a closer look at this approximation).
The following lemma shows how to maximize this approximation over the weights, treating the matching~$M$ as fixed.
\begin{lemma}
\label{Lemma:OracleWeights}
Assume that $\mathrm{Var}\bigl(\psi_\ell(X_{i_\ell},X_{j_\ell}) \mid \bY,\bZ\bigr) > 0$ for $\ell \in [L]$. Considered as a function of $\bw = (w_1,\ldots,w_L) \in [0,\infty)^L$, the function in~\eqref{eqn:estimate_power} is maximized by the choice
\begin{equation}\label{eq:oracle-weights}
w^*_\ell =\frac{\max\left\{\EEst{\psi(X_{i_\ell},X_{j_\ell})}{\bY,\bZ},0\right\}}{\mathrm{Var}\bigl(\psi_\ell(X_{i_\ell},X_{j_\ell}) \mid \bY,\bZ\bigr)},
\end{equation}
for $\ell \in [L]$.
\end{lemma}
\begin{proof}
Writing $
u_\ell = \EEst{\psi(X_{i_\ell},X_{j_\ell})}{\bY,\bZ}$ and $v_\ell = \bigl[\mathrm{Var}\bigl(\psi(X_{i_\ell},X_{j_\ell}) \mid \bY,\bZ\bigr)\bigr]^{1/2}$ for convenience, for any $w\in[0,\infty)^L$ we calculate
\[ \frac{\sum_{\ell=1}^L w_\ell u_\ell}{\sqrt{\sum_{\ell=1}^L w_\ell^2 v_\ell^2}} \leq \frac{\sum_{\ell=1}^L w_\ell \max\{u_\ell,0\}}{\sqrt{\sum_{\ell=1}^L w_\ell^2 v_\ell^2}} \\
= \frac{\sum_{\ell=1}^L (w_\ell v_\ell) \cdot \frac{\max\{u_\ell,0\}}{v_\ell}}{\sqrt{\sum_{\ell=1}^L (w_\ell v_\ell)^2}}\leq \biggl(\sum_{\ell=1}^L \frac{\max\{u_\ell,0\}^2}{v_\ell^2}\biggr)^{1/2},
\]
with equality if and only if $w_\ell \propto \max\{u_\ell,0\}/v_\ell^2$ for $\ell \in [L]$.
\end{proof}

\paragraph{Using a plug-in estimate for the moments.}
In Lemma~\ref{Lemma:OracleWeights}, the oracle weight vector $\bw^*$ depends on the conditional expected value and variance, $\EEst{\psi(X_{i_\ell}, X_{j_\ell})}{\bY, \bZ}$ and $\mathrm{Var}\bigl(\psi(X_{i_\ell}, X_{j_\ell}) \bigm| \bY, \bZ\bigr)$.  For the remainder of this subsection, we will assume that we have access to estimates $\hat{E}_{ij}$ of $\EEst{\psi(X_{i_\ell}, X_{j_\ell})}{\bY, \bZ}$ and $\hat{V}_{ij} > 0$ of $\mathrm{Var}\bigl(\psi(X_i,X_j) \bigm| \bY,\bZ\bigr)$, for each $i,j \in [n]$, constructed from data that are independent of $\bX$.  For example, these might be obtained from a fitted model for the conditional distribution of $X$ given $Y,Z$ based on a separate dataset (e.g.~via sample splitting).  For the linear kernel $\psi(x,x')=x-x'$, this would amount to estimating the first two moments of $X_i\mid Y_i, Z_i$ for $i \in [n]$, whereas if $\psi(x,x') = \textnormal{sign}(x-x')$, then we would need to estimate $\PPst{X_i > X_j }{ Y_i, Z_i,Y_j,Z_j}$ for distinct indices $i,j$.  With these estimates in place, we can seek to maximize the power of our test by choosing weights
\[
\hat{w}_\ell =\frac{\max\{\hat{E}_{i_\ell j_\ell},0\}}{\hat{V}_{i_\ell j_\ell}}
\]
for $\ell \in [L]$. We emphasise that even highly inaccurate estimates of the oracle weights still yield a valid test of $H_0^{\textnormal{ICI}}$, though accurate estimates of the conditional moments above enhance the power of our test.  This is in contrast to methods of type (b) and (c) from the introduction, where the quality of estimates of (aspects of) the conditional distribution of $X$ given $(Y,Z)$, or of $X$ given $Z$, affects their Type~I error control.

\paragraph{Choosing the matching.}
 If the oracle weights in~\eqref{eq:oracle-weights} were known, then the estimator~\eqref{eqn:estimate_power} of the test's conditional power
would be maximized by solving a maximum-weight matching problem, namely that of finding 
\begin{align}\label{def:oracle-matching}
        M^* \in \argmax_{M \in \mathcal{M}_n(\bZ)}\sum_{(i,j)\in M} \bigl(W_{ij}^*\bigr)^2 \ \ \textnormal{where}\ \  
    W_{ij}^* = \frac{\max\left\{\EEst{\psi(X_i,X_j)}{\bY,\bZ},0\right\}}{\VVst{\psi(X_i,X_j)}{\bY,\bZ}}.
\end{align}
We refer to this version of our procedure as \emph{oracle matching}.
In practice, when the moments of the distribution are estimated rather than known, we can instead use the plug-in estimates 
\begin{align}\label{def:plug-in-matching}
        \hat{M} \in \argmax_{M \in \mathcal{M}_n(\bZ)} \sum_{(i,j)\in M} \hat{W}_{ij}^2\ \ \textnormal{where} \  \hat{W}_{ij} = \frac{\max\{\hat{E}_{ij}, 0\}}{\hat{V}_{ij}}.
\end{align}
To run \textnormal{\texttt{PairSwap-ICI}}, we then take weights $w_\ell = \hat{W}_{i_\ell j_\ell}$ for each pair $(i_\ell,j_\ell) \in \hat{M}$.
This plug-in matching $\hat{M}$ can be computed efficiently in polynomial time \citep{edmonds1965paths, duan2014linear} using an algorithm of \cite{gabow1985scaling}.

\subsection{Heuristic strategies for choosing the matching and weights}\label{sec:heuristic matching}

In Section~\ref{sec:oracle} above, the test statistic $T(\bx)$ was designed with the aim of maximizing the power of our test, and we argue in Theorem~\ref{thm:oracle-matching-estimated-mu} in Section~\ref{sec:power} that this strategy is particularly effective when we can construct good estimates of the first two conditional moments of our data.  This may be possible if, for instance, we have been able to fit a model using a separate dataset; naturally, an accurate estimate of the true model can enable a very powerful test (since we have a good approximation of the alternative that we are testing against).  The aim of this section, on the other hand, is to propose an alternative approach designed for settings when we may not be able to estimate reliably the required conditional moments, or because the computations are too costly (in particular, the matching strategy in~\eqref{def:plug-in-matching}).  We demonstrate a simple scheme for designing a weight vector $\bw$, and two easy strategies for choosing a matching $M$, that do not require extensive prior knowledge or costly calculations.  Of course, this will come at some cost in terms of the resulting power of the test, since we are no longer mimicking an oracle test---but, as we will see in both our theoretical guarantees and our empirical results below, these simple strategies can often attain high power nonetheless.

Throughout this section, we will restrict our attention to the univariate setting where $\mathcal{X}, \mathcal{Y}, \Zcal = \R$, and will also choose the kernel $\psi(x,x') = x-x'$ when defining our test statistic as in~\eqref{def:sum-of-psi}. We will work in the setting where we hypothesise that, under the alternative, there is a \emph{positive} association between $X$ and $Y$ even after controlling for $Z$ (of course, if our hypothesis is a negative association, we can follow an analogous strategy). The idea is simple: we will take pairs $(i_\ell,j_\ell)$ such that
\begin{itemize}
    \item $Z_{i_\ell} \leq Z_{j_\ell}$ (as required for validity), but $Z_{i_\ell} \approx Z_{j_\ell}$; and
    \item $Y_{i_\ell} > Y_{j_\ell}$ 
\end{itemize}
 Here, the choice $Z_{i_\ell} \approx Z_{j_\ell}$ is desirable to ensure that the joint conditional distributions of $(X_{i_\ell},Y_{i_\ell})$ given $Z_{i_\ell}$ and $(X_{j_\ell},Y_{j_\ell})$ given $Z_{j_\ell}$ should be similar.  In this case, it is anticipated that under $H_0^{\mathrm{ICI}}$, we should observe $\sum_{\ell=1}^L (X_{i_\ell} - X_{j_\ell}) \approx 0$, and since $Y_{i_\ell} > Y_{j_\ell}$ for each $\ell \in [L]$, if $\sum_{\ell=1}^L (X_{i_\ell} - X_{j_\ell}) \gg 0$, then this provides evidence against $H_0^{\mathrm{ICI}}$.  On the other hand, if we were to choose pairs with $Z_{i_\ell} \ll Z_{j_\ell}$, then these joint conditional distributions may be rather different, and the stochastic monotonicity assumption might mean that $\sum_{\ell=1}^L (X_{i_\ell} - X_{j_\ell}) \gg 0$ occurs only rarely even when the conditional independence is violated.  In other words, if we choose pairs with $Z_{i_\ell} \ll Z_{j_\ell}$, then under the stochastic monotonicity assumption, we retain Type~I error control, but lose power relative to choosing pairs with $Z_{i_\ell} \approx Z_{j_\ell}$.
\subsubsection{A simple weighting scheme}\label{sec:simple_matching}
We begin by defining a mechanism for choosing the weights: we will take
\[
w_\ell = \big(Y_{i_\ell} - Y_{j_\ell}\big)_+,
\]
which is simply the difference in $Y$ values, whenever this difference is positive for the pair $(i_\ell,j_\ell)$.

Why is this simple strategy a reasonable choice across broad settings? With this choice of weights, our test statistic $T(\bX)$ is then given by
\[
T(\bX) = \sum_{\ell : Y_{i_\ell} > Y_{j_\ell}}  (Y_{i_\ell} - Y_{j_\ell}) \cdot (X_{i_\ell} - X_{j_\ell} ).
\]
Since under the alternative, we expect $X_{i_\ell} > X_{j_\ell}$ whenever $Y_{i_\ell}>Y_{j_\ell}$, this means that the expected value of $T(\bX)$ is positive under the alternative (but, of course, is non-positive under the null).

As another motivation, consider the partially linear model where 
\[
\EEst{X}{Y=y,Z=z} = \beta^*y + \mu^*_Z(z),
\]
and the conditional variance is constant,
\[\VVst{X}{Y=y,Z=z} = \sigma^*{}^2.\]
Then, following the oracle strategy of Section~\ref{sec:oracle}, as in~\eqref{eq:oracle-weights} we have
oracle weights 
\begin{align*}
    w^*_\ell = \frac{\max\{\EEst{X_{i_\ell}-X_{j_\ell}}{\bY,\bZ},0\}}{\VVst{X_{i_\ell}-X_{j_\ell}}{\bY,\bZ}} &=  \frac{\max\bigl\{\beta^*(Y_{i_\ell}-Y_{j_\ell}) + \bigl(\mu^*_Z(Z_{i_\ell}) - \mu^*_Z(Z_{j_\ell})\bigr),0\bigr\}}{2\sigma^*{}^2}\\ &\approx  \frac{\beta^*}{2\sigma^*{}^2} \max\bigl\{(Y_{i_\ell}-Y_{j_\ell}),0\bigr\},
\end{align*}
where the last step holds since we have assumed $Z_{i_\ell}\approx Z_{j_\ell}$ in our choice of the pair, and so we can expect $\mu^*_Z(Z_{i_\ell}) \approx \mu^*_Z(Z_{j_\ell})$, at least when $\mu_Z^*$ is smooth.  But crucially, our test is invariant to rescaling the weights---that is, choosing weights $w_\ell = \max\bigl\{Y_{i_\ell} - Y_{j_\ell},0\bigr\}$ is equivalent to choosing weights $\frac{\beta^*}{2\sigma^*{}^2} \max\bigl\{(Y_{i_\ell}-Y_{j_\ell}),0\bigr\}$, and thus is nearly equivalent to the oracle weights.

\subsubsection{Two simple matching schemes}
Next, we propose two matching strategies that, again, do not require any knowledge or estimate of the model.

\paragraph{Neighbor matching.}
Our first simple matching strategy is to choose pairs that are nearest neighbors in the list of sorted $Z$ values, as formalized in Algorithm~\ref{alg:neighbour-matching}. This strategy ensures that, for all pairs $(i_\ell,j_\ell) = \bigl(\pi(2m-1),\pi(2m)\bigr)$ that are included in the matching, we have $Z_{i_\ell}\le Z_{j_\ell}$ and $Y_{i_\ell}>Y_{j_\ell}$ by construction, and moreover, it likely holds that $Z_{i_\ell}\approx Z_{j_\ell}$ (since we have chosen two consecutive $Z$ values in the sorted list). 
\begin{algorithm}[ht]
   \caption{Neighbor matching}
   \label{alg:neighbour-matching}
    \begin{algorithmic}
    \STATE \textbf{Preliminaries:} sort $Z$ values, i.e., find a permutation $\pi$  of $[n]$ such that
\[Z_{\pi(1)} \leq Z_{\pi(2)} \leq \dots \leq Z_{\pi(n-1)}\leq Z_{\pi(n)}.\]
\FOR{$m=1,\dots,\lfloor n/2\rfloor$}
\STATE If $Y_{\pi(2m-1)}>Y_{\pi(2m)}$, then add the new pair $\bigl(\pi(2m-1),\pi(2m)\bigr)$ to the matching.
\ENDFOR
\end{algorithmic}
\end{algorithm}
However, an obvious limitation of this na\"ive matching strategy is that many consecutive pairs $\bigl(\pi(2m-1),\pi(2m)\bigr)$ in the sorted list  can fail to have $Y_{\pi(2m-1)}>Y_{\pi(2m)}$ just by chance. In particular, if the $Z$ values in this pair are approximately equal (as we might expect), then the values $Y_{\pi(2m-1)},Y_{\pi(2m)}$ are expected to be approximately independent and identically distributed---which  means that the event $Y_{\pi(2m-1)}>Y_{\pi(2m)}$ will fail approximately half the time. In other words, we are discarding approximately half of the data: we will expect to have $L\approx n/4$ pairs in this matching (meaning that $2L\approx n/2$ data points have been assigned to a matched pair).

\paragraph{Cross-bin matching.}
Our next strategy is cross-bin matching, which aims to avoid the inefficiency of neighbor matching by using (nearly) all of the data. 
To implement this strategy, we will partition the list of sorted $Z$ values into $K$ bins, and will allow a pair of data points to be matched as long as the $Z$ values are in adjacent bins (rather than requiring consecutive $Z$ values, as for neighbor matching). 

\begin{algorithm}[ht]
   \caption{Cross-bin matching}
   \label{alg:cross-bin-matching}
    \begin{algorithmic}
    \STATE \textbf{Preliminaries:} sort $Z$ values, i.e., find a permutation $\pi$ of $[n]$ such that
\[Z_{\pi(1)} \leq Z_{\pi(2)} \leq \dots \leq Z_{\pi(n-1)}\leq Z_{\pi(n)},\]
and define $K$ bins of indices,
\begin{align*}
    A_1 &= \{\pi(1),\dots,\pi(m)\},\\
    A_2 &= \{\pi(m+1),\dots,\pi(2m)\},\\
    \vdots\\
    A_K&=\bigl\{\pi\bigl((K-1)m + 1\bigr) , \dots, \pi(Km)\bigr\},
\end{align*}
where $m=\lfloor n/K\rfloor$.
\FOR{$k=1,\dots,K$}
\STATE Define $r_{k,1},\dots,r_{k,m}$ as a permutation of $A_k$ such that
$Y_{r_{k,1}} \leq \dots \leq Y_{r_{k,m}}$.
\ENDFOR
\FOR{$k=1,\dots,K-1$}
\FOR{$s=1,\dots,\lfloor m/2\rfloor$}
\STATE If $Y_{r_{k,m+1-s}} > Y_{r_{k+1,s}}$, then add the new pair $(r_{k,m+1-s},r_{k+1,s})$ to the matching.
\ENDFOR
\ENDFOR
\end{algorithmic}
\end{algorithm}

To explain this procedure in words:
\begin{itemize}
    \item First we group the $Z$ values into $K$ bins, with each bin $A_k$ containing $m=\lfloor n/K\rfloor$ consecutive $Z$ values.
    \item We then attempt to match the largest values of $Y$ in bin $k$ with the smallest values of $Y$ in bin $k+1$---that is, in the inner ``for loop'' in Algorithm~\ref{alg:cross-bin-matching}, at step $s=1$ we are attempting to match the largest $Y$ value in bin $k$ (i.e., $Y_{r_{k,m}}$) with the smallest $Y$ value in bin $k+1$ (i.e., $Y_{r_{k+1,1}}$), and then at step $s=2$ we proceed to matching the second-largest and second-smallest, and so on.
\end{itemize}
This scheme is illustrated in Figure~\ref{fig:cross-bin}, and is stated formally in Algorithm~\ref{alg:cross-bin-matching}.

\begin{figure}[ht]
    \centering
    \includegraphics[width=.8\textwidth]{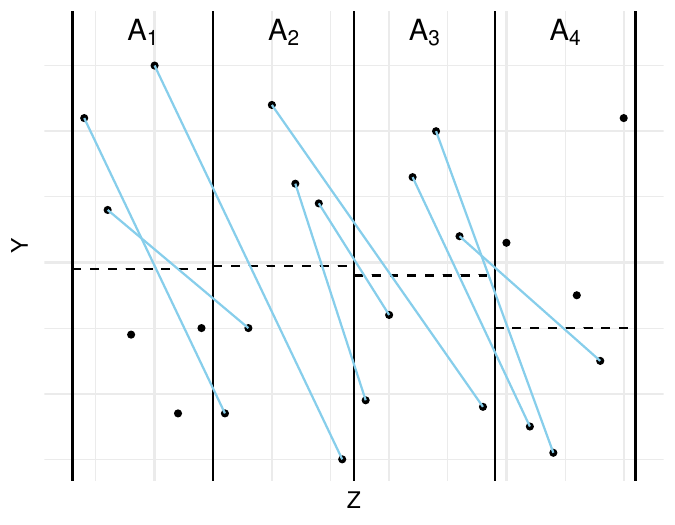}
    \caption{Demonstration of the cross-bin matching scheme described in \cref{alg:cross-bin-matching}.  The dashed lines indicate the median within each bin.}   
    \label{fig:cross-bin}
\end{figure}

Why do we expect that this strategy will be more powerful than the simpler neighbor matching strategy? At a high level, while neighbor matching is expected to discard around half of the data, the cross-bin matching strategy can potentially assign nearly all data points to a matched pair. However, there is a potential tradeoff: while the pairs produced by both strategies will likely satisfy $Z_{i_\ell}\approx Z_{j_\ell}$,  this approximation will be closer to equality for neighbor matching (where pairs consist of consecutive $Z$ values) than cross-bin matching (where pairs consist of $Z$ values in neighboring bins, i.e., they may be up to $2m$ positions apart in the sorted list). If $m$ is not too large (i.e., the number of bins $K$ is not too small), though, then we may hope that this difference is negligible.
We will examine both methods theoretically in the following section to study these tradeoffs in more detail.

\section{Power analysis}\label{sec:power}

In this section, we study the power of the \textnormal{\texttt{PairSwap-ICI}} test under the following general model. The data $(\bX,\bY,\bZ) = (X_i,Y_i,Z_i)_{i\in [n]} \iidsim P_{X,Y,Z}$ are drawn according to
\begin{equation}\label{model:general alternative}
    \bX=\mu_n(\bY,\bZ)+\sigma_n\bzeta,
\end{equation}
with $(Y_i,Z_i)_{i\in [n]}\iidsim P_{Y,Z}$ drawn independently from $(\zeta_i)_{i\in [n]}\iidsim P_\zeta$. We assume that $\mu_n:\mathcal{Y} \times \mathcal{Z} \rightarrow \Xcal$ (applied componentwise) is a measurable function and $\sigma_n> 0$, and also that $P_\zeta$ has mean $0$, variance $1$, and finite fourth moment. Throughout this section,  
we assume that the statistic $T$ admits the form in \eqref{def:sum-of-psi} with $\psi(x,x') = x-x'$. 
 
\subsection{Asymptotic power guarantee for the general \texttt{PairSwap-ICI} test}\label{sec:power_general}

We start by studying the power of \texttt{PairSwap-ICI} test in the setting where $\mathcal{Y}$ and $\mathcal{Z}$ are general measurable spaces. We use $\|\cdot\|_2$ to denote the Euclidean norm, and write $\mathrm{d}_{\rm TV}$ for the total variation (TV) distance, so that for probability measures $P$ and $Q$ defined on a common measurable space $(W,\mathcal{W})$, $\mathrm{d}_{\rm TV}(P,Q) :=\sup_{A \in \mathcal{W}}|P(A)-Q(A)|$.  Given any vector~$\bf{v}$, we write $\bf{v}^+$ to denote the vector with $i$th component $v_{i}^{+} =\max\{v_i,0\}$. We write $\bf{a}\circ \bf{b}$ for the Hadamard product of vectors $\bf{a}$, $\bf{b}$ of the same dimension, with $i$th component $a_i\cdot b_i$.
Given a matching $M= \bigl((i_\ell,j_\ell)\bigr)_{\ell \in [L]} \in \mathcal{M}_n(\bZ)$ and a vector $\bv = (v_i)_{i \in [n]} \in \mathbb{R}^n$, we define the differencing operator $\Delta:\mathbb{R}^n \rightarrow \mathbb{R}^L$ by
\[
\Delta \bv \equiv (\Delta_\ell \bv)_{\ell \in [L]} := (v_{i_\ell} - v_{j_\ell})_{\ell \in [L]}.
\]
In particular, it will be convenient to define the \emph{mean difference vector}
\[
\Delta \mu_n(\bY,\bZ) \equiv \Delta \mu_n(\bY,\bZ;M) = \bigl(\mu_n(Y_{i_\ell},Z_{i_\ell})-\mu_n(Y_{j_\ell},Z_{j_\ell})\bigr)_{\ell \in [L]} \in \R^L.
\]
We will make the following assumptions on the weight vector $\bw \equiv \bw_n \in [0,\infty)^L$ and mean difference vector:
\begin{assumption}\label{A1}
    We have $\|\bw\|_2 > 0$
     and $\|\bw\|_\infty = {\rm o}_P(\|\bw\|_2)$.
    \end{assumption}
    \begin{assumption}\label{A2}
    We have $\|\bw\circ\Delta\mu_n(\bY,\bZ)\|_2 = {\rm o}_P\left( \big|\bw^\top \Delta \mu_n(\bY,\bZ)\big|\vee\sigma_n\|\bw\|_2\right)$ and $\|\bw\circ\Delta\mu_n(\bY,\bZ)\|_3 = \mathrm{o}_P\left(\|\bw\circ\Delta\mu_n(\bY,\bZ)\|_2\vee\sigma_n\|\bw\|_2\right)$.
    \end{assumption}
    To explain the origin of these assumptions, observe that since we are under the model class \eqref{model:general alternative}, and we use the shared linear kernel $\psi(x,x')=x-x'$, the \textnormal{\texttt{PairSwap-ICI}} test revolves around the statistic 
    \[
\bw\circ\Delta\bX=\bw\circ\Delta\mu_n(\bY,\bZ)+\sigma_n\bigl(\bw\circ\Delta\bzeta\bigr);
    \]
    see~\eqref{eqn:test_as_fucntion_ofWcircDeltaX}.  Here, the first term can be interpreted as the signal of the test, while the second term represents the additional noise. Under this interpretation, Assumption~\textbf{A1} ensures that the noise term is well behaved by requiring that the weight vector $\bw$ is not too spiky, while Assumption~\textbf{A2} requires the signal term not to be too spiky.

The following theorem gives our main general result on the asymptotic conditional power of the \texttt{PairSwap-ICI} test.

\begin{theorem}\label{thm:general-power-asymptotic}
    Suppose that the triple $(\bX,\bY,\bZ)$ satisfies \eqref{model:general alternative}.  Suppose further that the weight vector $\bw \in [0,\infty)^L$ and $M \in \mathcal{M}_n(\bZ)$ (which may both depend on $\bY,\bZ$ but not $\bX$) satisfy Assumptions~\ref{A1} and~\ref{A2}.  Then for $\alpha \in (0,1)$, we have that
    \[
    \PPst{p\leq \alpha}{\bY,\bZ}=\Phi\biggl(\Phi^{-1}(\alpha)+\frac{\bw^\top \Delta \mu_n(\bY,\bZ)}{\sqrt{2}\sigma_n\|\bw\|_2}\biggr)+\mathrm{o}_P(1).
    \]
\end{theorem}
Theorem~\ref{thm:general-power-asymptotic} follows from Theorem~\ref{thm:general power analysis} in Appendix~\ref{app:general_power}, which provides finite-sample upper and lower bounds on the conditional power of the \texttt{PairSwap-ICI} test for any feasible matching and weight vector.  One key insight of Theorem~\ref{thm:general-power-asymptotic} is that it completely characterizes the asymptotic power of our test through the \emph{\texttt{PairSwap-ICI} signal-to-noise ratio}
\begin{equation}\label{eqn:signal_general_wandM}
    S(\bw,M):=\frac{\bw^\top \Delta \mu_n(\bY,\bZ)}{\sqrt{2}\sigma_n\|\bw\|_2}.
\end{equation}
It is interesting to compare the conclusion of Theorem~\ref{thm:general-power-asymptotic} with the informal approximation to the conditional power given by~\eqref{eqn:estimate_power}.  Referring to the terms in that expression, in the model~\eqref{model:general alternative}, we have
\[
\frac{\sum_{\ell=1}^L w_\ell \EEst{\psi(X_{i_\ell},X_{j_\ell})}{\bY,\bZ}}{\sqrt{\sum_{\ell=1}^L w_\ell^2 \VVst{\psi(X_{i_\ell},X_{j_\ell})}{\bY,\bZ}}}
= \frac{\sum_{\ell=1}^L w_\ell \EEst{X_{i_\ell} - X_{j_\ell}}{\bY,\bZ}}{\sqrt{\sum_{\ell=1}^L w_\ell^2 \VVst{X_{i_\ell} - X_{j_\ell}}{\bY,\bZ}}} =  S(\bw,M).
\]
Thus, we can regard Theorem~\ref{thm:general-power-asymptotic} as formalizing the heuristic arguments in Section~\ref{sec:oracle}.  By analogy with Lemma~\ref{Lemma:OracleWeights}, we see that the oracle weight vector maximizes the correlation between $\bw$ and $\Delta \mu_n(\bY,\bZ)$, and is achieved by setting
\begin{equation}\label{eqn:oracle_weight_general_model}
    \bw^* \propto \Delta \mu_n^+(\bY,\bZ).
\end{equation}

\subsection{Hardness of testing the null \texorpdfstring{$H_0^{\textnormal{ICI}}$}{the isotonic null}}\label{sec:characterizing hardness}
 
We now turn to the question of the fundamental hardness of testing $H_0^{\mathrm{ICI}}$ in the model~\eqref{model:general alternative}.  
 
Let $\mathcal{C}_\ISO$ denote the closed convex cone of nondecreasing functions on $\Zcal$, i.e.~the set of $f:\Zcal\to \R$ such that $f(z_1)\le f(z_2)$ whenever $z_1\preceq z_2$.    Define 
\begin{equation}\label{eq: isotonic_signal_strength}
\ISNR_n=\frac{\inf_{g\in\mathcal{C}_\ISO}\Ep{P_{Y,Z}^n}{\|\mu_n(\bY,\bZ)-g(\bZ)\|_2}}{\sigma_n}, \quad \widehat\ISNR_n=\frac{\inf_{g\in\mathcal{C}_\ISO}\|\mu_n(\bY,\bZ)-g(\bZ)\|_2}{\sigma_n},
\end{equation}
which we call the oracle and empirical \emph{isotonic signal-to-noise ratio} respectively (here $g(\bZ)$ is again applied componentwise). These quantities capture the extent to which the true model for $X\mid Y,Z$ contains signal that cannot be explained by the isotonic null, $H_0^{\textnormal{ICI}}$. In particular, under $H_0^{\textnormal{ICI}}$, we have $\ISNR_n=\widehat\ISNR_n=0$ (since we can simply take $g$ to be the true regression function, which does not depend on $Y$), while under the alternative these quantities might be large.  We naturally expect $\widehat{\ISNR}_n \approx \ISNR_n$ when~$n$ is large.

We can relate the \texttt{PairSwap-ICI} signal-to-noise ratio~\eqref{eqn:signal_general_wandM} with oracle weight vector $\bw^*$ from~\eqref{eqn:oracle_weight_general_model} to the empirical isotonic signal-to-noise ratio via
\begin{equation}\label{eqn:upper_bd_signal_by_ISNR}
    S(\bw^*,M) =\frac{\|\Delta \mu_n^+(\bY,\bZ)\|_2}{\sqrt{2}\sigma_n} \leq \widehat{\ISNR}_n;
\end{equation}
see Theorem~\ref{thm:key-property-oracle} for a proof.

The goal of this subsection is to show that the population counterpart of $\widehat{\ISNR}_n$, i.e.~the oracle isotonic signal-to-noise ratio $\ISNR_n$, provides a fundamental limit to the power of \emph{any} test of~$H_0^{\mathrm{ICI}}$.  To this end, we begin with a standard total variation calculation. 
\begin{proposition}\label{thm: hardness result general}
    Fix $\alpha\in (0,1)$. Fix any test\footnote{Traditionally, we think of a hypothesis test $\phi$ as a map from data to a decision, i.e., $\phi(\bX,\bY,\bZ)\in\{0,1\}$. Why, then, do we define tests $\phi$ as mapping to the space $[0,1]$? This is because, in some settings, we may want to consider randomized tests---for instance, in the notation above where $\phi$ maps to $[0,1]$, an outcome $\phi(\bX,\bY,\bZ)=0.75$ represents that, given the data, our randomized test rejects the null with probability 0.75. Of course, a nonrandomized test is simply a special case, obtained by restricting the output of $\phi$ to lie in $\{0,1\}$.} $\phi: (\Xcal\times\Ycal\times\Zcal)^n \to [0,1]$ that controls the Type~I error at level $\alpha$, i.e.,
\[
\sup_{P\in H_0^{\textnormal{ICI}}}\Ep{P}{\phi(\bX,\bY,\bZ)} \leq \alpha.
\]
    Then for any distribution $P_{X,Y,Z}$,
    \[
    \Ep{P_{X,Y,Z}}{\phi(\bX,\bY,\bZ)}\leq \alpha+ \inf_{Q_{X,Y,Z}\in H_0^\textnormal{ICI}}\TV\bigl(P^n_{X,Y,Z},Q^n_{X,Y,Z}\bigr).
    \]
\end{proposition}
 This last total variation term is the distance of $P_{X,Y,Z}$ from the null class $H_0^\textnormal{ICI}$, meaning the closer $P$ is to null models, the harder it will be to have non-trivial power against $P$. While in general it is hard to derive exact expressions for the total variation term, under some additional model assumptions on $P$, we can derive interpretable upper bounds for it.  Specifically, consider the model~\eqref{model:general alternative} where the noise distribution $P_\zeta$ has a density function~$f_\zeta$ satisfying
 \begin{equation}\label{eqn:assume_hellinger}\left\{ \int_\R \left(\sqrt{f_\zeta(x)} - \sqrt{f_\zeta(x + t)}\right)^2\;\mathsf{d}x \right\}^{1/2} \leq L_\zeta  |t|\textnormal{ for all $t\in\R$.} \end{equation}
 Note that the left-hand side is equal to
 \[{\rm H}(P_\zeta, t+P_\zeta),\]
 where ${\rm H}$ denotes the Hellinger distance between distributions, and $t+P_\zeta$ denotes the distribution of $t+\zeta$ when $\zeta\sim P_\zeta$ (i.e., a translation of the original distribution $P_\zeta$). In other words, we are assuming that the map $t\mapsto t+P_\zeta$ is $L_\zeta$-Lipschitz with respect to the Hellinger distance. For example, if $P_\zeta = \normal(0,1)$ is the standard Gaussian, then ${\rm H}^2(P_\zeta, t+P_\zeta) = 2\bigl(1-e^{-t^2/8}\bigr)\leq t^2/4$ \citep{pardo2018statistical}, and so the assumption~\eqref{eqn:assume_hellinger} holds with $L_\zeta = 1/2$.

We now state a result establishing a bound on power for any test, as a function of the isotonic SNR.
     \begin{theorem}\label{lem:hardness_result_hellinger}
        Suppose $(\bX,\bY,\bZ)$ satisfy~\eqref{model:general alternative} for some noise distribution $P_\zeta$ satisfying the condition~\eqref{eqn:assume_hellinger}. Then 
        \[\inf_{Q_{X,Y,Z}\in H_0^\textnormal{ICI}}\TV\left(P^n_{X,Y,Z},Q^n_{X,Y,Z}\right) \leq L_\zeta\cdot \ISNR_n,\]
        and consequently, for any test $\phi$
     that controls Type~I error at level $\alpha$,
         \[
         \Ep{P_{X,Y,Z}}{\phi(\bX,\bY,\bZ)}\leq \alpha + L_\zeta\cdot \ISNR_n.
         \]
    \end{theorem}
To summarize, we have shown that under a mild assumption on the noise distribution $P_\zeta$, the oracle isotonic SNR gives a universal upper bound on the power of \emph{any} test of the isotonic null, $H_0^{\textnormal{ICI}}$---no valid test can achieve non-trivial power when $\ISNR_n$ is negligible.

\subsection{Asymptotic power of oracle and plug-in procedures}\label{sec:power_oracle}

The theory of Sections~\ref{sec:power_general} and~\ref{sec:characterizing hardness}, and in particular the general inequality~\eqref{eqn:upper_bd_signal_by_ISNR}, raise the natural question of the extent to which the \texttt{PairSwap-ICI} test, with appropriate choices of weight vector and matching, is able to approximate the upper bound on the power of any valid test, as articulated by Theorem~\ref{lem:hardness_result_hellinger}.  Ignoring the difference between the empirical and oracle isotonic SNRs, this amounts to asking about the tightness of the inequality~\eqref{eqn:upper_bd_signal_by_ISNR}.  We address this by specializing Theorem~\ref{thm:general-power-asymptotic} initially to the setting of oracle weights and matching, as well as $\mathcal{Z} = \mathbb{R}$.  Since these oracle weights and matching depend on properties of the underlying data generating distribution that would typically be unknown to the practitioner, we then establish that the same asymptotic power properties hold for the plug-in version of the procedure outlined in Section~\ref{sec:oracle}.

\subsubsection{\texorpdfstring{$\widehat \ISNR_n$}{The emprical isotonic SNR} drives the conditional power of oracle matching}\label{sec:power_general_alternatives}

Theorem~\ref{thm:oracle-matching-asymptote} below provides asymptotic upper and lower bounds on the conditional power $\PPst{p\leq \alpha}{ \bY,\bZ}$ for oracle weights and matching.
\begin{theorem}\label{thm:oracle-matching-asymptote}
    Suppose that $\widehat{\ISNR}_n > 0$ and that $\bigl(\|\mu_n(\bY,\bZ)\|_\infty/\sigma_n\bigr)\vee\bigl(\|\mu_n(\bY,\bZ)\|_\infty/\sigma_n\bigr)^4= \mathrm{o}_P(\widehat\ISNR_n)$. Then, under the model~\eqref{model:general alternative}, the conditional power of the \textnormal{\texttt{PairSwap-ICI}} procedure with oracle weights and matching satisfies for $\alpha \in (0,1)$ that
    \begin{equation}\label{eq:asymptote-power-oracle-bounds}
        \Phi\biggl(\Phi^{-1}(\alpha)+\frac{\widehat\ISNR_n}{\sqrt{2}}\biggr)-\mathrm{o}_P(1)\leq \PPst{p\leq\alpha}{\bY,\bZ}\leq 
        \Phi\bigl(\Phi^{-1}(\alpha)+\widehat\ISNR_n\bigr)+\mathrm{o}_P(1).
    \end{equation}
\end{theorem}
The condition  $\bigl(\|\mu_n(\bY,\bZ)\|_\infty/\sigma_n\bigr)\vee\bigl(\|\mu_n(\bY,\bZ)\|_\infty/\sigma_n\bigr)^4= \mathrm{o}_P(\widehat\ISNR_n)$ is a mild assumption on the quality of the matching. For example, if $\mu_n$ is bounded and $\sigma_n=\sigma$ for every $n\in \mathbb{N}$, then $\bigl(\|\mu_n(\bY,\bZ)\|_\infty/\sigma_n\bigr)\vee\bigl(\|\mu_n(\bY,\bZ)\|_\infty/\sigma_n\bigr)^4 = \mathrm{O}(1)$, so for this assumption to hold, it is sufficient to have $\widehat\ISNR_n\overset{P}{\to}\infty$; since this latter quantity represents the Euclidean distance between the $n$-dimensional vector $\mu_n(\bY,\bZ)$ and the convex cone $\mathcal{C}_{\mathrm{ISO}}$, this amounts to a weak requirement.

From Theorem~\ref{thm:oracle-matching-asymptote}, we see that the empirical isotonic SNR drives both the upper and lower bound on the conditional power.   Furthermore, the upper and lower bounds on conditional power match, up to the factor of $1/\sqrt{2}$ appearing in the expression for the lower bound but not the upper bound. In Theorem~\ref{thm:cross_bin_matching} in Section~\ref{sec:partial_linear_model} below, we will see that in a partially linear model setting and under the hypothesis of the symmetry of the distribution $P_{Y|Z}$, it is possible to close this~gap.

\subsubsection{Power guarantees for plug-in matching with an estimate of $\mu_n$}\label{sec:power_with_estimated_mu}
 Suppose now that no oracle information is available and that we run the \texttt{PairSwap-ICI} test using the plug-in matching from Section~\ref{sec:power_with_estimated_mu}.  By~\eqref{eqn:oracle_weight_general_model}, under the model~\eqref{model:general alternative} with the linear kernel $\psi(x,x')=x-x'$, and since the $p$-value does not depend on the scale of the weights, we may assume without loss of generality that the optimal weight vector is 
\[
\bw^*=\Delta\mu_n^+(\bY,\bZ).
\]
In particular, the oracle weights (and hence matching) depend only on $\mu_n$ and not on $\sigma_n$. Thus, we may construct an estimate $\hat\mu_n$ of $\mu_n$ on a random split of the data, and then apply \texttt{PairSwap-ICI} on the remaining observations with the plug-in matching $\hat M$ from~\eqref{def:plug-in-matching}. For simplicity of presentation, however, in Theorem~\ref{thm:oracle-matching-estimated-mu} below, we assume that $\hat{\mu}_n$ is constructed based on an independent dataset. 
\begin{theorem}\label{thm:oracle-matching-estimated-mu}
    Consider the setting and assumptions of Theorem \ref{thm:oracle-matching-asymptote}, except that we use plug-in matching $\hat{M}$ with an estimate $\hat\mu_n$ constructed based on independent data.  Suppose further that
    \[
    \frac{\|\hat\mu_n(\bY,\bZ)-\mu_n(\bY,\bZ)\|_2}{\sigma_n} =\mathrm{o}_P(\widehat{\mathrm{ISNR}}_n), ~~\frac{\|\hat\mu_n(\bY,\bZ)-\mu_n(\bY,\bZ)\|_\infty^4}{\sigma_n^4} =\mathrm{o}_P(\widehat{\mathrm{ISNR}}_n),
    \]
    with the probability statements taken over the joint distribution of $(\bY,\bZ)$ and the data used to construct $\hat{\mu}_n$.  Then the conditional power of the \textnormal{\texttt{PairSwap-ICI}} satisfies \eqref{eq:asymptote-power-oracle-bounds} for $\alpha \in (0,1)$.
\end{theorem}
Thus, under suitable consistency assumptions on $\hat\mu_n$, we can recover the power guarantees in Theorem~\ref{thm:oracle-matching-asymptote}. Since the requirements on $\hat{\mu}_n$ in Theorem~\ref{thm:oracle-matching-estimated-mu} are quite weak, it may be advantageous to employ a data split of unequal sizes when using plug-in matching---in particular, we might choose to use only a small fraction of the available data points to estimate $\mu_n$ and compute the weights and matching, and save the rest for the \textnormal{\texttt{PairSwap-ICI}} test once these choices are fixed.

\subsection{Near-optimal power guarantees without knowledge of \texorpdfstring{$\mu_n$}{mu}}\label{sec:partial_linear_model}

hile the results of Section~\ref{sec:power_oracle} establish the optimality properties of the oracle and plug-in versions of the \texttt{PairSwap-ICI} procedure, we remark that when $n$ is large, it may be time-consuming to compute the optimal matching in~\eqref{def:oracle-matching}.   The aim of this subsection, then is to show that the very simple matching schemes of Section~\ref{sec:heuristic matching} enjoy similar optimality properties, at least under some further assumptions on the model~\eqref{model:general alternative}.  Specifically, assume here that $\Ycal = \Zcal=\R$ and consider the class of partially linear Gaussian models given by
\begin{equation}\label{model:partial_linear_model}
(\bX,\bY,\bZ)~\textnormal{satisfy}~\eqref{model:general alternative} \  \textnormal{with}~
 \mu_n(\bY,\bZ)=\mu_0(\bZ)+\beta_n\,\bY
\end{equation}
for some $\mu_0\in \mathcal{C}_\ISO$ (applied componentwise) and some $\beta_n\in (0,\infty)$; we assume also that $Y$ is a mean-zero variable taking values in $[-1,1]$, with $\EE{\VVst{Y}{Z}}>0$ (i.e., $Y$ is not equal to a deterministic function of $Z$).  
Note that $\mu_0$ is a fixed function, so the dependence on $n$ of the distribution of $(X,Y,Z)$ is solely through $\beta_n$ and $\sigma_n$.  The following lemma relates $\ISNR_n$ to a more interpretable quantity, namely the expected conditional variance of $Y$ given~$Z$.
\begin{lemma}\label{lemma:beta-phase-transition}
Under the model class \eqref{model:partial_linear_model}, provided that $\mathbb{E}|\mu_0(Z)| < \infty$,
    \[
    \frac{\sqrt{n}\beta_n}{\sigma_n} \cdot \bigl(\EE{\VVst{Y}{Z}}\bigr)^{1/2}\,\bigl(1 + \mathrm{o}_P(1)\bigr)\leq \ISNR_n \leq \frac{\sqrt{n}\beta_n}{\sigma_n} \bigl(\mathrm{Var}(Y)\bigr)^{1/2} \leq \frac{\sqrt{n}\beta_n}{\sigma_n} .
    \]
\end{lemma}
We can draw several conclusions from this result.
By Lemma~\ref{lemma:beta-phase-transition} and Theorem~\ref{lem:hardness_result_hellinger}, when $P_\zeta=\mathcal{N}(0,1)$, no valid test 
can achieve asymptotically non-trivial power against an alternative in the model class \eqref{model:partial_linear_model} when $\frac{\sqrt{n}\beta_n}{\sigma_n} = \mathrm{o}(1)$. On the other hand, if $\frac{\sqrt{n}\beta_n}{\sigma_n}\to \infty$, then $\ISNR_n\to\infty$, and the \textnormal{\texttt{PairSwap-ICI}} test should have power converging to $1$. %{\color{blue} Our results, stated below, closely mirror this behavior. More importantly, we provide a precise characterization of power in terms of $\beta_n/\sigma_n$ in the intermediate regime $\frac{\beta_n}{\sigma_n}=\theta(n^{1/2})$, where the power is nontrivial but strictly less than $1$.}

Under mild conditions on the model, we expect that $\widehat{\ISNR}_n = \ISNR_n\bigl(1+\mathrm{o}_P(1)\bigr)$. With this approximation, another consequence of Lemma~\ref{lemma:beta-phase-transition} and \eqref{eq:asymptote-power-oracle-bounds} is that the conditional power of oracle matching also satisfies \begin{equation}\label{eqn:limits_of_general_power_guarantee}
      \PPst{p\leq\alpha}{\bY,\bZ} \geq \Phi\Biggl(\Phi^{-1}(\alpha)+\frac{\sqrt{n}\beta_n\cdot  \bigl\{\mathbb{E}\bigl[\mathrm{Var}_{Y\sim P_{Y\mid Z}}(Y)\bigr]\bigr\}^{1/2}}{\sqrt{2}\sigma_n}\Biggr)-\mathrm{o}_P(1).
\end{equation}

We can now compare the asymptotic conditional power of the \textnormal{\texttt{PairSwap-ICI}} test with neighbor matching and cross-bin matching against this lower bound.
\begin{theorem}\label{thm:neighbour_matching}
    Under the model class \eqref{model:partial_linear_model},  fix $\alpha\in(0,1)$ and assume that $\sigma_n\gg \sqrt{\log n/n}$. Assume also that $\mu_0(Z)$ is a sub-Gaussian random variable. Then the conditional power of neighbor matching (\cref{alg:neighbour-matching}), implemented with kernel $\psi(x,x')=x-x'$ and weights $w_\ell = \max\bigl\{Y_{i_\ell}-Y_{j_\ell},0\bigr\}$, satisfies 
    \[
   \PPst{p\leq\alpha}{\bY,\bZ}=\Phi\biggl(\Phi^{-1}(\alpha)+\frac{\sqrt{n}\beta_n\cdot  \bigl\{\mathbb{E}\bigl[\mathrm{Var}_{Y\sim P_{Y\mid Z}}(Y)\bigr]\bigr\}^{1/2}}{2\sigma_n}\biggr)+\mathrm{o}_P(1). 
    \]
\end{theorem}

Notably, the conditional power of neighbor matching matches the lower bound from \eqref{eqn:limits_of_general_power_guarantee} up to a factor of $1/\sqrt{2}$ in the signal-to-noise ratio term. This factor arises due to the inherent inefficiency of neighbor matching: as discussed in Section~\ref{sec:heuristic matching}, neighbor matching discards roughly half the data.

We now turn to cross-bin matching. Since we would like paired $Z$ values $Z_{i_\ell}\leq Z_{j_\ell}$ to be approximately equal (in order to avoid an overly conservative test), we further allow the number of bins $K$ to depend on sample size $n$, and write this as $K_n$ from now on.  Consequently, the power of cross-bin matching can meet the lower bound in~\eqref{eqn:limits_of_general_power_guarantee}.

\begin{theorem}\label{thm:cross_bin_matching}
    Under the same setting and assumptions as Theorem \ref{thm:neighbour_matching}, the conditional power of cross-bin matching (Algorithm~\ref{alg:cross-bin-matching}) with $K_n$ bins (where $\sigma_n^{-1}\sqrt{n \log n}\ll K_n \ll n$), implemented with kernel $\psi(x,x')=x-x'$ and weights $w_\ell=\max\bigl\{Y_{i_\ell}-Y_{j_\ell},0\bigr\}$, satisfies 
    \begin{equation}
    \label{Eq:CPLowerBound}
   \PPst{p\leq\alpha}{\bY,\bZ}\geq\Phi\biggl(\Phi^{-1}(\alpha)+\frac{\sqrt{n}\beta_n\cdot  \bigl\{\mathbb{E}\bigl[\mathrm{Var}_{Y\sim P_{Y\mid Z}}(Y)\bigr]\bigr\}^{1/2}}{\sqrt{2}\sigma_n}\biggr)-\mathrm{o}_P(1).
    \end{equation}
    Further, if the conditional distribution 
    $P_{Y\mid Z}$ is symmetric about its mean almost surely, then the conditional power of cross-bin matching satisfies 
    \[
  \PPst{p\leq\alpha}{\bY,\bZ}=\Phi\biggl(\Phi^{-1}(\alpha)+\frac{\sqrt{n}\beta_n\cdot  \bigl\{\mathbb{E}\bigl[\mathrm{Var}_{Y\sim P_{Y\mid Z}}(Y)\bigr]\bigr\}^{1/2}}{\sigma_n}\biggr)+\mathrm{o}_P(1).
    \]
\end{theorem}
In particular, the first part of the result shows that for cross-bin matching, the conditional power asymptotically matches the lower bound in~\eqref{eqn:limits_of_general_power_guarantee}, without an additional factor of $1/\sqrt{2}$ as for neighbor matching under the same assumptions. Moreover, with an additional assumption of symmetry of the conditional distribution $P_{Y\mid Z}$, we can also remove another factor of $1/\sqrt{2}$.  Finally, in Theorem~\ref{thm:cross_bin_matching_general} in Appendix~\ref{app:power-cross-bin}, we characterize the asymptotic conditional power of cross-bin matching (in contrast to the lower bound in Theorem~\ref{thm:cross_bin_matching}) even without the symmetry assumption.
%, the above result provides only an asymptotic lower bound on the conditional power of cross-bin matching, we can characterize the power more exactly; the details are deferred to Theorem~\ref{thm:cross_bin_matching_general} in Appendix~\ref{app:power-cross-bin}.

\section{Extension to multivariate response spaces}\label{sec:extension}

So far we have worked under Assumption~\ref{asm:st} and in the setting where the response space $\Xcal=\R$. In this section, we extend the \texttt{PairSwap-ICI} framework to multivariate responses, allowing $\Xcal = \R^{d_X}$ with $d_X \ge 1$. This extension broadens the applicability of our method, enabling it to handle settings in which the response variable is vector-valued, or other settings where there is no natural response variable and $X$ is vector-valued.

We begin by generalizing the stochastic monotonicity assumption to the case $\Xcal = \R^{d_X}$.  For $x,y \in \Xcal$, we write $x \le y$ to denote the coordinate-wise ordering $x_j \le y_j$ for all $j \in [d_X]$, and we say that a function $f:\Xcal \to \R$ is nondecreasing if $f(x) \le f(y)$ whenever $x \le y$.  
\begin{customasm}{2}[Monotonicity of the conditional distribution $P_{X\mid Z}$]\label{asm:st_gen}
    Let $\preceq$ be a partial order on $\Zcal$. We assume that $X$ is stochastically nondecreasing in $Z$, meaning that for all nondecreasing functions $f:\Xcal \to \R$,
    \[
    \textnormal{if $z \preceq z'$ then } \Epst{P}{f(X)}{Z=z} \le \Epst{P}{f(X)}{Z=z'}.
    \]
\end{customasm}

In particular, Assumption~\ref{asm:st_gen} specializes to Assumption~\ref{asm:st} when $d_X = 1$.\footnote{For any $x \in \Xcal$, the function $y \mapsto \One{y \ge x}$ is nondecreasing, so Assumption~\ref{asm:st_gen} imposes a monotonicity condition on the upper orthant probabilities of $X \mid Z$. When $d_X=1$, this reduces to monotonicity of the conditional survival function, as in Assumption~\ref{asm:st}.}

As before, to circumvent the hardness of conditional independence testing, we consider the restricted null hypothesis
\begin{align}\label{eq:null_gen}
    H_{0,d_X}^{\textnormal{ICI}} : X \independent Y \mid Z,
    \textnormal{ and $P_{X\mid Z}$ satisfies Assumption~\ref{asm:st_gen},}
\end{align}
so that $H_{0,1}^{\textnormal{ICI}}$ coincides with $H_0^{\textnormal{ICI}}$. 

The only modification of the \texttt{PairSwap-ICI} methodology required to accommodate multivariate $X$ concerns the definition of anti-monotonic functions.  Specifically, a function $\psi:\Xcal \times \Xcal \to \R$ is now said to satisfy anti-monotonicity if 
\begin{align}
\label{defn:monotonicity-of-psi-general}
    \psi(x +\Delta, x'-\Delta')-\psi(x'-\Delta',x+\Delta) \geq \psi(x,x')-\psi(x',x)
\end{align}
for all $x,x'$ and all $\Delta,\Delta' \geq 0$, where the inequality is interpreted elementwise.  Otherwise, the testing procedure proceeds exactly as in Section~\ref{sec:method}: given a matching $M = \bigl((i_1,j_1),\ldots,(i_L,j_L)\bigr) \in \mathcal{M}_n(\bZ)$ and anti-monotonic functions $\psi_1,\ldots,\psi_L:\mathcal{X} \times \mathcal{X} \rightarrow \mathbb{R}$, we define the test statistic $T$ as in~\eqref{eq:test-stat}, construct the $p$-value $p$ as in~\eqref{eq:def-test}, and reject $H_{0,d_X}^{\textnormal{ICI}}$ at level $\alpha \in (0,1)$ whenever $p \le \alpha$.

Theorem~\ref{thm:main-general} below confirms that the multivariate response \texttt{PairSwap-ICI} procedure retains finite-sample Type~I error control.

\begin{theorem}
    \label{thm:main-general}
    Under $H_{0,d_X}^{\textnormal{ICI}}$, the conditional Type~I error of the \textnormal{\texttt{PairSwap-ICI}} test satisfies $\PPst{p\le \alpha }{ \bY, \bZ}\le \alpha$ for all $\alpha \in [0,1]$.  
\end{theorem}

The proof of Theorem~\ref{thm:main-general} is analogous to that of Theorem~\ref{thm:main}. For completeness, we give the proof in Appendix~\ref{app:proof_validity_extension}.  We emphasise again that $\mathcal{Y}$ and $\mathcal{Z}$ here may be arbitrary measurable spaces, provided that  $\mathcal{Z}$ is equipped with a partial order.

Following the methodological development of Section~\ref{sec:design}, we consider anti-monotonic functions of the form $\psi_\ell=w_\ell\cdot\psi$ with $w_\ell\ge 0$ for all $\ell\in [L]$ and an anti-monotonic kernel $\psi$.  The oracle and plug-in matching schemes introduced in~\eqref{def:oracle-matching} and~\eqref{def:plug-in-matching} can be used directly in this setting. The main practical challenge therefore lies in designing a kernel $\psi$ that effectively captures the information contained in $\bX$.

A natural approach is to project $X$ along directions in the positive orthant. Specifically, writing $S^{+}_{d_X}:=\{u\in \R^{d_X}: u_j\ge 0~\text{for all}~j\in [d_X],\, \|u\|_2=1\}$, we may take any $u\in S^+_{d_X}$ and then consider the projection $\Pi_{d_X}:\mathcal{X} \rightarrow \mathbb{R}$ given by $\Pi_{d_X}(x) := u^\top x$. In particular, the direction $u$ can also be chosen based on $\bY,\bZ$ to gain higher power. Any such projection reduces the multivariate response to a scalar quantity while preserving the required monotonicity structure. Consequently, we can re-use the anti-monotonic $\psi$ functions introduced in Section~\ref{sec:design}; 
for the remainder of this section, we focus on the linear kernel $\psi(x,x')=\Pi_{d_X}(x)-\Pi_{d_X}(x')$.

To study the asymptotic conditional power of the multivariate response \texttt{PairSwap-ICI} procedure, we assume that the data $(\bX,\bY,\bZ) = (X_i,Y_i,Z_i)_{i\in [n]} \iidsim P_{X,Y,Z}$ are drawn according to the model~\eqref{model:general alternative}, where $\mu_n:\mathcal{Y} \times \mathcal{Z} \rightarrow \Xcal$ (applied componentwise) is a measurable function and $P_\zeta$ is a distribution on $\mathcal{X}$ having mean $0$, ones on the diagonal of its covariance matrix and $\int_{\mathcal{X}} \|x\|_2^4 \, dP_\zeta < \infty$.  Here $\mu_n(\bY,\bZ)\in\R^{n\times d_X}$, and given the projection direction $u \in S_{d_X}^+$, we have
\[
\Pi_{d_X}(\bX)=\bX u
=\mu_n(\bY,\bZ)u
+\sigma_n \,\bzeta'
\]
for some $\bzeta'$ with independent components having mean $0$, variance $1$ and finite fourth moment. Therefore, conditional on $u$, an application of Theorem~\ref{thm:general-power-asymptotic} immediately yields an analogous asymptotic characterization of the conditional power of the \texttt{PairSwap-ICI} test:

\begin{corollary}\label{thm:general-power-asymptotic_multivariate}
    In the setting of Theorem~\ref{thm:general-power-asymptotic}, suppose the weights, the matching and the projection direction $u$ satisfy Assumptions~\ref{A1} and~\ref{A2} after replacing $\Delta\mu_n(\bY,\bZ)$ with $\Delta\mu_n(\bY,\bZ)u$.  Then for $\alpha \in (0,1)$, we have
    \[
    \PPst{p\leq \alpha}{\bY,\bZ}=\Phi\biggl(\Phi^{-1}(\alpha)+\frac{\bw^\top \Delta\bigl(\mu_n(\bY,\bZ)u\bigr)}{\sqrt{2}\sigma_n\|\bw\|_2}\biggr)+\mathrm{o}_P(1).
    \]
\end{corollary}

Similar to~\eqref{eqn:oracle_weight_general_model}, for a given matching $M\in \mathcal{M}_n(\bZ)$ and a projection direction $u\in S^{+}_{d_X}$, the oracle weights are 
$\bw^\star \propto \Delta \bigl(\mu_n^+(\bY,\bZ) u\bigr)$.  With that choice, the \texttt{PairSwap-ICI} signal-to-noise ratio satisfies
\[
S(\bw^*,M;u) = \bigl\|\Delta\bigl(\mu_n^+(\bY,\bZ)u\bigr)\bigr\|_2,
\]
and the optimal choice of $u$ is 
\[
u^* \in \argmax_{u \in S_{d_X}^+} S(\bw^*,M;u).
\]
In practice, the oracle weights may again be estimated using a plug-in approach, so that a numerical approximation to $u^*$ (using either a deterministic or random grid) may be found.

\section{Simulations}\label{sec:simulations}

In this section, we evaluate the performance of our method on simulated data, and compare the matching strategies from Section~\ref{sec:heuristic matching}. 
Code for reproducing all experiments can be found at \url{https://github.com/jake-soloff/PairSwap-ICI-Experiments}. 
We will examine two versions of the \textnormal{\texttt{PairSwap-ICI}} test:
\begin{itemize}
    \item Neighbor matching (Algorithm~\ref{alg:neighbour-matching}), with the linear kernel $\psi(x,x') = x-x'$ and weights $w_\ell = \max\{Y_{i_\ell} - Y_{j_\ell},0\}$ as discussed in Section~\ref{sec:heuristic matching};
    \item Cross-bin matching (Algorithm~\ref{alg:cross-bin-matching}) with $K_n=\left\lfloor n^{2/3}\right\rfloor$ bins, again with the linear kernel $\psi(x,x') = x-x'$ and weights $w_\ell = \max\{Y_{i_\ell} - Y_{j_\ell},0\}$. This particular choice of $K_n$ is motivated by Theorem~\ref{thm:cross_bin_matching}, which suggests choosing the number of bins $K_n$ to satisfy $\sqrt{n \log n}\ll K_n \ll n$, in order to achieve competitive power with neighbor matching. 
\end{itemize}

\subsection{Extent of conservativeness under \texorpdfstring{$H_0^{\textnormal{ICI}}$}{the isotonic null}}\label{sec:validity_uni}

\cref{thm:main} establishes valid, finite-sample Type~I error control for our method. The purpose of this section is to evaluate the extent of the conservativeness of the Type~I error under various null distributions. Since our inference relies on the fact that matched pairs $(X_{i_\ell}, X_{j_\ell})$ are stochastically ordered under the null, intuitively this should depend on the strength of the monotonicity in the conditional distribution. 

To see how the dependence between $X$ and $Z$ affects the rejection probability, we sample $X$ from an additive noise model 
\[
X \mid Y, Z \sim \mathcal{N}\bigl(\mu(\gamma Z), 1\bigr),
\]
where $Y,Z$ are independent standard normal random variables. As long as $\mu$ is nondecreasing and $\gamma \ge 0$, this joint distribution belongs to the null~$H_0^{\textnormal{ICI}}$. The scalar $\gamma$ controls the strength of the monotonicity of $X\mid Z$. In particular, as $\gamma \searrow 0$ we expect the Type~I error $\PP{p\le \alpha}$ to approach $\alpha$.

In our simulations, we consider two functions $\mu$, the identity $\mu(z)=z$ and the standard Gaussian CDF $\mu(z) = \Phi(z)$. \cref{fig:validity} shows the Type~I error as a function of $\gamma$ for two levels of $\alpha$. We observe similar results for each $\alpha$, where the test typically becomes more conservative as $\gamma$ increases, as expected. Under the null, our test is more conservative for cross-bin matching than for neighbor matching, since the $Z$ values are further apart in cross-bin matching.  We remark that while other methods with asymptotic Type~I error guarantees may exhibit substantial violations of Type~I error control in finite samples, our procedure \texttt{PairSwap-ICI} test offers finite-sample validity.

\begin{figure}[!ht]
    \centering
    \fbox{\includegraphics[scale=0.8]{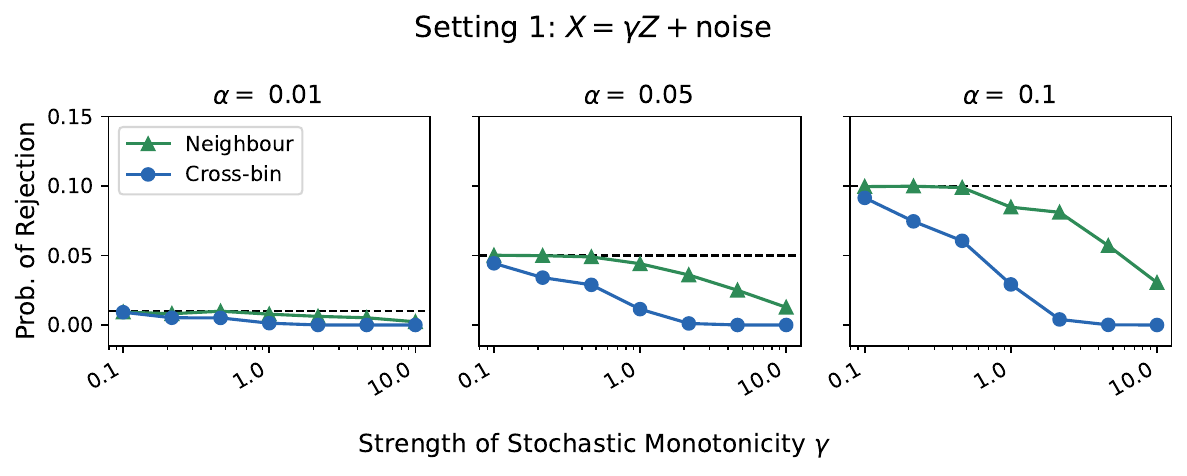}}\\ \bigskip
    \fbox{\includegraphics[scale=0.8]{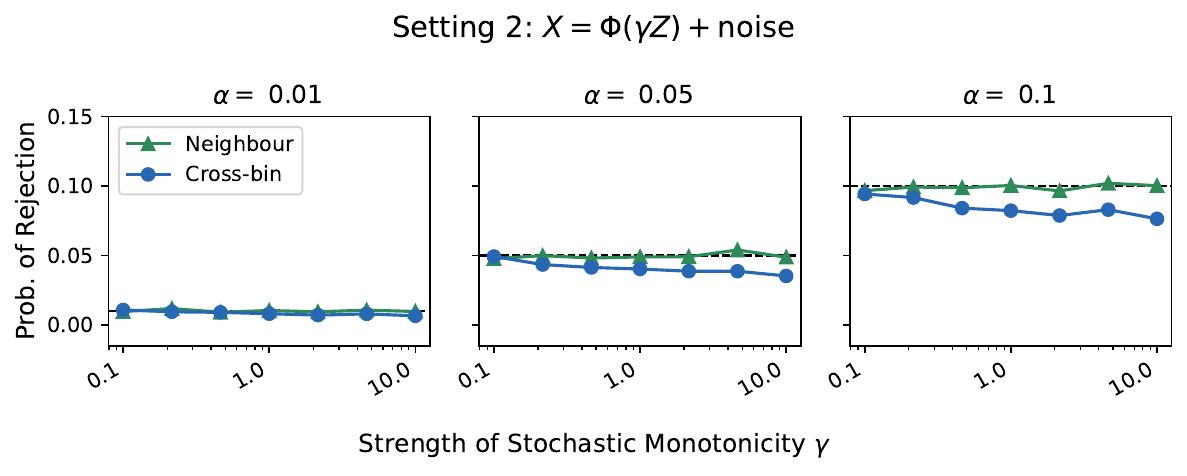}}
    \caption{Simulation results illustrating Type~I error control under the null~$H_0^{\textnormal{ICI}}$ for two forms of the conditional mean~$\EEst{X}{Z}$. Each subplot shows the rejection probability of the \textnormal{\texttt{PairSwap-ICI}} test on a dataset of size $1000$, averaged over $10^4$ simulation trials, as a function of the strength of stochastic monotonicity~$\gamma$.}
    \label{fig:validity}
\end{figure}

\begin{figure}[!ht]
    \centering
    \fbox{\includegraphics[scale=0.895]{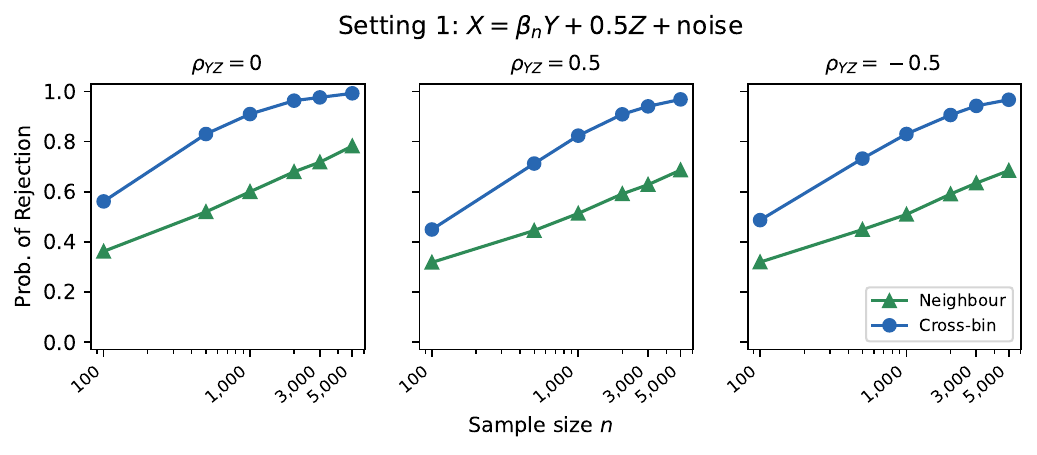}}\\ \bigskip
    \fbox{\includegraphics[scale=0.895]{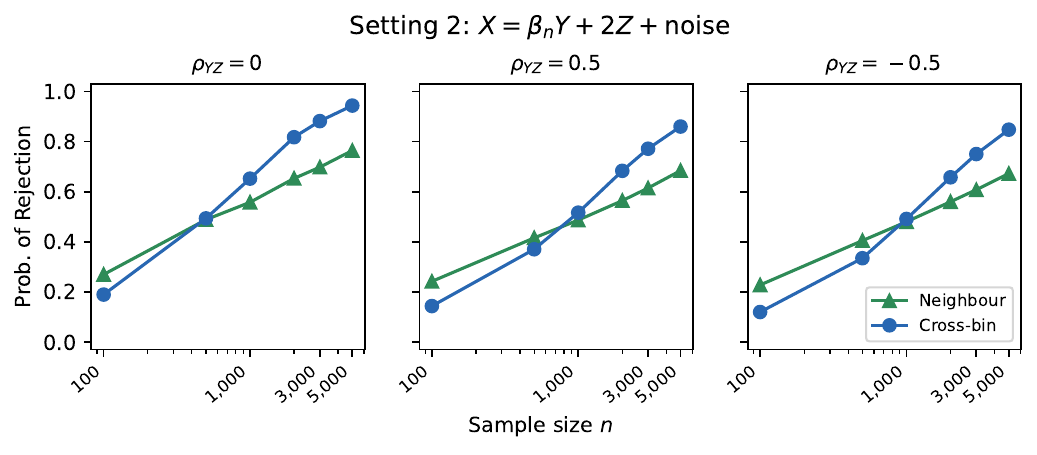}}
    \caption{Simulation results demonstrating the power of the \texttt{PairSwap-ICI} for two alternatives at level $\alpha=0.1$. Each subplot shows the rejection probability, averaged over $10^4$ simulation trials, as a function of the sample size~$n$. 
    Columns correspond to different relationships between $Y$ and $Z$. 
    In each setting, $X$ follows a Gaussian linear model with mean $\beta_n Y + \gamma Z$, where $\beta_n = n^{-1/3}$ and $\gamma = 0.5$ (above) or $\gamma = 2$ (below).
    }\label{fig:power-sim}
\end{figure}

\subsection{Power under alternatives}\label{sec:power_uni}

In Section~\ref{sec:partial_linear_model} we showed theoretically that our heuristic methods---neighbor matching and cross-bin matching---achieve high power (in fact, power that tends to 1) in the partial linear model~\eqref{model:partial_linear_model} with $\sigma_n=1$ for all $n$, provided the signal~$\beta_n$ exceeds the detection threshold~$n^{-1/2}$. In this section, we now examine this setting empirically. We sample data from the Gaussian linear model 
\[
X \mid Y, Z \sim \mathcal{N}\bigl(\beta_n Y + \gamma Z, 1\bigr),
\]
with $\beta_n = n^{-1/3}$. The pair $(Y,Z)$ is drawn from a bivariate Gaussian 
\[
\mathcal{N}_2\biggl(\boldsymbol{0}, \begin{bmatrix}
    1 & \rho_{YZ} \\
    \rho_{YZ} & 1
\end{bmatrix}\biggr).
\]
\cref{fig:power-sim} shows the (unconditional) power as a function of the sample size~$n$ for various choices of $\gamma$ and $\rho_{YZ}$. In Setting 1, we set $\gamma = 0.5$, and cross-bin matching uniformly dominates neighbor matching because it allows us to make many more matches of similar quality. On the other hand, in Setting 2 we set $\gamma = 2$, so the strong dependence of $X$ on $Z$ means the quality of a match $(i_\ell, j_\ell)$ degrades much more quickly as the gap $Z_{j_\ell} - Z_{i_\ell}$ increases---that is, for cross-bin matching, where $Z_{i_\ell}$ and $Z_{j_\ell}$ may be farther apart than for neighbor matching, this gap may lead to conservativeness that results in a loss of power. However, with a sufficiently large sample size, cross-bin matching performs at least as well as the neighbor matching. This is because the bin width decreases as~$n$ increases, so the quality of the cross-bin matches rivals that of the neighbor matches (with many more matches). The dependence $\rho_{YZ}$ between $Y$ and $Z$ does not have a major impact on the power of these two methods.

One take-away message from these simulations is that the stronger the stochastic monotonicity is believed to be, the more bins one should use in cross-bin matching, to avoid excessive conservativeness.  
On the other hand, if the extent of stochastic monotonicity is believed to be weak, then one should use fewer bins in cross-bin matching, to allow for a greater number of matches as well as larger weights for the matched pairs, and thereby increase power.

\subsection{Empirical performance in multivariate settings}\label{sec:mult_expt}

We now assess the empirical performance of the \texttt{PairSwap-ICI} procedure when the variables $X,Y,Z$ may all be multivariate, and we seek to test $H_{0,d_X}^{\textnormal{ICI}}$.  Fix $d_X = d_Y=5$ and $d_Z \in \{2,3\}$ and generate 
\[
Y\sim \mathcal{N}(0,I_{d_Y}), \qquad Z\sim \mathcal{N}(0,\Sigma),
\]
where $\Sigma$ is an equicorrelation covariance matrix with $\Sigma_{jj}=1$ and $\Sigma_{jk}=0.8$ for $j\neq k$. Independently, we generate $\zeta 
\sim \mathcal{N}(0,I_{d_Y})$. 

We consider three data generating mechanisms for the response $X$, designed to capture different forms of conditional dependence between $X$ and $Y$ given $Z$. In each setting we study two regimes: a null case with $\beta=0$, corresponding to conditional independence, and an alternative case with $\beta=0.5$, where conditional dependence is present.

\begin{enumerate}[(i)]
\item \textbf{Setting 1 (Linear dependence):} $X=\beta Y+0.1\sum_{j=1}^{d_Z}Z_j\,\mathbf{1}_{d_Y}+\zeta$.

\item \textbf{Setting 2 (Quadratic dependence):} $X=\beta Y^{2}+0.1\sum_{j=1}^{d_Z}Z_j\,\mathbf{1}_{d_Y}+\zeta$.
where $Y^{2}$ denotes the elementwise square of $Y$.

\item \textbf{Setting 3 (Interaction dependence):} $X=(0.1+\beta Y)\sum_{j=1}^{d_Z}Z_j+\zeta$.
\end{enumerate}

For each dataset, matched pairs are constructed using a simple random pairing strategy: observations are paired at random, and only those pairs satisfying the coordinate-wise order constraint on the $\Zcal$ space are retained. These retained pairs form the matched sample used by the \texttt{PairSwap-ICI} test. 
This matching scheme inevitably leads to a loss in effective sample size, since many $Z$ pairs are not comparable under the coordinate-wise order.  Nevertheless, as the experiments below demonstrate, even this simple matching strategy can achieve reasonable power.

With the aforementioned matching scheme, we will use the plug-in weights from~\eqref{def:plug-in-matching} and a projection direction $u=\Ep{v\sim \text{Unif}(S_{d_X}^+)}{v}$. We estimate the conditional mean function using sample splitting. The dataset is divided into two equal parts: the first split is used to estimate the conditional mean $\EEst{X}{Y,Z}$ required for the plug-in weights, and the second split is used to run the \texttt{PairSwap-ICI} testing procedure. This separation prevents the weight estimation step from introducing bias in the testing stage.  We evaluate performance across sample sizes  $n\in\{200,400,600,800,1000,1500,2000\}$,
and repeat each experiment over $1000$ Monte Carlo trials.

Figure~\ref{fig:pairswap_multivariate} presents the Type~I error curves and power curves of the \texttt{PairSwap-ICI} test as functions of the sample size for the three data generating mechanisms. Across all settings, the Type~I error curves remain below the pre-determined significance level $\alpha=0.1$, as anticipated by Theorem~\ref{thm:main-general}.  Under the alternative $\beta=0.5$, the power increases steadily with the sample size and approaches one as $n \rightarrow \infty$. These results indicate that the proposed method is capable of detecting a variety of conditional dependence structures even when the response variable and the conditioning variables are multivariate.

\begin{figure}[t]
    \centering
    \includegraphics[width=0.9\linewidth]{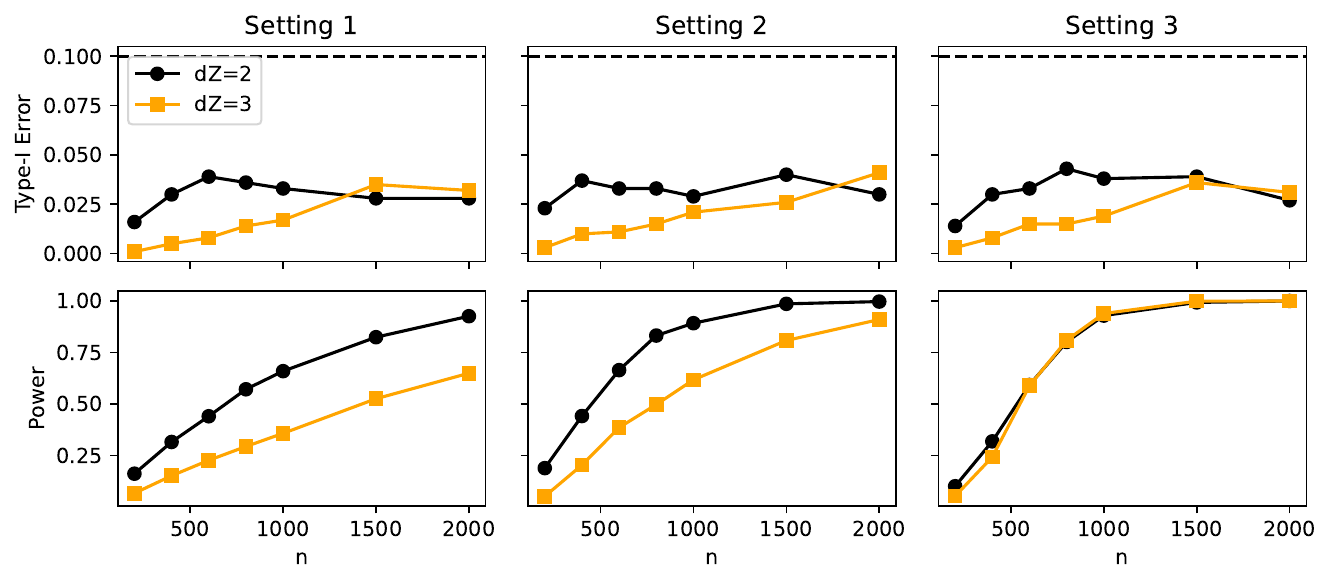}
    \caption{Empirical rejection probabilities of the \texttt{PairSwap-ICI} test in multivariate settings. 
    Each column corresponds to one of the three data generating mechanisms, while the curves show results for different covariate dimensions $d_Z\in\{2,3\}$. 
    The first row reports Type~I error rates under the null ($\beta=0$), and the second row reports power under the alternative ($\beta=0.5$).}
    \label{fig:pairswap_multivariate}
\end{figure}

\paragraph{Additional experiments.} 
In Appendix~\ref{app:additional_experiments}, we compare the Type~I error control and power of the proposed \texttt{PairSwap-ICI} test with several well-known benchmark methods for conditional independence testing under the aforementioned simulation settings, namely the Generalized Covariance Measure (GCM), Projected Covariance Measure (PCM), and Conditional Randomization Test (CRT).

\section{Experiment on real data: risk factors for diabetes}\label{sec:experiments}

We now evaluate the performance of our proposed testing procedure on a dataset\footnote{The data for this experiment were obtained from \url{https://archive.ics.uci.edu/dataset/34/diabetes}. Additional data descriptions can be found in \citet{smith1988using}.} on the incidence of diabetes among the Pima population near Phoenix, Arizona, originally collected by the US National Institute of Diabetes and Digestive and Kidney Diseases. The dataset comprises $768$ observations, and includes information on whether each patient has been diagnosed with diabetes according to World Health Organization standards. Additional variables provide data on the number of pregnancies, plasma glucose concentration, diastolic blood pressure, triceps skinfold thickness, 2-hour serum insulin levels, body mass index (BMI), diabetes pedigree function and age. 
Recognizing that age is a primary risk factor for diabetes, we examine whether diabetes incidence is independent of any other individual risk factors (e.g., glucose concentration, BMI, insulin level), conditional on age.  

It is well-known that the likelihood of developing diabetes increases with age (e.g., the Centers for Disease Control and Prevention\footnote{For more details, see 
\url{https://www.cdc.gov/diabetes/risk-factors/index.html} for a list of diabetes risk factors as officially reported by the US Centers for Disease Control and Prevention.} lists advanced age as one of the risk factors for Type~$1$ and Type~$2$ diabetes). Therefore, if we choose $X$ to be a binary variable representing incidence of diabetes and~$Z$ as the age of the patient, then we would expect $X$ to exhibit stochastic monotonicity with respect to $Z$ (i.e., we expect that Assumption~\ref{asm:st} holds, at least approximately).
% This is supported by the nondecreasing trend we observe in \cref{fig:monotonicity of age and diabetes}.
Most of the other variables, such as BloodPressure, BMI, Glucose, and Pregnancies, are also considered potential risk factors for diabetes—but does this association remain
after we control for age? See Figure~\ref{fig:monotonicity of risk factors for diabetes} for a visualization of the association between diabetes and each of these variables.

\begin{figure}[!t]
    \centering
    \includegraphics[width=\textwidth]{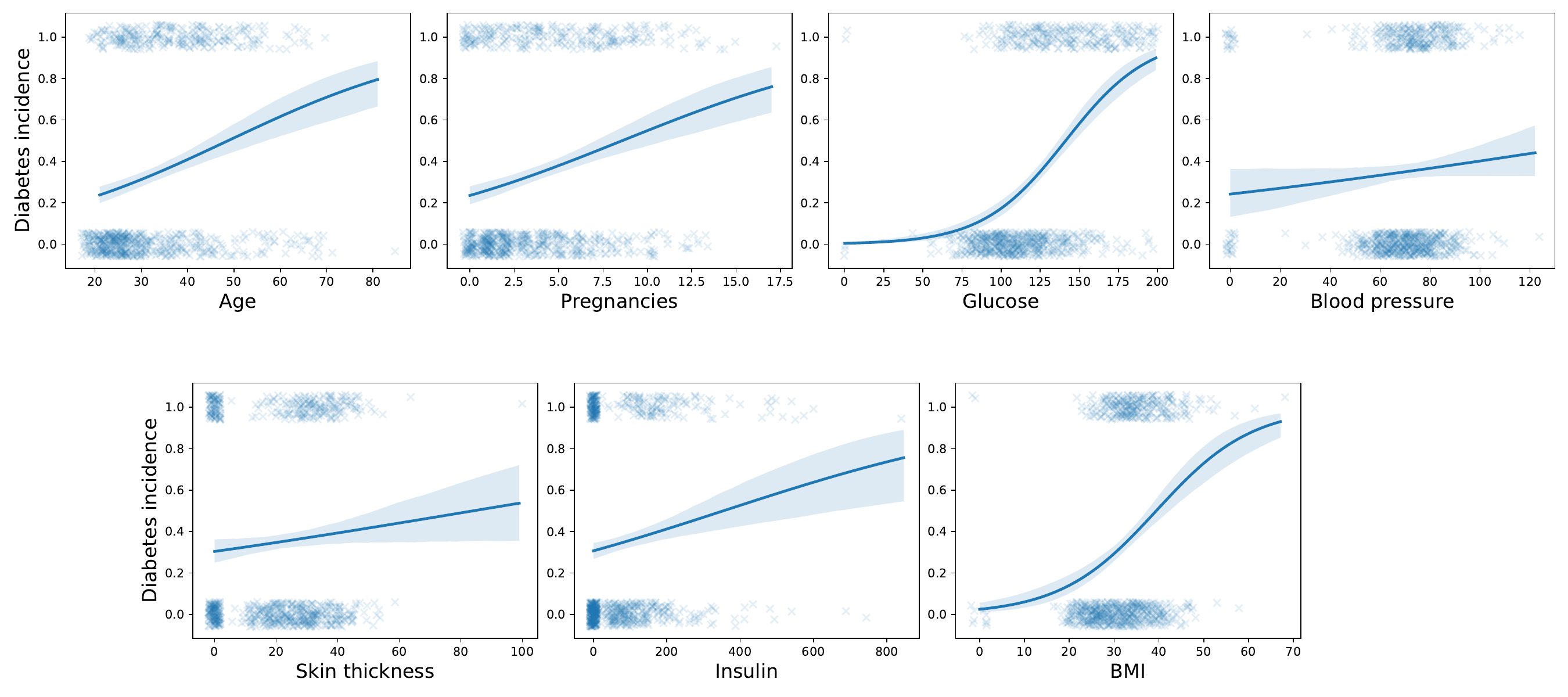}
    \caption{Scatter plots of $X$ (jittered for better visibility) and other feature variables along with the fitted logistic regression models to demonstrate the dependence among these variables and \texttt{Diabetes Incidence}.}
    \label{fig:monotonicity of risk factors for diabetes}
\end{figure}

% \begin{figure}[!ht]
%     \centering
%     \includegraphics[width=0.45\linewidth]{figures/relation_with_age.pdf}
%     \caption{A scatter plot (jittered for better visibility) of \texttt{Age} and \texttt{Diabetes Incidence} along with the fitted logistic regression model to demonstrate the stochastic monotonicity.}
%     \label{fig:monotonicity of age and diabetes}
% \end{figure}

\paragraph{Experiment $1$: marginal independence testing:}
We start with a warm-up experiment. We consider six variables: \texttt{Pregnancies}, \texttt{Glucose}, \texttt{BloodPressure}, \texttt{SkinThickness}, \texttt{Insulin}, and \texttt{BMI}, and aim to assess whether each of them is an individual risk factor for diabetes incidence. Specifically, we test the hypothesis $H_0: X \independent Y$, where $Y$ represents one of the six variables listed above, while $X$ is \texttt{Diabetes} (and since we are testing marginal rather than conditional independence, we do not attempt to control for $Z$, i.e., \texttt{Age}). For this purpose, we will be using the permutation test for independence with $T(\bX,\bY)=\bX^\top \bY$, as outlined in Section~\ref{sec: marginal independence}.

\begin{table}[ht!]
\centering
\scalebox{0.85}{
\begin{tabular}{|c|c|c|}
\hline
$X$ & $Y$ & Experiment $1$ \\
\hline
\multirow{6}{*}{\makecell{Diabetes \\ incidence}}
& Pregnancies 
& \textbf{0.001} (0.000)\\
% \cline{3-7}
& Glucose 
& \textbf{0.001} (0.000)  \\
% \cline{3-7}
&  Blood pressure & 0.155 (0.003) \\
% \cline{3-7}
& Skin thickness 
& 0.117 (0.002) \\
% \cline{3-7}
&  Insulin 
& \textbf{0.023} (0.001) \\
% \cline{3-7}
&  BMI 
& \textbf{0.001} (0.000) \\
\hline
\end{tabular}}
\caption{$p$-values of Experiment~$1$, averaged over random $3{,}000$ subsamples with standard errors in parentheses. Values significant at the $0.05$ level are shown in bold.}
\label{tab:p_values_marginal}
\end{table}
\paragraph{Experiment $2$: conditional independence testing, after controlling for age:}
Next, for the same set of six choices for $Y$, we test the hypothesis $H_0^{\textnormal{ICI}} : X \independent Y \mid Z$, where~$Z$ denotes \texttt{Age} (and $X$ is \texttt{Diabetes} as before). This allows us to identify risk factors for diabetes after controlling for age. As noted earlier, we expect the distribution of $X \mid Z = z$ to be stochastically monotone in~$z$, which supports the application of the \textnormal{\texttt{PairSwap-ICI}} testing procedure developed in this paper for this purpose. The \textnormal{\texttt{PairSwap-ICI}} method is implemented with neighbor matching or with cross-bin matching (with $K=50$ bins), and our statistic $T$ takes the form in \eqref{def:sum-of-psi} with linear kernel $\psi(x,x')=x-x'$ and weights $w_\ell=\max\{Y_{i_\ell}-Y_{j_\ell},0\}$.

\paragraph{Experiment $3$: conditional independence testing, with synthetic control $\tilde X$:}
Finally, we consider a semi-synthetic experiment where $X$ is replaced by synthetic observations~$\tilde{X}$, generated from an estimated model for $P_{X\mid Z}$ that satisfies stochastic monotonicity. We then test the hypothesis $\tilde X\independent Y\mid Z$ for the same choices of $Y$ from Experiment~$1$. Since $\tilde X$ is generated solely based on $Z$ and obeys stochastic monotonicity, the null hypothesis $H_0^{\textnormal{ICI}}$ holds trivially in this synthetic setting. The validity of our procedure should therefore ensure that the $p$-values generated by \textnormal{\texttt{PairSwap-ICI}} are (super-)uniformly distributed. 

To generate the synthetic feature $\tilde{X}$, since $X$ is binary, it suffices to fit an isotonic regression to estimate the conditional mean
$\EEst{X}{Z}$ and then sample $\tilde X$ from the Bernoulli distribution with this fitted conditional mean. Following the theory of \citet[Theorem~1]{henzi2021isotonic}, this is the best approximation to $P_{X\mid Z}$ under the \textit{continuous ranked probability score (CRPS)}, while respecting the monotonicity constraint.

\begin{table}[ht!]
\centering
\scalebox{0.85}{
\begin{tabular}{|c|c|c|c|c|c|c|}
\hline
\multicolumn{3}{|c|}{} 
& \multicolumn{2}{c|}{Experiment 2} 
& \multicolumn{2}{c|}{Experiment 3} \\
\hline
$X$ & $Z$ & $Y$ 
& \makecell{Neighbor \\ matching} 
& \makecell{Cross-bin \\ matching}
& \makecell{Neighbor \\ matching} 
& \makecell{Cross-bin \\ matching} \\
\hline
\multirow{6}{*}{\makecell{Diabetes \\ incidence}}
& \multirow{6}{*}{Age}
& Pregnancies 
& 0.433 (0.005) & 0.424 (0.004)
& 0.511 (0.005) & 0.544 (0.005) \\
% \cline{3-7}
& & Glucose 
& \textbf{0.002} (0.000) & \textbf{0.001} (0.000)
& 0.499 (0.005) & 0.550 (0.005) \\
% \cline{3-7}
& & Blood pressure 
& 0.496 (0.005) & 0.499 (0.004)
& 0.501 (0.005) & 0.555 (0.005) \\
% \cline{3-7}
& & Skin thickness 
& 0.243 (0.004) & 0.123 (0.003)
& 0.506 (0.005) & 0.558 (0.005) \\
% \cline{3-7}
& & Insulin 
& 0.159 (0.003) & 0.224 (0.004)
& 0.497 (0.005) & 0.545 (0.005) \\
% \cline{3-7}
& & BMI 
& \textbf{0.033} (0.001) & \textbf{0.001} (0.000)
& 0.499 (0.005) & 0.555 (0.005) \\
\hline
\end{tabular}}
\caption{\texttt{PairSwap-ICI} $p$-values, averaged over random $3{,}000$ subsamples with standard errors in parentheses, from Experiments~2 and~3. Values significant at the $0.05$ level are shown in bold.}
\label{tab:p_values_univar}
\end{table}

\paragraph{Results:}
 For each experiment, we repeat the following procedure $3{,}000$ times. We randomly split the full dataset into two halves: a training set and a test set.  For Experiment 3, since constructing a synthetic control requires estimating $P_{X\mid Z}$, we use the training half for this preliminary estimation problem (for Experiments~1 and~2 the training set is not used), and the test set to compute the test statistic and the corresponding $p$-value.
We compute $p$-values on the test set using the permutation test for marginal independence and the \textnormal{\texttt{PairSwap-ICI}} test for conditional independence. Finally, we report the average of the $3{,}000$ test-set $p$-values, along with the corresponding choice of the~$Y$ variable, in Tables~\ref{tab:p_values_marginal} and~\ref{tab:p_values_univar}.

Under the marginal independence test, four of the six variables are identified as having significant association with \texttt{Diabetes}, but, once we test conditional independence with the \textnormal{\texttt{PairSwap-ICI}} test, only two of these associations are identified as significant. 
Specifically, \texttt{Glucose} and $\texttt{BMI}$ both are identified as potential risk factors at the $0.05$ level of significance by the marginal independence test, and also by the \textnormal{\texttt{PairSwap-ICI}} test, even after controlling for \texttt{Age}.
On the other hand, the variables \texttt{Pregnancies} and \texttt{Insulin} are significant only under the marginal test; this suggests that, after controlling for \texttt{Age}, the data does not provide sufficient evidence to support them as risk factors for \texttt{Diabetes}.

Finally, we also note that all the averaged $p$-values from Experiment~3 with synthetic control~$\tilde{X}$ are concentrated around $0.5$, for each of the choices of $Y$. Since $(\tilde{X},Y,Z)$ satisfy $H_0^\textnormal{ICI}$, the $p$-values from Experiment~3 should be roughly uniform (or, if the test is conservative, super-uniform), and thus this behavior is expected from Theorem~\ref{thm:main}.

\section{Discussion}\label{sec:discussion}

In this paper, we have developed a nonparametric test of conditional independence assuming only stochastic monotonicity of the conditional distribution $P_{X\mid Z}$. 
This nonparametric constraint is natural in many applications, and allows us to circumvent the impossibility of assumption-free conditional independence testing \citep{shah2020hardness}.
We have introduced a variety of approaches to constructing a valid test statistic. 
Our test controls the Type~I error in finite samples and has power against an array of alternatives. We close with some connections to the literature, and potential avenues for future work.

\begin{itemize}
    \item \emph{Optimal power in general settings.} 
    In Theorem~\ref{thm:cross_bin_matching}, we were able to eradicate the earlier gap between the upper and lower bounds on the asymptotic conditional power of our test under the assumption that the conditional distribution of $Y\mid Z$ is symmetric, but it remains an open question whether other tests (or, perhaps, the \textnormal{\texttt{PairSwap-ICI}} but with a different kernel) may be able to avoid this assumption.
    \item \emph{Alternative methods for finite-sample Type~I error control?} While the stochastic monotonicity assumption facilitates consistent estimation of the conditional distribution~$P_{X\mid Z}$ using isotonic distributional regression \citep{mosching2020monotone,henzi2021isotonic}, and then applying a conditional independence test which assumes knowledge (or a good estimate) of $P_{X\mid Z}$ \citep{berrett2020conditional,candes2018panning}, these tests will only be valid asymptotically. However, under the same setting, the proposed \texttt{PairSwap-ICI} test enjoys finite sample Type~I error control without any further restrictions on the choice of matching, or weighting. An interesting open question would be to investigate whether these established approaches for CI testing can also be modified to attain validity in finite samples under the monotonicity assumption.
    \item \emph{Avoiding data splitting.} The oracle matching test derived in Section~\ref{sec:design} requires modelling the conditional mean and conditional variance of the kernel $\psi(X_i, X_j)$ as a function of $Y_i,Y_j,Z_i,Z_j$. We proposed to estimate these moments on a hold-out dataset. Can we instead perform cross-fitting to improve power and retain finite-sample error control? 
    \item \emph{Connection with knockoffs and conditional randomization tests.} Creating synthetic copies of $\bX$ via pairwise swaps resembles other conditional independence testing procedures, such as knockoffs and the conditional randomization test \citep{candes2018panning}, and the conditional permutation test \citep{berrett2020conditional}. An advantage of the \texttt{PairSwap-ICI} test is that we do not require knowledge of $P_{X|Z}$ apart from its stochastic monotonicity, though the potential price is that the test may be more conservative (i.e., the resulting $p$-value may be super-uniform), since we are not working under the ``sharp null''. 
    \item \emph{Alternative shape constraints.} We view stochastic monotonicity as a natural form of positive dependence for the joint distribution~$(X, Z)$. Are there approaches to test conditional independence under other models of dependence, such as likelihood ratio ordering or total positivity, or under other shape constraints, such as unimodality of $P_{X|Z}$ \citep{karlin1968total,shaked2007stochastic,mosching2024estimation}?
\end{itemize}
\subsection*{Acknowledgements}
J.A.S gratefully acknowledges the support of the Office of Naval Research via grant N00014-20-1-2337, and the Margot and Tom Pritzker Foundation.
R.F.B. was partially supported by the National Science Foundation via grant DMS-2023109, and by the Office of Naval Research via grant N00014-24-1-2544. The research of R. J. S. was supported by European Research Council Advanced Grant 101019498.

\subsection*{Availability of Data}
The data for this experiment was obtained from \url{https://archive.ics.uci.edu/dataset/34/diabetes} and the codes for reproducing all experiments can be found at \url{https://github.com/jake-soloff/PairSwap-ICI-Experiments}.

% \subsection*{Conflict of Interest}
% We have no conflicts of interest to disclose.
% {
 \bibliographystyle{apalike}
 \bibliography{references}

@article{canonne2020survey,
  title={A survey on distribution testing: Your data is big. {B}ut is it blue?},
  author={Canonne, Cl{\'e}ment L},
  journal={Theory of Computing},
  pages={1--100},
  year={2020},
  publisher={Theory of Computing Exchange}
}

@book{pardo2018statistical,
  title={Statistical Inference based on Divergence Measures},
  author={Pardo, Leandro},
  year={2018},
  publisher={Chapman and Hall/CRC}
}

@book{le2000asymptotics,
  title={Asymptotics in Statistics: Some Basic Concepts},
  author={Le Cam, Lucien Marie and Yang, Grace Lo},
  year={2000},
  publisher={Springer Science \& Business Media}
}

@article {shah2020hardness,
    AUTHOR = {Shah, Rajen D. and Peters, Jonas},
     TITLE = {The hardness of conditional independence testing and the
              generalised covariance measure},
   JOURNAL = {The Annals of Statistics},
  FJOURNAL = {The Annals of Statistics},
    VOLUME = {48},
      YEAR = {2020},
    NUMBER = {3},
     PAGES = {1514--1538},
      ISSN = {0090-5364,2168-8966},
   MRCLASS = {62G10 (62G08)},
  MRNUMBER = {4124333},
MRREVIEWER = {Feng\ Yao},
       DOI = {10.1214/19-AOS1857},
       URL = {https://doi.org/10.1214/19-AOS1857},
}

@article {berrett2020conditional,
    AUTHOR = {Berrett, Thomas B. and Wang, Yi and Barber, Rina Foygel and
              Samworth, Richard J.},
     TITLE = {The conditional permutation test for independence while
              controlling for confounders},
   JOURNAL = {J. R. Stat. Soc. Ser. B. Stat. Methodol.},
  FJOURNAL = {Journal of the Royal Statistical Society. Series B.
              Statistical Methodology},
    VOLUME = {82},
      YEAR = {2020},
    NUMBER = {1},
     PAGES = {175--197},
      ISSN = {1369-7412,1467-9868},
   MRCLASS = {62G10},
  MRNUMBER = {4060981},
MRREVIEWER = {Melissa\ A.\ Bingham},
}

@article {candes2018panning,
    AUTHOR = {Cand\`es, Emmanuel and Fan, Yingying and Janson, Lucas and Lv,
              Jinchi},
     TITLE = {Panning for gold: `model-{$X$}' knockoffs for high dimensional
              controlled variable selection},
   JOURNAL = {J. R. Stat. Soc. Ser. B. Stat. Methodol.},
  FJOURNAL = {Journal of the Royal Statistical Society. Series B.
              Statistical Methodology},
    VOLUME = {80},
      YEAR = {2018},
    NUMBER = {3},
     PAGES = {551--577},
      ISSN = {1369-7412,1467-9868},
   MRCLASS = {62G10 (62H99 62J12 62J15 62P10)},
  MRNUMBER = {3798878},
       DOI = {10.1111/rssb.12265},
       URL = {https://doi.org/10.1111/rssb.12265},
}

@article{neykov2021minimax,
  title={Minimax optimal conditional independence testing},
  author={Neykov, Matey and Balakrishnan, Sivaraman and Wasserman, Larry},
  fjournal={The Annals of Statistics},
  JOURNAL = {The Annals of Statistics},
  volume={49},
  number={4},
  pages={2151--2177},
  year={2021},
  publisher={Institute of Mathematical Statistics}
}

@article {kim2022local,
    AUTHOR = {Kim, Ilmun and Neykov, Matey and Balakrishnan, Sivaraman and
              Wasserman, Larry},
     TITLE = {Local permutation tests for conditional independence},
   JOURNAL = {The Annals of Statistics},
  FJOURNAL = {The Annals of Statistics},
    VOLUME = {50},
      YEAR = {2022},
    NUMBER = {6},
     PAGES = {3388--3414},
      ISSN = {0090-5364,2168-8966},
   MRCLASS = {62G10 (62C20 62G09 62H20)},
  MRNUMBER = {4524501},
MRREVIEWER = {Zuzana\ Pr\'{a}\v{s}kov\'{a}},
       DOI = {10.1214/22-aos2233},
       URL = {https://doi.org/10.1214/22-aos2233},
}

@book {davison1997bootstrap,
    AUTHOR = {Davison, A. C. and Hinkley, D. V.},
     TITLE = {Bootstrap Methods and their Application},
 PUBLISHER = {Cambridge University Press, Cambridge},
      YEAR = {1997},
     PAGES = {x+582},
      ISBN = {0-521-57391-2},
   MRCLASS = {62G09 (62E20)},
  MRNUMBER = {1478673},
       DOI = {10.1017/CBO9780511802843},
       URL = {https://doi.org/10.1017/CBO9780511802843},
}

@article {phipson2010permutation,
    AUTHOR = {Phipson, Belinda and Smyth, Gordon K.},
     TITLE = {Permutation {$p$}-values should never be zero: calculating
              exact {$p$}-values when permutations are randomly drawn},
   JOURNAL = {Statistical Applications in Genetics and Molecular Biology},
  FJOURNAL = {Statistical Applications in Genetics and Molecular Biology},
    VOLUME = {9},
      YEAR = {2010},
     PAGES = {Art. 39, 14},
      ISSN = {1544-6115},
   MRCLASS = {99-01},
  MRNUMBER = {2746025},
       DOI = {10.2202/1544-6115.1585},
       URL = {https://doi.org/10.2202/1544-6115.1585},
}

@article{edmonds1965paths,
  title={Paths, trees, and flowers},
  author={Edmonds, Jack},
  journal={Canadian Journal of Mathematics},
  volume={17},
  pages={449--467},
  year={1965},
  publisher={Cambridge University Press}
}

@article{duan2014linear,
  title={Linear-time approximation for maximum weight matching},
  author={Duan, Ran and Pettie, Seth},
  journal={Journal of the ACM},
  volume={61},
  number={1},
  pages={1--23},
  year={2014},
  publisher={ACM New York, NY, USA}
}

@inproceedings{gabow1985scaling,
  title={A scaling algorithm for weighted matching on general graphs},
  author={Gabow, Harold N},
  booktitle={26th Annual Symposium on Foundations of Computer Science (SFCS 1985)},
  pages={90--100},
  year={1985},
  organization={IEEE}
}

@article{henzi2021isotonic,
  title={Isotonic distributional regression},
  author={Henzi, Alexander and Ziegel, Johanna F and Gneiting, Tilmann},
  journal={Journal of the Royal Statistical Society Series B: Statistical Methodology},
  volume={83},
  number={5},
  pages={963--993},
  year={2021},
  publisher={Oxford University Press}
}

@inproceedings{smith1988using,
  title={Using the {ADAP} learning algorithm to forecast the onset of diabetes mellitus},
  author={Smith, Jack W and Everhart, James E and Dickson, WC and Knowler, William C and Johannes, Robert Scott},
  booktitle={Proceedings of the Annual Symposium on Computer Application in Medical Care},
  pages={261},
  year={1988},
  organization={American Medical Informatics Association}
}

@article{lundborg2022projected,
  title={The Projected Covariance Measure for assumption-lean variable significance testing},
  author={Lundborg, Anton Rask and Kim, Ilmun and Shah, Rajen D and Samworth, Richard J},
  journal={The Annals of Statistics},
  volume={25},
  pages={2851--2878},
  year={2024}
}

@article{niu2024reconciling,
  title={Reconciling model-{$X$} and doubly robust approaches to conditional independence testing},
  author={Niu, Ziang and Chakraborty, Abhinav and Dukes, Oliver and Katsevich, Eugene},
  journal={The Annals of Statistics},
  volume={52},
  number={3},
  pages={895--921},
  year={2024},
  publisher={Institute of Mathematical Statistics}
}

@article{berrett2021usp,
  title={{USP}: an independence test that improves on Pearson’s chi-squared and the {$G$}-test},
  author={Berrett, Thomas B and Samworth, Richard J},
  journal={Proceedings of the Royal Society A},
  volume={477},
  number={2256},
  pages={20210549},
  year={2021},
  publisher={The Royal Society}
}

@article{albert2022adaptive,
  title={Adaptive test of independence based on {HSIC} measures},
  author={Albert, M{\'e}lisande and Laurent, B{\'e}atrice and Marrel, Amandine and Meynaoui, Anouar},
  journal={The Annals of Statistics},
  volume={50},
  number={2},
  pages={858--879},
  year={2022},
  publisher={Institute of Mathematical Statistics}
}

@article{pfister2018kernel,
  title={Kernel-based tests for joint independence},
  author={Pfister, Niklas and B{\"u}hlmann, Peter and Sch{\"o}lkopf, Bernhard and Peters, Jonas},
  journal={Journal of the Royal Statistical Society Series B: Statistical Methodology},
  volume={80},
  number={1},
  pages={5--31},
  year={2018},
  publisher={Oxford University Press}
}

@article{kim2022minimax,
  title={Minimax optimality of permutation tests},
  author={Kim, Ilmun and Balakrishnan, Sivaraman and Wasserman, Larry},
  journal={The Annals of Statistics},
  volume={50},
  number={1},
  pages={225--251},
  year={2022},
  publisher={Institute of Mathematical Statistics}
}

@article{berrett2021optimal,
  title={Optimal rates for independence testing via {$U$}-statistic permutation tests},
  author={Berrett, Thomas B and Kontoyiannis, Ioannis and Samworth, Richard J},
  journal={The Annals of Statistics},
  volume={49},
  number={5},
  pages={2457--2490},
  year={2021},
  publisher={Institute of Mathematical Statistics}
}

@article{berrett2019nonparametric,
  title={Nonparametric independence testing via mutual information},
  author={Berrett, Thomas B and Samworth, Richard J},
  journal={Biometrika},
  volume={106},
  number={3},
  pages={547--566},
  year={2019},
  publisher={Oxford University Press}
}

@article{lundborg2022conditional,
  title={Conditional independence testing in {H}ilbert spaces with applications to functional data analysis},
  author={Lundborg, Anton Rask and Shah, Rajen D and Peters, Jonas},
  journal={Journal of the Royal Statistical Society Series B: Statistical Methodology},
  volume={84},
  number={5},
  pages={1821--1850},
  year={2022},
  publisher={Oxford University Press}
}

@article{barber2020robust,
  title={Robust inference with knockoffs},
  author={Barber, Rina Foygel and Cand{\`e}s, Emmanuel J and Samworth, Richard J},
  journal={The Annals of Statistics},
  volume={48},
  number={3},
  pages={1409--1431},
  year={2020},
  publisher={JSTOR}
}

@article{kalisch2007estimating,
  title={Estimating high-dimensional directed acyclic graphs with the {PC}-algorithm.},
  author={Kalisch, Markus and B{\"u}hlmann, Peter},
  journal={Journal of Machine Learning Research},
  volume={8},
  number={3},
  pages={613--636},
  year={2007}
}

@book {shaked2007stochastic,
    AUTHOR = {Shaked, Moshe and Shanthikumar, J. George},
     TITLE = {Stochastic Orders},
    SERIES = {Springer Series in Statistics},
 PUBLISHER = {Springer, New York},
      YEAR = {2007},
     PAGES = {xvi+473},
      ISBN = {978-0-387-32915-4; 0-387-32915-3},
   MRCLASS = {60-02 (60E15 62H05 62N05)},
  MRNUMBER = {2265633},
MRREVIEWER = {B.\ L. S. Prakasa Rao},
       DOI = {10.1007/978-0-387-34675-5},
       URL = {https://doi.org/10.1007/978-0-387-34675-5},
}

@article{shevtsova2010improvement,
  title={An improvement of convergence rate estimates in the {L}yapunov theorem},
  author={Shevtsova, Irina G},
  journal={Doklady Mathematics},
  volume={82},
  number={3},
  pages={862--864},
  year={2010}
}

@article {mosching2020monotone,
    AUTHOR = {M\"osching, Alexandre and D\"umbgen, Lutz},
     TITLE = {Monotone least squares and isotonic quantiles},
   JOURNAL = {Electronic Journal of Statistics},
  FJOURNAL = {Electronic Journal of Statistics},
    VOLUME = {14},
      YEAR = {2020},
    NUMBER = {1},
     PAGES = {24--49},
      ISSN = {1935-7524},
   MRCLASS = {62G08 (62G20 62G30)},
  MRNUMBER = {4047593},
MRREVIEWER = {Thekke\ V.\ Ramanathan},
       DOI = {10.1214/19-EJS1659},
       URL = {https://doi.org/10.1214/19-EJS1659},
}

@article{mosching2024estimation,
author = {M\"{o}sching, Alexandre and D\"{u}mbgen, Lutz},
title = {Estimation of a likelihood ratio ordered family of distributions},
year = {2024},
issue_date = {Feb 2024},
publisher = {Kluwer Academic Publishers},
address = {USA},
volume = {34},
number = {1},
issn = {0960-3174},
url = {https://doi.org/10.1007/s11222-023-10370-9},
doi = {10.1007/s11222-023-10370-9},
journal = {Statistics and Computing},
month = dec,
pages = {58--74}
}

@book {karlin1968total,
    AUTHOR = {Karlin, Samuel},
     TITLE = {Total Positivity. {V}ol. {I}},
 PUBLISHER = {Stanford University Press, Stanford, CA},
      YEAR = {1968},
     PAGES = {xii+576},
   MRCLASS = {46.00 (41.00)},
  MRNUMBER = {230102},
MRREVIEWER = {I.\ I.\ Hirschman, Jr.},
}

@article{lei2020,
author = {Jing Lei},
title = {{Convergence and concentration of empirical measures under Wasserstein distance in unbounded functional spaces}},
volume = {26},
journal = {Bernoulli},
number = {1},
publisher = {Bernoulli Society for Mathematical Statistics and Probability},
pages = {767--798},
year = {2020},
doi = {10.3150/19-BEJ1151},
URL = {https://doi.org/10.3150/19-BEJ1151}
}

@book{rudin1987real,
  title={Real and Complex Analysis Third edition.},
  author={Rudin, Walter},
  pages={71},
  year={1987},
  publisher={McGraw-Hill Book Co., New York}
}

@article{Epskamp04072018,
author = {Sacha Epskamp and Lourens J. Waldorp and René Mõttus and Denny Borsboom},
title = {The {G}aussian graphical model in cross-sectional and time-series data},
journal = {Multivariate Behavioral Research},
volume = {53},
number = {4},
pages = {453--480},
year = {2018},
publisher = {Routledge},
doi = {10.1080/00273171.2018.1454823},
note ={PMID: 29658809}
}

@article{colombo12,
author = {Diego Colombo and Marloes H. Maathuis and Markus Kalisch and Thomas S. Richardson},
title = {{Learning high-dimensional directed acyclic graphs with latent and selection variables}},
volume = {40},
journal = {The Annals of Statistics},
number = {1},
publisher = {Institute of Mathematical Statistics},
pages = {294--321},
year = {2012},
doi = {10.1214/11-AOS940},
URL = {https://doi.org/10.1214/11-AOS940}
}
% }

\appendix
\counterwithin{theorem}{section}
\renewcommand{\thetheorem}{\Alph{section}.\arabic{theorem}}
\renewcommand{\thelemma}{\Alph{section}.\arabic{lemma}}

\section{Proof of Theorem~\ref{thm:mc-validity}}\label{app:proof-mc-validity}
The proof of this result follows the same structure as the proof of Theorem~\ref{thm:main}. 
\paragraph{Step 1: some deterministic properties of the $p$-value.}
Define a function $\hat{p}_M:\R^n\times (\{\pm 1\}^L)^M\to [0,1]$ by
\[
\hat{p}_M(\bx; \bs^{(1)},\dots,\bs^{(M)}) = \frac{1 +\sum_{m=1}^M \One{T(\bx^{\bs^{(m)}})\geq T(\bx)}}{1+M}.
\]
As in the proof of Theorem~\ref{thm:main}, this function is monotone nonincreasing in each $x_{i_\ell}$, and monotone nondecreasing in each $x_{j_\ell}$.

\paragraph{Step 2: compare to the sharp null.} 
Define $\bX_\sharp$ as in the proof of Theorem~\ref{thm:main}. Following identical arguments as in that proof, we can verify that, for any fixed $\bs^{(1)},\dots,\bs^{(M)}$, it holds that
\[\hat{p}_M\big(\bX_\sharp; \bs^{(1)},\dots,\bs^{(M)}\big) \preceq_{\rm st} \hat{p}_M\big(\bX; \bs^{(1)},\dots,\bs^{(M)}\big)\]
conditional on $\bY,\bZ$. Since $\hat{p}_M = \hat{p}_M(\bX;\bs^{(1)},\dots,\bs^{(M)})$ by construction, we therefore have
\[\PPst{\hat{p}_M\leq \alpha}{ \bY,\bZ,\bs^{(1)},\dots,\bs^{(M)}} \leq \PPst{\hat{p}_M\big(\bX_\sharp; \bs^{(1)},\dots,\bs^{(M)}\big)\leq \alpha}{ \bY,\bZ,\bs^{(1)},\dots,\bs^{(M)}}.\]
Marginalizing over the random draw of the swaps, $\bs^{(1)},\dots,\bs^{(M)}\iidsim \textnormal{Unif}(\{\pm 1\}^L)$, we therefore have
\[\PPst{\hat{p}_M\leq \alpha}{ \bY,\bZ} \leq \PPst{\hat{p}_M\big(\bX_\sharp; \bs^{(1)},\dots,\bs^{(M)}\big)\leq \alpha}{\bY,\bZ}.\]

\paragraph{Step 3: validity under the sharp null.} We now need to verify the validity of the Monte Carlo $p$-value, under the sharp null. Unlike the first two steps, for this step the arguments are somewhat different from the proof of Theorem~\ref{thm:main}.

First, let $\bs^{(0)}$ be an additional draw from $\textnormal{Unif}(\{\pm 1\}^L)$, sampled independently from all other random variables. Then
\[ 
(\bs^{(1)},\dots,\bs^{(M)})\eqd \big(\bs^{(0)}\circ \bs^{(1)},\dots,\bs^{(0)}\circ \bs^{(M)}\big),
\]
where $\circ$ denotes the elementwise product,
and so
\[
\hat{p}_M\big(\bX_\sharp; \bs^{(1)},\dots,\bs^{(M)}\big) \eqd \hat{p}_M\big(\bX_\sharp; \bs^{(0)}\circ \bs^{(1)},\dots,\bs^{(0)}\circ \bs^{(M)}\big)
\]
conditional on $\bY,\bZ$.
Moreover, by construction of the sharp null data $\bX_\sharp$, 
\[
\bX_\sharp\eqd (\bX_\sharp)^{\bs^{(0)}}
\]
holds conditional on $\bY,\bZ,\bs^{(0)},\bs^{(1)},\dots,\bs^{(M)}$, and therefore
\[
\hat{p}_M\Big(\bX_\sharp; \bs^{(0)}\circ \bs^{(1)},\dots,\bs^{(0)}\circ \bs^{(M)}\Big) \eqd \hat{p}_M\Big((\bX_\sharp)^{\bs^{(0)}}; \bs^{(0)}\circ \bs^{(1)},\dots,\bs^{(0)}\circ \bs^{(M)}\Big)
\]
holds conditional on $\bY,\bZ,\bs^{(0)},\bs^{(1)},\dots,\bs^{(M)}$. Combining all these calculations so far, then, 
\begin{equation}\label{eqn:eqd_monte_carlo}
\hat{p}_M\big(\bX_\sharp; \bs^{(1)},\dots,\bs^{(M)}\big) \eqd \hat{p}_M\Big((\bX_\sharp)^{\bs^{(0)}}; \bs^{(0)}\circ \bs^{(1)},\dots,\bs^{(0)}\circ \bs^{(M)}\Big),
\end{equation}
conditional on $\bY,\bZ$.

Next we calculate this last $p$-value: by definition,
\begin{align*}
    \hat{p}_M\Big((\bX_\sharp)^{\bs^{(0)}}; \bs^{(0)}\circ \bs^{(1)},\dots,\bs^{(0)}\circ \bs^{(M)}\Big)
    &=\frac{1 + \sum_{m=1}^M \One{T\bigl((\bX_\sharp)^{\bs^{(m)}}\bigr) \geq T\bigl((\bX_\sharp)^{\bs^{(0)}}\bigr)}}{1+M}\\
    &=\frac{\sum_{m=0}^M \One{T\bigl((\bX_\sharp)^{\bs^{(m)}}\bigr) \geq T\bigl((\bX_\sharp)^{\bs^{(0)}}\bigr)}}{1+M},
\end{align*}
where the first step holds since, for each $m \in [M]$, 
\[
\bigl((\bX_\sharp)^{\bs^{(0)}}\bigr)^{\bs^{(0)}\circ\bs^{(m)}} = (\bX_\sharp)^{\bs^{(0)}\circ \bs^{(0)}\circ \bs^{(m)}} = (\bX_\sharp)^{\bs^{(m)}}
\]
by definition of the swap operation. In other words, $\hat{p}_M\Big((\bX_\sharp)^{\bs^{(0)}}; \bs^{(0)}\circ \bs^{(1)},\dots,\bs^{(0)}\circ \bs^{(M)}\Big)$ simply compares the statistic $T\bigl((\bX_\sharp)^{\bs^{(0)}}\bigr)$ against the $M+1$ values $ T\bigl((\bX_\sharp)^{\bs^{(0)}}\bigr),\dots,T\bigl((\bX_\sharp)^{\bs^{(M)}}\bigr)$.
We therefore have
\[
\PPst{\hat{p}_M\Big((\bX_\sharp)^{\bs^{(0)}}; \bs^{(0)}\circ \bs^{(1)},\dots,\bs^{(0)}\circ \bs^{(M)}\Big) \leq \alpha}{\bX_\sharp,\bY,\bZ}
\leq \alpha,
\]
since, conditional on $\bX_\sharp,\bY,\bZ$, the sign vectors $\bs^{(0)},\dots,\bs^{(M)}$ are independent and identically distributed, and therefore the rank of $T((\bX_\sharp)^{\bs^{(0)}})$ among the list $ T((\bX_\sharp)^{\bs^{(0)}}),\dots,T((\bX_\sharp)^{\bs^{(M)}})$ is uniformly distributed.
Marginalizing over $\bX_\sharp$, therefore,
\[\PPst{\hat{p}_M\Big((\bX_\sharp)^{\bs^{(0)}}; \bs^{(0)}\circ \bs^{(1)},\dots,\bs^{(0)}\circ \bs^{(M)}\Big) \leq \alpha}{\bY,\bZ}
\leq \alpha.\]
Finally, combining this with our earlier calculation~\eqref{eqn:eqd_monte_carlo}, we have
\[\PPst{\hat{p}_M\big(\bX_\sharp; \bs^{(1)},\dots,\bs^{(M)}\big) \leq \alpha}{\bY,\bZ}
\leq \alpha,\]
which completes the proof. $\hfill{\square}$

%------------------------------------------------------------
\section{Proof of the results from Section~\ref{sec:power}}
In this section, we prove the results presented in Section~\ref{sec:power}. Throughout this appendix, we assume that the statistic $T$ admits the form in \eqref{def:sum-of-psi} with $\psi(x,x') = x-x'$, and that $\Ycal=\Zcal=\R$; the partial ordering $\preceq$ for $Z$ will simply be the usual ordering $\leq$ on $\R$. The organization of this appendix is as follows.
\begin{itemize}
    \item We begin in Appendix~\ref{app:general_power} by proving finite sample and asymptotic upper and lower bounds on the conditional power of \textnormal{\texttt{PairSwap-ICI}} test for any valid matching and weighting scheme, and  any statistic $T$ of the form \eqref{def:sum-of-psi} with a shared linear kernel $\psi(x,x')=x-x'$.
    \item Next, in Appendix~\ref{app:power-oracle}, we specialize these results to the oracle matching under two cases: one assuming access to oracle knowledge of $\mu$ (i.e., Theorem~\ref{thm:oracle-matching-asymptote}) and another with $\mu$ estimated from data (i.e., Theorem~\ref{thm:oracle-matching-estimated-mu}).
    \item Then, we shift our attention to the partially linear Gaussian models in \eqref{model:partial_linear_model}. In Appendices~\ref{app:power-neighbour} and~\ref{app:power-cross-bin} we prove the asymptotic behavior of conditional power for neighbor matching (Theorem~\ref{thm:neighbour_matching}) and for cross-bin matching (Theorem~\ref{thm:cross_bin_matching}), respectively.
    \item In Appendix~\ref{app:lemma-from-section-power} we prove the corollaries and lemmas from Section~\ref{sec:power}.
    \item In Appendix~\ref{app:IMM} we introduce an oracle matching, namely isotonic median matching, and discuss a key property of the same, which allows us to prove the results in Appendix~\ref{app:power-oracle}.
    \item Finally, in Appendix~\ref{app:lemmas-from-appendix-c1c4} we prove the lemmas used in Appendices~\ref{app:general_power}--\ref{app:power-cross-bin}.
\end{itemize}

\subsection{A general result on power of \textnormal{\texttt{PairSwap-ICI}} test}\label{app:general_power}
Here, we consider any valid matching and weighting scheme, and state finite-sample upper and lower bounds on the conditional power $\PPst{p\leq\alpha}{\bY,\bZ}$ of the \textnormal{\texttt{PairSwap-ICI}} test in Theorem~\ref{thm:general power analysis}. Further, under the asymptotic regime of Section~\ref{sec:power_general_alternatives}, we derive asymptotic high-probability upper and lower bounds for the same quantity in Corollary~\ref{thm:general-power-asymptotic}. 

\begin{theorem}\label{thm:general power analysis}
    Suppose that the triple $(\bX,\bY,\bZ)$ satisfies \eqref{model:general alternative}. Fix any matching $M=\{(i_1,j_1),\dots,(i_L,j_L)\}\in \mathcal{M}_n(\bZ)$, and any weight vector
    $\bw \in [0,\infty)^L$, which may both depend arbitrarily on $\bY,\bZ$ (but not on $\bX$). Assume that $\|\bw\|_2 > 0$ and fix any $\alpha\in(0,1)$. Let
    \[
    \epsilon_1 := \frac{\|\bw\|_\infty}{\|\bw\|_2}, \quad \epsilon_2 := \frac{\|\bw\circ\Delta\mu_n(\bY,\bZ)\|^3_3}{\sigma_n^3\|\bw\|^3_2 + \|\bw\circ\Delta\mu_n(\bY,\bZ)\|^3_2}.
    \]
    Then, writing $\bar{\Phi} := 1 - \Phi$, for any $\delta>0$,
    \begin{align}\label{eq:exact-upper-lower-bounf-on-conditional-power}
    &\Phi\Biggl(\frac{\bw^\top\Delta\mu_n(\bY,\bZ)}{\sqrt{2}\sigma_n\|\bw\|_2} -  \bar\Phi^{-1}(\alpha_{\rm{lo},\delta}) \cdot\sqrt{1 +\frac{ \|\bw\circ\Delta\mu_n(\bY,\bZ)\|^2_2}{2\sigma_n^2\|\bw\|^2_2}}\Biggr) -  (C'\epsilon_1 + \delta)\nonumber\\
    &\hspace{1cm}\leq \PPst{p\leq \alpha}{\bY,\bZ} \nonumber\\
    &\hspace{2cm}\leq 
    \Phi\Biggl(\frac{\bw^\top\Delta\mu_n(\bY,\bZ)}{\sqrt{2}\sigma_n\|\bw\|_2} -  \bar\Phi^{-1}(\alpha_{\rm{hi},\delta}) \cdot\sqrt{1 +\frac{ \|\bw\circ\Delta\mu_n(\bY,\bZ)\|^2_2}{2\sigma_n^2\|\bw\|^2_2}}\Biggr) +  (C'\epsilon_1 + \delta),
    \end{align}
where we write
$\alpha_{\rm{lo},\delta} = \alpha - C\delta^{-1}\cdot (\epsilon_1+\epsilon_2),\,
\alpha_{\rm{hi},\delta} = \alpha + C\delta^{-1}\cdot (\epsilon_1+\epsilon_2),$
with $C > 0$ depending only on $\alpha$ and on $\Ep{\zeta\sim P_\zeta}{\zeta^4}$, and $C' > 0$ depending only on $\Ep{\zeta\sim P_\zeta}{\zeta^4}$.
\end{theorem}
\begin{proof}
    Recall from \eqref{eq:def-test} that the $p$-value of the \textnormal{\texttt{PairSwap-ICI}} test is given by 
    \[
    p=\Pp{\bs\sim\textnormal{Unif}(\{\pm1 \}^L)}{T_{\bs}\geq T}=\Pp{\bs\sim\textnormal{Unif}(\{\pm1 \}^L)}{\bs^\top(\bw\circ\Delta\bX)\geq \mathbf{1}^\top (\bw\circ\Delta\bX)}.
    \] From here, the proof is divided into four key steps.
    
    \paragraph{Step 1: a CLT approximation to the $p$-value.}
    
    Conditional on $(\bX,\bY,\bZ)$ and with $\bs\sim\textnormal{Unif}(\{\pm1 \}^L)$, we see that
    $\bs^\top(\bw\circ\Delta\bX)$ is a sum of $L$ independent random variables, where the $\ell$th term $s_\ell\cdot w_\ell \cdot \Delta_\ell\bX$ has mean zero, variance $w_\ell^2(\Delta_\ell\bX)^2$, and third absolute moment $w_\ell^3|\Delta_\ell\bX|^3$ for $\ell \in [L]$. Therefore, by the Berry--Esseen theorem \citep[Theorem~1]{shevtsova2010improvement},  
    \begin{align}\label{eqn:test_as_fucntion_ofWcircDeltaX}
        \left|p-\bar\Phi\left(\frac{\bw^\top \Delta\bX}{\|\bw\circ \Delta\bX\|_2}\right)\right| &\leq \sup_{x\in \RR}\biggl|\PPst{\frac{\bs^\top (\bw\circ\Delta\bX)}{\|\bw\circ \Delta\bX\|_2}\geq x }{ \bX,\bY,\bZ } - \bar\Phi(x)\biggr|\notag\\&\leq 0.56 \cdot \frac{\|\bw\circ \Delta\bX\|_3^3}{\|\bw\circ \Delta\bX\|_2^3} =: \Gamma_1,    
        \end{align}
        as long as $\|\bw\circ \Delta\bX\|_2>0$. 
        Thus
        \[
        \One{\bar\Phi\left(\frac{\bw^\top \Delta\bX}{\|\bw\circ \Delta\bX\|_2}\right) < \alpha - \Gamma_1} \leq \One{p\leq \alpha}\leq \One{\bar\Phi\left(\frac{\bw^\top \Delta\bX}{\|\bw\circ \Delta\bX\|_2}\right) \leq \alpha + \Gamma_1},
        \]
        or equivalently (since $\bar\Phi$ is strictly decreasing), 
         \begin{align*}
         \One{\bw^\top \Delta\bX > \|\bw\circ \Delta\bX\|_2\cdot \bar\Phi^{-1}(\alpha - \Gamma_1)} &\leq \One{p\leq \alpha} \\ 
         &\leq \One{\bw^\top \Delta\bX \geq \|\bw\circ \Delta\bX\|_2\cdot \bar\Phi^{-1}(\alpha + \Gamma_1)}.
         \end{align*}
        Moreover, when $\|\bw\circ \Delta\bX\|_2=0$, we have $p=1$, so the above inequality holds for that case as well. 
        
        After marginalizing over the distribution of $\bX$, then,
        \begin{align} \label{eq:main-proof-at-step1-end}
    &\PPst{\bw^\top \Delta\bX > \|\bw\circ \Delta\bX\|_2\cdot \bar\Phi^{-1}(\alpha - \Gamma_1)}{\bY,\bZ} \nonumber \\
    &\hspace{2cm}\leq \PPst{p\leq \alpha}{\bY,\bZ}\leq \PPst{\bw^\top \Delta\bX \geq \|\bw\circ \Delta\bX\|_2\cdot \bar\Phi^{-1}(\alpha + \Gamma_1)}{\bY,\bZ}.
    \end{align}

    \paragraph{Step 2: approximating $\|\bw\circ \Delta\bX\|_2$ in \eqref{eq:main-proof-at-step1-end}.}
    Write $V_n := 2\sigma_n^2\|\bw\|^2_2 + \|\bw\circ\Delta\mu_n(\bY,\bZ)\|^2_2$ and define
    \[
    \Gamma_2 :=\biggl|\frac{\|\bw\circ\Delta\bX\|_2}{V_n^{1/2}} - 1\biggr|.
    \]
    Then by \eqref{eq:main-proof-at-step1-end}, we have that
    \begin{align*}
        &\PPst{p\leq \alpha}{\bY,\bZ} \geq \PPst{\frac{\bw^\top \Delta\bX}{V_n^{1/2}} > \frac{\|\bw\circ \Delta\bX\|_2}{V_n^{1/2}}\cdot \bar\Phi^{-1}(\alpha - \Gamma_1)}{\bY,\bZ}\\
       &\hspace{1cm}\geq\PPst{\frac{\bw^\top \Delta\bX}{V_n^{1/2} } > \bar\Phi^{-1}(\alpha - \Gamma_1)+\bigl|\bar\Phi^{-1}(\alpha - \Gamma_1)\bigr|\Gamma_2}{\bY,\bZ}.
    \end{align*}
    We can do a similar calculation for the upper bound on the conditional power from~\eqref{eq:main-proof-at-step1-end}. Together, we have that
    \begin{align*}
        &\PPst{\frac{\bw^\top \Delta\bX}{V_n^{1/2} } > \bar\Phi^{-1}(\alpha - \Gamma_1)+\left|\bar\Phi^{-1}(\alpha - \Gamma_1)\right|\Gamma_2}{\bY,\bZ}\\
        &\leq \PPst{p\leq \alpha}{\bY,\bZ}\leq \PPst{\frac{\bw^\top \Delta\bX}{V_n^{1/2} } \geq \bar\Phi^{-1}(\alpha + \Gamma_1) -\left|\bar\Phi^{-1}(\alpha + \Gamma_1)\right|\Gamma_2}{\bY,\bZ}.
    \end{align*}
    Now fix $\varepsilon_\alpha>0$. Define the event 
    \begin{align}\label{eqn:alpha_Delta_event}
    \Omega_\alpha &:= \bigl\{\bar\Phi^{-1}(\alpha - \Gamma_1)+\bigl|\bar\Phi^{-1}(\alpha - \Gamma_1)\bigr|\Gamma_2 \leq \bar\Phi^{-1}( \alpha - \varepsilon_\alpha)\bigr\} \nonumber \\ 
    &\hspace{3cm} \bigcap\bigl\{\bar\Phi^{-1}(\alpha + \Gamma_1) -\bigl|\bar\Phi^{-1}(\alpha + \Gamma_1)\bigr|\Gamma_2 \geq \bar\Phi^{-1}(\alpha + \varepsilon_\alpha)\bigr\}.
    \end{align}
    Then
    \begin{align}\label{eq:main-proof-at-step2-end}
    &\PPst{\frac{\bw^\top \Delta\bX}{V_n^{1/2} } >\bar\Phi^{-1}(\alpha-\varepsilon_\alpha)}{\bY,\bZ} -\PPst{\Omega_\alpha^c}{\bY,\bZ}\nonumber\\
    &\hspace{1.5cm}\leq \PPst{p\leq \alpha}{\bY,\bZ} \leq \PPst{\frac{\bw^\top \Delta\bX}{V_n^{1/2} } \geq \bar\Phi^{-1}(\alpha+\varepsilon_\alpha)}{\bY,\bZ} +  \PPst{\Omega_\alpha^c}{\bY,\bZ}.\end{align}
    
\paragraph{Step 3: a CLT approximation to $\bw^\top\Delta\bX$ in \eqref{eq:main-proof-at-step1-end}.}
    
    Conditional on $\bY,\bZ$, we have that $\bw^\top\Delta\bX$ is a sum of $L$ independent terms, where for the $\ell$th term we have
    \[\EEst{w_\ell\cdot\Delta_\ell\bX}{\bY,\bZ} = w_\ell\cdot \Delta_\ell\mu_n(\bY,\bZ),\]
    and
    \[\VVst{w_\ell\cdot\Delta_\ell\bX}{\bY,\bZ} = w_\ell^2\cdot \VVst{\Delta_\ell\bX}{\bY,\bZ} = w_\ell^2 \cdot 2\sigma_n^2.\]
    Moreover, we have
    \begin{align*}
        \EEst{\bigl|w_\ell \cdot \Delta_\ell\bX - \EEst{w_\ell \cdot\Delta_\ell\bX}{\bY,\bZ}\bigr|^3}{\bY,\bZ}  &= \EEst{|w_\ell\sigma_n(\zeta_{i_\ell} - \zeta_{j_\ell})|^3}{\bY,\bZ}\\
        &\leq 8 \sigma_n^3 w_\ell^3\Ep{P_\zeta}{|\zeta|^3} \leq 8\rho_4^3 \sigma_n^3w_\ell^3,
    \end{align*}
    where we write $\rho_4:=\Ep{P_\zeta}{\zeta^4}^{1/4}$.
    Therefore, another application of the Berry--Esseen theorem gives 
    \begin{align*}   
    \sup_{x\in \RR}\biggl|\PPst{\frac{\bw^\top\Delta\bX }{\sqrt{2}\sigma_n\|\bw\|_2}\geq x }{ \bY,\bZ } &- \Phi\left(\frac{\bw^\top\Delta\mu_n(\bY,\bZ)}{\sqrt{2}\sigma_n\|\bw\|_2}-x \right)\biggr|\\
    &\leq 0.56 \cdot \frac{8\sigma_n^3\rho^3_4\|\bw\|_3^3}{(\sqrt{2}\sigma_n^3\|\bw\|_2)^3} 
    \leq 1.59\rho_4^3\cdot\frac{\|\bw\|_\infty}{\|\bw\|_2} .\end{align*}
    By rescaling to match the denominator in \eqref{eq:main-proof-at-step2-end}, we then have
    \[
    \sup_{x\in \RR}\Biggl|\PPst{\frac{\bw^\top\Delta\bX }{V_n^{1/2}}\geq x }{ \bY,\bZ }- \Phi\left(\frac{\bw^\top\Delta\mu_n(\bY,\bZ)}{\sqrt{2}\sigma_n\|\bw\|_2} - \frac{ xV_n^{1/2}}{\sqrt{2}\sigma_n\|\bw\|_2}\right)\Biggr|\leq 1.59\rho_4^3\frac{\|\bw\|_\infty}{\|\bw\|_2}  .
    \]
    Moreover, by \eqref{eq:main-proof-at-step2-end} we now have
    \begin{align*} &
    \Phi\biggl(\frac{\bw^\top\Delta\mu_n(\bY,\bZ)}{\sqrt{2}\sigma_n\|\bw\|_2} -  \bar\Phi^{-1}(\alpha - \varepsilon_\alpha) \cdot\frac{ V_n^{1/2}}{\sqrt{2}\sigma_n\|\bw\|_2}\biggr) -  1.59\rho_4^3\cdot\frac{\|\bw\|_\infty}{\|\bw\|_2} - \PPst{\Omega_\alpha^c}{\bY,\bZ}\\
    &\hspace{0.6cm}\leq \PPst{p\leq \alpha}{\bY,\bZ}\\ 
    &\hspace{0.6cm}\leq 
    \Phi\biggl(\frac{\bw^\top\Delta\mu_n(\bY,\bZ)}{\sqrt{2}\sigma_n\|\bw\|_2} -  \bar\Phi^{-1}(\alpha +\varepsilon_\alpha) \cdot\frac{ V_n^{1/2}}{\sqrt{2}\sigma_n\|\bw\|_2}\biggr) + 1.59\rho_4^3\cdot\frac{\|\bw\|_\infty}{\|\bw\|_2} +  \PPst{\Omega_\alpha^c}{\bY,\bZ}.
    \end{align*}
    \paragraph{Step 4: a concentration result, and completion of the proof.}
    Finally, we show how to choose $\varepsilon_\alpha$ to control $\PPst{\Omega_\alpha^c}{\bY,\bZ}$ appropriately.  We will prove in Lemma~\ref{lem:concentration_L2_L3_for_DeltaX} that the terms  $\|\bw\circ\Delta\bX\|_2$ and $\|\bw\circ\Delta\bX\|_3$ are likely to be well-behaved---in particular, $\|\bw\circ\Delta\bX\|_2$ concentrates near $V_n^{1/2}$, while $\|\bw\circ\Delta\bX\|_3$ is vanishing. Concretely, Lemma~\ref{lem:concentration_L2_L3_for_DeltaX} will show that for any $\delta>0$,
    \begin{align*}
        \mathbb{P}\Bigl[\bigl\{\|\bw\circ \Delta\bX\|^2_2 \in  \bigl[V_n(1 -\epsilon_*)\vee 0,V_n(1+\epsilon_*)\bigr]\bigr\} \bigcap \bigl\{
        \|\bw\circ\Delta\bX\|^2_3 \leq V_n \epsilon_*'\bigr\}\Bigm| \bY,\bZ\Bigr]\geq 1-\delta,
    \end{align*}
    where
    \[
    \epsilon_* = \sqrt{\frac{3}{\delta}}\cdot\biggl(  2\rho^2_4\cdot \epsilon_1 + 2(\epsilon_1^2\epsilon_2)^{1/3}\biggr), \quad \epsilon_*' = \frac{(24)^{2/3}\rho_4^2}{\delta^{2/3}}\cdot\epsilon_1^{2/3} + 4\epsilon_2^{2/3}.
    \]
    Hence, since $(1+\epsilon_*)^{1/2} - 1 \leq \epsilon_*$ and $1 - \{(1-\epsilon_*) \vee 0\}^{1/2} \leq \epsilon_* \wedge 1 \leq \epsilon_*$, we have
   \[
   \PPst{ \Gamma_1 \leq \Gamma_{1,*},\Gamma_2 \leq \epsilon_*}{\bY,\bZ}\geq 1-\delta,
   \]
   where $\Gamma_{1,*} := 0.56\bigl(\epsilon'_*/\{(1-\epsilon_*)\vee 0\}\bigr)^{3/2}$ and with the convention that $a/0=\infty$ for any $a> 0$.
    Now, we take
    \begin{align*}
    \varepsilon_\alpha :=& \max\Bigl\{\alpha -  \bar\Phi\bigl(\bar\Phi^{-1}(\alpha-\Gamma_{1,*})+\bigl|\bar\Phi^{-1}(\alpha-\Gamma_{1,*})\bigr|\epsilon_*\bigr),\\
    &\hspace{6cm}\bar\Phi\bigl(\bar\Phi^{-1}(\alpha+\Gamma_{1,*})-\bigl|\bar\Phi^{-1}(\alpha+\Gamma_{1,*})\bigr|\epsilon_*\bigr)-\alpha  \Bigr\},
    \end{align*}
    if $\epsilon_*<1$, and if $\epsilon_*\geq 1$, we take $\varepsilon_\alpha=\max(\alpha,1-\alpha)$. Also, here we take the convention that $\bar\Phi^{-1}(t)=\infty$ if $t\leq 0$ and $\bar\Phi^{-1}(t)=-\infty$ if $t\geq 1$. Then, by construction, with this choice of~$\varepsilon_\alpha$, the event $\Omega_\alpha$ defined in~\eqref{eqn:alpha_Delta_event} satisfies $\PPst{\Omega_\alpha^c}{\bY,\bZ} \leq \delta$.
    Combining both cases, we finally have that
    \begin{align*} 
    &\Phi\Biggl(\frac{\bw^\top\Delta\mu_n(\bY,\bZ)}{\sqrt{2}\sigma_n\|\bw\|_2} -  \bar\Phi^{-1}(\alpha - \varepsilon_\alpha) \cdot\sqrt{1 +\frac{ \|\bw\circ\Delta\mu_n(\bY,\bZ)\|^2_2}{2\sigma_n^2\|\bw\|^2_2}}\Biggr) -  1.59\rho_4^3\cdot\frac{\|\bw\|_\infty}{\|\bw\|_2} -  \delta\\
    &\hspace{0.5cm}\leq \PPst{p\leq \alpha}{\bY,\bZ}\\ 
    &\hspace{0.5cm}\leq 
    \Phi\Biggl(\frac{\bw^\top\Delta\mu_n(\bY,\bZ)}{\sqrt{2}\sigma_n\|\bw\|_2} -  \bar\Phi^{-1}(\alpha +\varepsilon_\alpha) \cdot\sqrt{1 +\frac{ \|\bw\circ\Delta\mu_n(\bY,\bZ)\|^2_2}{2\sigma_n^2\|\bw\|^2_2}}\Biggr) + 1.59\rho_4^3\cdot\frac{\|\bw\|_\infty}{\|\bw\|_2} +  \delta.\end{align*}
    Since by construction, $\varepsilon_\alpha\in[0,\max(\alpha,1-\alpha)]$, and by properties of $\bar\Phi$ and $\bar{\Phi}^{-1}$, there exists $C_\alpha > 0$, depending only on $\alpha$, such that
    \[\varepsilon_\alpha \leq C_\alpha(\epsilon_*+\epsilon_*'{}^{3/2}).\]
    Therefore, choosing $C > 0$, depending only on $\alpha$ and $\rho_4$, appropriately,
    \[
    \varepsilon_\alpha \leq C\delta^{-1} \cdot (\epsilon_1 + \epsilon_2),
    \]
    as required.
\end{proof}

\subsubsection{Proof of Theorem~\ref{thm:general-power-asymptotic}}
    We start by noting that by Assumption~\textbf{A1}, $\epsilon_1 = \frac{\|\bw\|_\infty}{\|\bw\|_2}=\mathrm{o}_P(1)$, and by Assumption~\textbf{A2}, 
\[
\epsilon_2=\frac{\|\bw\circ\Delta\mu_n(\bY,\bZ)\|^3_3}{\sigma_n^3\|\bw\|^3_2 + \|\bw\circ\Delta\mu_n(\bY,\bZ)\|^3_2}\leq \frac{\|\bw\circ\Delta\mu_n(\bY,\bZ)\|^3_3}{\sigma_n^3\|\bw\|^3_2 \vee \|\bw\circ\Delta\mu_n(\bY,\bZ)\|^3_2}=\mathrm{o}_P(1).
\]
Moreover, we have that
\begin{align*}
    \frac{\bw^\top\Delta\mu_n(\bY,\bZ)}{\sqrt{2}\sigma_n\|\bw\|_2} &-  \bar\Phi^{-1}(\alpha_{\rm{lo},\delta}) \cdot\sqrt{1 +\frac{ \|\bw\circ\Delta\mu_n(\bY,\bZ)\|^2_2}{2\sigma_n^2\|\bw\|^2_2}} \\
    &\geq \frac{\bw^\top\Delta\mu_n(\bY,\bZ)}{\sqrt{2}\sigma_n\|\bw\|_2} -  \biggl\{\bar\Phi^{-1}(\alpha_{\rm{lo},\delta})\biggl(1 +\frac{ \|\bw\circ\Delta\mu_n(\bY,\bZ)\|_2}{\sqrt{2}\sigma_n\|\bw\|_2}\biggr)\vee  \bar\Phi^{-1}(\alpha_{\rm{lo},\delta})\biggr\}\\
    &\geq \frac{\bw^\top\Delta\mu_n(\bY,\bZ)}{\sqrt{2}\sigma_n\|\bw\|_2}-\bar\Phi^{-1}(\alpha_{\rm{lo},\delta})-\bigl|\bar\Phi^{-1}(\alpha_{\rm{lo},\delta}) \bigr|\cdot \frac{ \|\bw\circ\Delta\mu_n(\bY,\bZ)\|_2}{\sqrt{2}\sigma_n\|\bw\|_2}\\
    &= \frac{\bw^\top\Delta\mu_n(\bY,\bZ)}{\sqrt{2}\sigma_n\|\bw\|_2}\left(1-\left|\bar\Phi^{-1}(\alpha_{\rm{lo},\delta})\right|\frac{ \|\bw\circ\Delta\mu_n(\bY,\bZ)\|_2}{\bw^\top\Delta\mu_n(\bY,\bZ)}\right)-\bar\Phi^{-1}(\alpha_{\rm{lo},\delta}).
\end{align*}
Now, recalling the definition of $\alpha_{\rm{lo},\delta}$ from the statement of Theorem~\ref{thm:general power analysis}, we observe that $\alpha_{\rm{lo},\delta}=\alpha -\mathrm{o}_P(1)$ for any fixed $\delta > 0$. By Assumption~\textbf{A2}, if $\sqrt{2}\sigma_n\|\bw\|_2 > \bw^\top\Delta\mu_n(\bY,\bZ)$, then
\[
\bigl|\bar\Phi^{-1}(\alpha_{\rm{lo},\delta})\bigr| \cdot\frac{ \|\bw\circ\Delta\mu_n(\bY,\bZ)\|_2}{\sqrt{2}\sigma_n\|\bw\|_2} =\mathrm{o}_P(1),
\]
and otherwise, if $2\sigma_n^2\|\bw\|_2\leq\bw^\top\Delta\mu_n(\bY,\bZ)$, then
\[
\bigl|\bar\Phi^{-1}(\alpha_{\rm{lo},\delta})\bigr|\cdot \frac{ \|\bw\circ\Delta\mu_n(\bY,\bZ)\|_2}{\bw^\top\Delta\mu_n(\bY,\bZ)}=\mathrm{o}_P(1).
\]
Combining both cases, since $\Bar{\Phi}^{-1}$ is continuous, we see that 
\begin{align*}
        \Phi\Biggl(\frac{\bw^\top\Delta\mu_n(\bY,\bZ)}{\sqrt{2}\sigma_n\|\bw\|_2} &-  \bar\Phi^{-1}(\alpha_{\rm{lo},\delta}) \cdot\sqrt{1 +\frac{ \|\bw\circ\Delta\mu_n(\bY,\bZ)\|^2_2}{2\sigma_n^2\|\bw\|^2_2}}\Biggr)-(C'\epsilon_1+\delta) \\
        &\geq \Phi\biggl(\frac{\bw^\top\Delta\mu_n(\bY,\bZ)}{\sqrt{2}\sigma_n\|\bw\|_2}\{1-\mathrm{o}_P(1)\} -  \bar\Phi^{-1}\bigl(\alpha \!-\! \mathrm{o}_P(1)\bigr)-\mathrm{o}_P(1)\biggr) - \mathrm{o}_P(1)- \delta\\
         &=\Phi\left(\frac{\bw^\top\Delta\mu_n(\bY,\bZ)}{\sqrt{2}\sigma_n\|\bw\|_2}\{1-\mathrm{o}_P(1)\}-  \bar\Phi^{-1}(\alpha)\right) - \mathrm{o}_P(1)- \delta\\
         &=\Phi\left(\Phi^{-1}(\alpha)+\frac{\bw^\top\Delta\mu_n(\bY,\bZ)}{\sqrt{2}\sigma_n\|\bw\|_2}\right) - \mathrm{o}_P(1)- \delta,
    \end{align*}
    where the last equality follows by Lemma~\ref{lem:gaussian-CDF-approximation}.
    Since $\delta>0$ was arbitrary, the desired asymptotic lower bound follows. The asymptotic upper bound follows by a very similar argument.
\hfill $\square$

\subsection{Hardness of testing $H_0^{\mathrm{ICI}}$ (proof of Theorem~\ref{lem:hardness_result_hellinger})}
By Proposition~\ref{thm: hardness result general}, it is enough to prove that 
    \[
    \inf_{Q_{X,Y,Z}\in H_0^\textnormal{ICI}}\TV\left(P^n_{X,Y,Z},Q^n_{X,Y,Z}\right)\leq L_\zeta\cdot \ISNR_n.
    \]
    Now consider any $\tilde{\mu}_n \in \mathcal{C}_\ISO$. Similar to \eqref{model:general alternative}, define $\tilde{Q}_{X,Y,Z}$ to be the joint distribution of $(X,Y,Z)$ when 
    \[
    X= \tilde{\mu}_n(Z)+\sigma_n\zeta,~~ (Y,Z) \sim P_{Y,Z},~~\zeta \sim P_\zeta.
    \]
    Now abusing notation, we will write $P_{\bX\mid \bY,\bZ}$ to denote the conditional distribution of $\bX$ given $(\bY,\bZ)$, under the joint distribution $(\bX,\bY,\bZ)\sim P^n_{X,Y,Z}$. Similarly, we will write $Q_{\bX\mid \bY,\bZ}$. Then by properties of Hellinger distance, since the marginal distribution of $(Y,Z)$ is identical under $P_{X,Y,Z}$ as under $Q_{X,Y,Z}$ by construction, we have
    \[\TV(P^n_{X,Y,Z},Q^n_{X,Y,Z}) = \EE{\TV(P_{\bX\mid \bY,\bZ},Q_{\bX\mid \bY,\bZ})} \leq \EE{\mathrm{H}(P_{\bX\mid \bY,\bZ},Q_{\bX\mid \bY,\bZ})},\]
    where the last step holds since total variation distance is bounded by Hellinger distance \citep{le2000asymptotics}.    
   Moreover, by construction, the distribution $P_{\bX\mid\bY,\bZ}$ has density
    \[\prod_{i=1}^n \sigma_n^{-1}f_\zeta\left(\frac{x_i - \mu_n(Y_i,Z_i)}{\sigma_n}\right),\]
    while 
    $Q_{\bX\mid\bY,\bZ}$ has density
    \[\prod_{i=1}^n \sigma_n^{-1}f_\zeta\left(\frac{x_i - \tilde\mu_n(Z_i)}{\sigma_n}\right).\]
Since the squared Hellinger distance is subadditive over product distributions \citep[Appendix C.2]{canonne2020survey}, 
\begin{align*}
    \mathrm{H}^2\left(P_{\bX\mid\bY,\bZ},Q_{\bX\mid\bY,\bZ}\right) 
    &\leq \sum_{i=1}^n \int_\R \left[ \sqrt{\sigma_n^{-1}f_\zeta\left(\frac{x_i - \mu_n(Y_i,Z_i)}{\sigma_n}\right)} - \sqrt{\sigma_n^{-1}f_\zeta\left(\frac{x_i - \tilde\mu_n(Z_i)}{\sigma_n}\right)}\right]^2 \!\!\!\mathsf{d}x\\
    &=\sum_{i=1}^n \int_\R \left[ \sqrt{f_\zeta\left(x_i - \frac{\mu_n(Y_i,Z_i)}{\sigma_n}\right)} - \sqrt{f_\zeta\left(x_i - \frac{\tilde\mu_n(Z_i)}{\sigma_n}\right)}\right]^2 \!\!\!\mathsf{d}x\\
    &\leq \sum_{i=1}^n \left(L_\zeta\cdot \left|\frac{\mu_n(Y_i,Z_i)}{\sigma_n} -  \frac{\tilde\mu_n(Z_i)}{\sigma_n}\right|\right)^2\\
    &=\left(L_\zeta \sigma_n^{-1} \|\mu_n(\bY,\bZ) - \tilde\mu_n(\bZ)\|_2\right)^2,
\end{align*}
where the second step holds by a change of variables in integration, and the third step holds by the condition~\eqref{eqn:assume_hellinger}. Returning to our calculations above, then,
\[\TV(P^n_{X,Y,Z},Q^n_{X,Y,Z})  \leq \EE{L_\zeta \sigma_n^{-1} \|\mu_n(\bY,\bZ) - \tilde\mu_n(\bZ)\|_2}.\]
Since we can choose any $\tilde\mu_n\in \mathcal{C}_\ISO$, then,
\[\inf_{Q_{X,Y,Z}\in H_0^{\textnormal{ICI}}}\TV(P^n_{X,Y,Z},Q^n_{X,Y,Z})
\leq \inf_{\tilde\mu_n\in\mathcal{C}_\ISO}\EE{L_\zeta \sigma_n^{-1} \|\mu_n(\bY,\bZ) - \tilde\mu_n(\bZ)\|_2} = L_\zeta\cdot \ISNR_n,\]
which completes the proof.$\hfill{\square}$

\subsection{Power of oracle and plug-in matching strategies (proof of Theorems~\ref{thm:oracle-matching-asymptote} and \ref{thm:oracle-matching-estimated-mu})}\label{app:power-oracle}

In this section we prove Theorems~\ref{thm:oracle-matching-asymptote} and~\ref{thm:oracle-matching-estimated-mu} using the general upper and lower bounds on the conditional power of the \textnormal{\texttt{PairSwap-ICI}} test from Theorem~\ref{thm:general power analysis}.  

 Before turning to the proofs, we recall that by the discussion in Section~\ref{sec:power_with_estimated_mu}, for the oracle matching, we may take $\bw^* = \Delta\mu_n^+(\bY,\bZ)$.  When we do not have access to oracle knowledge of $\mu_n$, we assume that we are able to learn estimates $\hat\mu_n$ of $\mu_n$ from an independent dataset, independent of $(\bX,\bY,\bZ)$. With these estimates in place, the plug-in weights is simply given by $\hat{\bw}=\Delta\hat\mu_n^+(\bY,\bZ)$.

\subsubsection{Proof of Theorem~\ref{thm:oracle-matching-asymptote}}

Given access to oracle knowledge of $\mu_n$, we observe that $\hat{\mu}_n=\mu_n$ satisfies the assumptions of Theorem~\ref{thm:oracle-matching-estimated-mu}. Hence, Theorem~\ref{thm:oracle-matching-asymptote} follows as an immediate corollary of Theorem~\ref{thm:oracle-matching-estimated-mu}. $\hfill \Box$

\subsubsection{Proof of Theorem~\ref{thm:oracle-matching-estimated-mu}}\label{app:power-oracle-estimated}

Our proof is split into three steps. First we state a key property of the oracle matching~$M^\star$ and oracle weight $\bw^\star$. Then, we use this to deduce similar properties for plug-in matching $\hat{M}$ and the plug-in weights $\hat{\bw}$. These in turn enable us to show that Assumptions~\textbf{A1}, \textbf{A2} and~\textbf{A3} of Corollary~\ref{thm:general-power-asymptotic} are satisfied by plug-in matching, and hence the result follows.

Our notation $\Delta\mu_n(\bY,\bZ)$ and $\Delta\hat{\mu}_n(\bY,\bZ)$ suppresses the dependence of the matching $M$ for which it is computed.  Since we will need both matchings $M^\star$ and $\hat{M}$ in this step, we make the dependence on the matching explicit by writing $\Delta\mu_n(\bY,\bZ;M)$ and $\Delta\hat{\mu}_n(\bY,\bZ;M)$, for any matching $M\in \mathcal{M}_n(\bZ)$.
Throughout this proof, we write 
\begin{align*}
    \textnormal{Err}_2(\hat{\mu}_n,\mu_n) &= \frac{\|\Delta\hat\mu_n(\bY,\bZ;\hat{M})-\Delta\mu_n(\bY,\bZ;\hat{M})\|_2}{\|\Delta\mu_n^+(\bY,\bZ;\hat{M})\|_2}, \\ 
    \textnormal{Err}_\infty(\hat{\mu}_n,\mu_n) &= \frac{\|\Delta\hat\mu_n(\bY,\bZ;\hat{M})-\Delta\mu_n(\bY,\bZ;\hat{M})\|_\infty^4}{\sigma_n^3\|\Delta\mu_n^+(\bY,\bZ;\hat{M})\|_2}.
\end{align*}

\paragraph{Step~1: a key property of oracle matching and its implications.}

We start with a key property of oracle matching, which relates $\|\Delta\mu_n^+(\bY,\bZ;M^\star)\|_2$ to the isotonic SNR: given $(\bY,\bZ)$, we have that
\begin{equation}\label{eq:key-property-oracle-matching}
\widehat{\ISNR}_n\leq \frac{\|\Delta\mu_n^+(\bY,\bZ;M^\star)\|_2}{\sigma_n} = \frac{\sup_{M \in \mathcal{M}_n(\bZ)} \|\Delta\mu_n^+(\bY,\bZ;M)\|_2}{\sigma_n} \leq \sqrt{2}\,\widehat{\ISNR}_n.
\end{equation}
The equality is simply due to the definition of $M^*$ as an optimal matching (note that $\mathcal{M}_n(\bZ)$ is a finite set and so the supremum must be attained at some $M^*\in\mathcal{M}_n(\bZ)$), while the two inequalities follow from Theorem~\ref{thm:key-property-oracle}, proved later in Appendix~\ref{app:IMM}.

Further, we require a very similar result for the plug-in matching.
To this end, we note that for any $M \in \mathcal{M}_n(\bZ)$,
\begin{align}\label{eq:relate-hatmu-mu-at-M}
    \bigl|\|\Delta\hat{\mu}_n^+(\bY,\bZ;M)\|_2-\|\Delta\mu_n^+(\bY,\bZ;M)\|_2\bigr|&\leq \|\Delta\hat{\mu}_n^+(\bY,\bZ;M)-\Delta\mu_n^+(\bY,\bZ;M)\|_2\nonumber\\
    &\leq \|\Delta\hat{\mu}_n(\bY,\bZ;M)-\Delta\mu_n(\bY,\bZ;M)\|_2\nonumber\\
    &\leq \sqrt{2}\|\hat\mu_n(\bY,\bZ)-\mu_n(\bY,\bZ)\|_2 = \mathrm{o}_P(\sigma_n\widehat{\ISNR}_n),
\end{align}
where the final step is from the assumptions of the theorem. 
Further, by definition of the plug-in matching in \eqref{def:plug-in-matching} and by \eqref{eq:relate-hatmu-mu-at-M},  we have that 
\begin{align*}
    \|\Delta\hat{\mu}_n^+(\bY,\bZ;\hat{M})\|_2 \geq \|\Delta\hat{\mu}_n^+(\bY,\bZ;M^\star)\|_2 &\geq \|\Delta\mu_n^+(\bY,\bZ;M^\star)\|_2 - \mathrm{o}_P(\widehat{\sigma_n\ISNR}_n) \\
    &\geq \sigma_n\widehat{\ISNR}_n \bigl(1-\mathrm{o}_P(1)\bigr).
\end{align*}
By a similar argument, then, we can establish 
\begin{align}\label{eq:lower-bound-at-Mhat}
    \sigma_n\widehat{\ISNR}_n \bigl(1-\mathrm{o}_P(1)\bigr)&\leq \|\Delta\hat{\mu}_n^+(\bY,\bZ;\hat{M})\|_2\leq \sqrt{2} \sigma_n\widehat{\ISNR}_n \bigl(1+\mathrm{o}_P(1)\bigr),\nonumber \\
     \sigma_n\widehat{\ISNR}_n \bigl(1-\mathrm{o}_P(1)\bigr)&\leq \|\Delta\mu_n^+(\bY,\bZ;\hat{M})\|_2\leq \sqrt{2} \sigma_n\widehat{\ISNR}_n,
\end{align}
i.e., a property similar to \eqref{eq:key-property-oracle-matching} for plug-in matching with both $\mu_n$, and $\hat{\mu}_n$.
As a result, we deduce that  for plug-in matching,
\begin{align}\label{eqn:oracle_assumption_del_mu_level}
         \max\biggl\{\frac{\sigma_n^3\|\mu_n(\bY,\bZ)\|_\infty\vee \|\mu_n(\bY,\bZ)\|_\infty^4}{\sigma_n^3\|\Delta\mu_n^+(\bY,\bZ;\hat{M})\|_2},\,\textnormal{Err}_2(\hat{\mu}_n,\mu_n),\,\textnormal{Err}_\infty(\hat{\mu}_n,\mu_n)\biggr\}=\mathrm{o}_P(1).
     \end{align}
     
\paragraph{Step~2: establishing Assumptions~\textbf{A1} and~\textbf{A2}.}
Henceforth we work with the matching~$\hat{M}$.  We have by~\eqref{eqn:oracle_assumption_del_mu_level} that
\begin{align*}
       \frac{\|\Delta\hat\mu_n^+(\bY,\bZ)\|_\infty}{\|\Delta\hat\mu_n^+(\bY,\bZ)\|_2}& \leq \frac{\|\Delta\mu_n^+(\bY,\bZ)\|_\infty+\|\Delta\hat\mu_n(\bY,\bZ)-\Delta\mu_n(\bY,\bZ)\|_2}{\|\Delta\mu_n^+(\bY,\bZ)\|_2-\|\Delta\hat\mu_n(\bY,\bZ)-\Delta\mu_n(\bY,\bZ)\|_2}\\
        & \leq \frac{2\bigl(\|\mu_n(\bY,\bZ)\|_\infty/\|\Delta\mu_n^+(\bY,\bZ)\|_2\bigr)+ \textnormal{Err}_2(\hat{\mu}_n,\mu_n)}{1-\textnormal{Err}_2(\hat{\mu}_n,\mu_n)} =\mathrm{o}_P(1),
   \end{align*}
   i.e., Assumption~\textbf{A1} holds for plug-in matching with weights $\hat\bw=\Delta\hat{\mu}_n^+(\bY,\bZ)$. Moreover, \eqref{eqn:oracle_assumption_del_mu_level} also yields that
    \begin{align*}
        &\|\Delta\hat\mu_n^+(\bY,\bZ)\circ\Delta\mu_n(\bY,\bZ)\|_2\\
        &\hspace{1.5cm}\leq \|\Delta\mu_n^+(\bY,\bZ)\circ \Delta\mu_n(\bY,\bZ)\|_2+\bigl\|\bigl(\Delta\hat\mu_n^+(\bY,\bZ) \!-\!\Delta\mu_n^+(\bY,\bZ)\bigr)\!\circ\! \Delta\mu_n(\bY,\bZ) \bigr\|_2\\
        &\hspace{1.5cm}\leq \|\Delta\mu_n(\bY,\bZ)\|_\infty\cdot\|\Delta\mu_n^+(\bY,\bZ)\|_2+\|\Delta\mu_n(\bY,\bZ)\|_\infty\cdot  \|\Delta\hat\mu_n(\bY,\bZ)-\Delta\mu_n(\bY,\bZ)\|_2 \\
        &\hspace{1.5cm}\leq 2\|\mu_n(\bY,\bZ)\|_\infty \|\Delta\mu_n^+(\bY,\bZ)\|_2 \cdot\bigl(1+\textnormal{Err}_2(\hat{\mu}_n,\mu_n)\bigr)=\mathrm{o}_P\bigl(\|\Delta\mu_n^+(\bY,\bZ)\|_2^2\bigr). 
    \end{align*}
    Next, by the Cauchy--Schwarz inequality we have that
    \begin{align*}
        \bigl|\Delta\hat\mu_n^+(\bY,\bZ)^\top\Delta\mu_n(\bY,\bZ) &- \|\Delta\hat{\mu}_n^+(\bY,\bZ)\|_2^2\bigr| \\
        &=\bigl|\bigl(\Delta\hat\mu_n(\bY,\bZ)-\Delta\mu_n(\bY,\bZ)\bigr)^\top\Delta\hat{\mu}_n^+(\bY,\bZ)\bigr|\\
        &\leq \|\Delta\hat{\mu}_n^+(\bY,\bZ)\|_2\cdot \|\Delta\mu_n^+(\bY,\bZ)\|_2\cdot \textnormal{Err}_2(\hat{\mu}_n,\mu_n)\\
        &\leq \|\Delta\mu_n^+(\bY,\bZ)\|^2_2\cdot \left(1 + \textnormal{Err}_2(\hat{\mu}_n,\mu_n)\right)\cdot \textnormal{Err}_2(\hat{\mu}_n,\mu_n).
\end{align*}
Hence, by
~\eqref{eqn:oracle_assumption_del_mu_level}, 
\begin{equation}\label{eqn:behaviour_of_key_term_oracle_unknown}
    \bigl|\Delta\hat\mu_n^+(\bY,\bZ)^\top\Delta\mu_n(\bY,\bZ) -\|\Delta\mu_n^+(\bY,\bZ)\|_2^2\bigr|=\mathrm{o}_P\bigl(\|\Delta\mu_n^+(\bY,\bZ)\|_2^2\bigr),
\end{equation}
meaning that
\[\|\Delta\mu_n^+(\bY,\bZ)\|_2^2 = {\rm O}_P(|\Delta\hat\mu_n^+(\bY,\bZ)^\top\Delta\mu_n(\bY,\bZ)|).\]
Combining these calculations verifies that
\[  \|\Delta\hat\mu_n^+(\bY,\bZ)\circ\Delta\mu_n(\bY,\bZ)\|_2= {\rm o}_P(|\Delta\hat\mu_n^+(\bY,\bZ)^\top\Delta\mu_n(\bY,\bZ)|),\]
or in other words,
\[  \|\hat\bw\circ\Delta\mu_n(\bY,\bZ)\|_2= {\rm o}_P(|\hat\bw^\top\Delta\mu_n(\bY,\bZ)|),\]
 and thus the first condition of Assumption~\textbf{A2} is satisfied.
   Next, we focus on the second condition of the same assumption. We start with writing
    \begin{align*}
        &\|\Delta\hat\mu_n^+(\bY,\bZ)\circ\Delta\mu_n(\bY,\bZ)\|_3\\
        &\leq \|\Delta\mu_n^+(\bY,\bZ)\circ \Delta\mu_n(\bY,\bZ)\|_3+\|\bigl(\Delta\hat\mu_n^+(\bY,\bZ)-\Delta\mu_n^+(\bY,\bZ)\bigr)\circ \Delta\mu_n(\bY,\bZ)\|_3\\
        &\leq \|\Delta\mu_n^+(\bY,\bZ)\circ \Delta\mu_n(\bY,\bZ)\|_3+\|\Delta\hat\mu_n^+(\bY,\bZ)-\Delta\mu_n^+(\bY,\bZ)\|_3\|\Delta\mu_n(\bY,\bZ)\|_\infty.
    \end{align*}
   For the first term in the last expression, by \eqref{eqn:oracle_assumption_del_mu_level} we have that
   \begin{multline*}
    \|\Delta\mu_n^+(\bY,\bZ)\circ \Delta\mu_n(\bY,\bZ)\|_3 = \|\Delta\mu_n^+(\bY,\bZ)\|_6^2\leq  \left(\|\Delta\mu_n^+(\bY,\bZ)\|_2^2\|\Delta\mu_n^+(\bY,\bZ)\|_\infty^4\right)^{1/3}\\\leq  \left(\|\Delta\mu_n^+(\bY,\bZ)\|_2^2(2\|\mu_n(\bY,\bZ)\|_\infty)^4\right)^{1/3}=\mathrm{o}_P(\sigma_n\|\Delta\mu_n^+(\bY,\bZ)\|_2).
   \end{multline*}
   For the second term, again by \eqref{eqn:oracle_assumption_del_mu_level},
   \begin{align*}
       &\|\Delta\hat\mu_n^+(\bY,\bZ)-\Delta\mu_n^+(\bY,\bZ)\|_3\|\Delta\mu_n(\bY,\bZ)\|_\infty\\
       &\leq \|\Delta\hat\mu_n^+(\bY,\bZ)-\Delta\mu_n^+(\bY,\bZ)\|_2^{2/3}\cdot \|\Delta\hat\mu_n^+(\bY,\bZ)-\Delta\mu_n^+(\bY,\bZ)\|_\infty^{1/3}\cdot\|\Delta\mu_n(\bY,\bZ)\|_\infty\\
       &= \sigma_n\|\Delta\mu_n^+(\bY,\bZ)\|_2 \left( \frac{\|\bigl(\Delta\hat\mu_n^+(\bY,\bZ)-\Delta\mu_n^+(\bY,\bZ)\bigr)\|_2^{2/3}}{\|\Delta\mu_n^+(\bY,\bZ)\|_2^{2/3}}\times \frac{\|\bigl(\Delta\hat\mu_n^+(\bY,\bZ)-\Delta\mu_n^+(\bY,\bZ)\bigr)\|_\infty^{1/3}}{\sigma_n^{1/4}\|\Delta\mu_n^+(\bY,\bZ)\|_2^{1/12}}\right.\\
       &\hspace{11.2cm}\left.\times\frac{\|\Delta\mu_n(\bY,\bZ)\|_\infty}{\sigma_n^{3/4}\|\Delta\mu_n^+(\bY,\bZ)\|_2^{1/4}}\right)\\
       &= \sigma_n\|\Delta\mu_n^+(\bY,\bZ)\|_2 \biggl(\textnormal{Err}_2(\hat{\mu}_n,\mu_n)^{2/3}\cdot \textnormal{Err}_\infty(\hat{\mu}_n,\mu_n)^{1/12}\cdot\frac{\|\Delta\mu_n(\bY,\bZ)\|_\infty}{\sigma_n^{3/4}\|\Delta\mu_n^+(\bY,\bZ)\|_2^{1/4}} \biggr)\\
       &=\mathrm{o}_P(\sigma_n\|\Delta\mu_n^+(\bY,\bZ)\|_2).
   \end{align*}
   Finally, we also note that by \eqref{eq:relate-hatmu-mu-at-M} and \eqref{eq:lower-bound-at-Mhat},
   \[
   \|\Delta\hat{\mu}_n^+(\bY,\bZ)\|_2\geq \|\Delta\mu_n^+(\bY,\bZ)\|_2-\mathrm{o}_P(\sigma_n\widehat{\ISNR}_n)\geq  \|\Delta\mu_n^+(\bY,\bZ)\|_2(1-\mathrm{o}_P(1)).
   \]
   Hence, combining all these calculations, $\|\Delta\hat\mu_n^+(\bY,\bZ)\circ\Delta\mu_n(\bY,\bZ)\|_3=\mathrm{o}_P(\sigma_n\|\Delta\hat{\mu}_n^+(\bY,\bZ)\|_2)$, or in other words,  $\|\hat\bw\circ\Delta\mu_n(\bY,\bZ)\|_3=\mathrm{o}_P(\sigma_n\|\hat\bw\|_2)$ as required. Therefore, 
    the second condition of Assumption~\textbf{A2} is also satisfied by plug-in matching with weights $\hat\bw=\Delta\hat{\mu}_n^+(\bY,\bZ)$.
    \paragraph{Step~3: applying Corollary~\ref{thm:general-power-asymptotic}.}
    Now that we have verified Assumptions~\textbf{A1} and~\textbf{A2}, we are ready to apply the corollary.
    By \eqref{eq:lower-bound-at-Mhat} and \eqref{eqn:behaviour_of_key_term_oracle_unknown}, 
    \[
        \left|\frac{\hat\bw^\top \Delta\mu_n(\bY,\bZ)}{\|\hat\bw\|_2} - \|\hat\bw\|_2\right| = \biggl|\frac{\Delta\hat\mu_n^+(\bY,\bZ)^\top\Delta\mu_n(\bY,\bZ)}{\|\Delta \hat\mu_n^+(\bY,\bZ)\|_2}-\|\Delta\mu_n^+(\bY,\bZ)\|_2\biggr| = \mathrm{o}_P(\sigma_n\widehat{\ISNR}_n).
    \]
    Then, by \eqref{eq:key-property-oracle-matching} and \eqref{eq:lower-bound-at-Mhat}, 
    \[
    \widehat{\ISNR}_n \bigl(1-\mathrm{o}_P(1)\bigr)\leq\frac{\hat\bw^\top \Delta\mu_n(\bY,\bZ)}{\sigma_n\|\hat\bw\|_2}\leq \sqrt{2}\cdot\widehat{\ISNR}_n \bigl(1+\mathrm{o}_P(1)\bigr).
    \]
    By Corollary~\ref{thm:general-power-asymptotic}, then,
\begin{align*}
        \Phi\biggl(\Phi^{-1}(\alpha)+\frac{\widehat{\ISNR}_n \bigl(1-\mathrm{o}_P(1)\bigr)}{\sqrt{2}}\biggr)-\mathrm{o}_P(1) &\leq \PPst{p\leq \alpha}{\bY,\bZ} \\
        &\leq \Phi\biggl(\Phi^{-1}(\alpha)+\widehat{\ISNR}_n \bigl(1+\mathrm{o}_P(1)\bigr)\biggr)+\mathrm{o}_P(1).~~
    \end{align*}
    Applying Lemma~\ref{lem:gaussian-CDF-approximation} completes the proof.$\hfill{\Box}$

\subsection{Power of neighbor matching (proof of Theorem \ref{thm:neighbour_matching})}\label{app:power-neighbour}
    Let us define $\pi$ to be any permutation of $[n]$ for which
    \[
    Z_{\pi(1)}\leq Z_{\pi(2)}\leq \cdots \leq Z_{\pi(n)},
    \]
    and note that the set of matched pairs for neighbor matching is given by $M:=\bigl\{\bigl(\pi(2\ell-1),\pi(2\ell)\bigr)\bigr\}_{\ell\in \lfloor n/2 \rfloor}$. By Theorem~\ref{thm:general-power-asymptotic} and Lemma~\ref{lem:gaussian-CDF-approximation}, with $\bw=\Delta\bY^+$, 
    it suffices to prove that 
    \begin{equation}
    \label{Eq:DomTermConcentration}
       \biggl|\frac{{\Delta\bY^+}^\top \Delta \mu_n(\bY,\bZ)}{\sigma_n\|\Delta\bY^+\|_2}- \frac{\sqrt{n/2}\cdot\beta_n}{\sigma_n}\cdot \Bigl\{\EE{\VPst{Y\sim P_{Y\mid Z}}{Y}}\Bigr\}^{1/2} \biggr|=\mathrm{o}_P\Bigl(\frac{\beta_n\sqrt{n}}{\sigma_n}\vee 1\Bigr),
    \end{equation}
    and that Assumptions~\textbf{A1} and \textbf{A2} are satisfied by neighbor matching.
    \paragraph{Step~1: establishing~\eqref{Eq:DomTermConcentration}.}
    Under the partially linear Gaussian model~\eqref{model:partial_linear_model},
    \begin{equation}\label{eq:decomposition_delY_delmu}
        {\Delta\bY^+}^\top\Delta\mu_n(\bY,\bZ)={\Delta\bY^+}^\top\Delta\mu_0(\bZ)+\beta_n \|\Delta\bY^+\|_2^2,
    \end{equation}
    where the first term is non-positive since $\mu_0\in\mathcal{C}_\ISO$.  Moreover, we also have $\bigl|{\Delta\bY^+}^\top\Delta\mu_0(\bZ)\bigr| \leq \|\Delta\bY^+\|_\infty\|\Delta\mu_0(\bZ)\|_1$.  Since $\mu_0\in \mathcal{C}_\ISO$,  we also have that
    \begin{align}\label{eq:bound-on-l1-norm-of-delmu0-neighbour}
    \|\Delta\mu_0(\bZ)\|_1=\sum_{\ell=1}^{\lfloor n/2\rfloor} \bigl(\mu_0(Z_{\pi(2\ell)})-\mu_0(Z_{\pi(2\ell-1)})\bigr)&\leq \sum_{i=2}^n \bigl(\mu_0(Z_{\pi(i)})-\mu_0(Z_{\pi(i-1)})\bigr)\nonumber\\
        &=\mu_0(Z_{\pi(n)})-\mu_0(Z_{\pi(1)})\leq 2\|\mu_0(\bZ)\|_\infty.
    \end{align} 
    Now $\|\Delta\bY^+\|_\infty \leq 2$, and since $\mu_0(Z)$ is sub-Gaussian, we have $\|\mu_0(\bZ)\|_\infty = \mathrm{O}_P(\sqrt{\log n})$.
    Additionally, by Proposition \ref{lem:neighbour_matching_conc}, 
     \begin{equation}\label{eq:deltaY-concen-neighbour}
         \Bigl|\frac{2}{n}\|\Delta\bY^+\|_2^2-\EE{\VPst{Y\sim P_{Y\mid Z}}{Y}}\Bigr|=\mathrm{o}_P(1).
     \end{equation}
      Moreover, $\EE{\VPst{Y\sim P_{Y\mid Z}}{Y}} > 0$, so it now follows that
     \begin{align*}
         \biggl|&\frac{{\Delta\bY^+}^\top \Delta \mu_n(\bY,\bZ)}{\sigma_n\|\Delta\bY^+\|_2}-\frac{\sqrt{n/2}\cdot\beta_n}{\sigma_n}\Bigl\{\EE{\VPst{Y\sim P_{Y\mid Z}}{Y}}\Bigr\}^{1/2}\biggr| \\
         &\leq\biggl|\frac{{\Delta\bY^+}^\top \Delta \mu_n(\bY,\bZ)}{\sigma_n\|\Delta\bY^+\|_2}-\frac{\beta_n}{\sigma_n}\|\Delta\bY^+\|_2\biggr|+\biggl|\frac{\beta_n}{\sigma_n}\|\Delta\bY^+\|_2-\frac{\sqrt{n/2}\cdot\beta_n}{\sigma_n}\Bigl\{\EE{\VPst{Y\sim P_{Y\mid Z}}{Y}}\Bigr\}^{1/2}\biggr|\\
         &\leq 4\,\frac{\|\mu_0(\bZ)\|_\infty}{\sigma_n\|\Delta\bY^+\|_2} + \mathrm{o}_P\Bigl(\frac{\beta_n\sqrt{n}}{\sigma_n}\Bigr) = \mathrm{O}_P\left(\frac{\sqrt{\log n}}{\sigma_n\sqrt{n}}\right) + \mathrm{o}_P\Bigl(\frac{\beta_n\sqrt{n}}{\sigma_n}\Bigr)=\mathrm{o}_P\Bigl(\frac{\beta_n\sqrt{n}}{\sigma_n}\vee 1\Bigr), 
     \end{align*}
     where in the penultimate line,
      since $\|\mu_0(\bZ)\|_\infty = \mathrm{O}_P(\sqrt{\log n})$ and $\sigma_n\gg \sqrt{\frac{\log n}{n}}$, by~\eqref{eq:deltaY-concen-neighbour}, we have that $\|\mu_0(\bZ)\|_\infty/(\sigma_n\|\Delta\bY^+\|_2)=\mathrm{O}_P\bigl(\sqrt{\log n}/(\sigma_n\sqrt{n})\bigr) = \mathrm{o}_P(1)$.
     This establishes \eqref{Eq:DomTermConcentration}.
    \paragraph{Step~2: establishing Assumptions~\textbf{A1} and~\textbf{A2}.} 
    Since $Y$ is bounded, and $\|\Delta\bY^+\|_2=\mathrm{O}_P(\sqrt{n})$ by~\eqref{eq:deltaY-concen-neighbour}, Assumption~\textbf{A1} holds for neighbor matching. Next, turning to Assumption~\textbf{A2}, we
    first observe that by~\eqref{eq:bound-on-l1-norm-of-delmu0-neighbour},
    \begin{align}
        \notag\|\Delta\bY^+\circ \Delta\mu_n(\bY,\bZ)\|_2 & \leq \|\Delta\bY^+\circ \Delta\mu_0(\bZ)\|_2+\beta_n \|\Delta\bY^+\|_4^2\\
        \notag&\leq \|\Delta\bY^+\|_\infty \|\Delta\mu_0(\bZ)\|_1 + \beta_n \|\Delta\bY^+\|_\infty \|\Delta\bY^+\|_2\\
        \label{eq:step2_firstbound_neighb}&\leq 4\|\mu_0(\bZ)\|_\infty+2\beta_n \|\Delta\bY^+\|_2.
    \end{align}
    Also, by \eqref{Eq:DomTermConcentration}, it follows that
    \begin{align*}
        \frac{\|\Delta\bY^+\circ \Delta\mu_n(\bY,\bZ)\|_2}{{\Delta\bY^+}^\top\Delta\mu_n(\bY,\bZ)}&\leq \frac{4\|\mu_0(\bZ)\|_\infty+2\beta_n \|\Delta\bY^+\|_2}{\|\Delta\bY^+\|_2\beta_n\sqrt{n/2}\Bigl\{\EE{\VPst{Y\sim P_{Y\mid Z}}{Y}}\Bigr\}^{1/2} \!-\! \mathrm{o}_P(\beta_n\sqrt{n}\|\Delta\bY^+\|_2)}\\
        &=\frac{\frac{4\|\mu_0(\bZ)\|_\infty}{\sigma_n\|\Delta\bY^+\|_2}+\frac{2\beta_n}{\sigma_n} }{\frac{\sqrt{n/2}\cdot\beta_n}{\sigma_n}\Bigl\{\EE{\VPst{Y\sim P_{Y\mid Z}}{Y}}\Bigr\}^{1/2} \!-\! \mathrm{o}_P\Bigl(\frac{\beta_n\sqrt{n}}{\sigma_n}\Bigr)}.
    \end{align*}
    As before, $\|\mu_0(\bZ)\|_\infty/(\sigma_n\|\Delta\bY^+\|_2) = \mathrm{o}_P(1)$. Now  consider first the case where $\beta_n > \sigma_n$.  Recalling that$\mathbb{E}\bigl[\VPst{Y\sim P_{Y\mid Z}}{Y}\bigr] > 0$ (by assumption), we obtain
    \[
    \frac{\|\Delta\bY^+\circ \Delta\mu_n(\bY,\bZ)\|_2}{{\Delta\bY^+}^\top\Delta\mu_n(\bY,\bZ)}=\mathrm{o}_P(1).
    \]
    Otherwise, if $\beta_n \leq \sigma_n$, then by~\eqref{eq:step2_firstbound_neighb}, 
    \[
    \frac{\|\Delta\bY^+\circ \Delta\mu_n(\bY,\bZ)\|_2}{\sigma_n \|\Delta\bY^+\|_2}\leq \frac{4\|\mu_0(\bZ)\|_\infty}{\sigma_n \|\Delta\bY^+\|_2} +\frac{2\beta_n}{\sigma_n}=\mathrm{o}_P(1).
    \]
    Hence, the first condition of Assumption~\textbf{A2} is satisfied under neighbor matching. For the second condition of Assumption~\textbf{A2}, we have immediately in the case $\beta_n \leq \sigma_n$ that 
    \[
    \frac{\|\Delta\bY^+\circ \Delta\mu_n(\bY,\bZ)\|_3}{\sigma_n \|\Delta\bY^+\|_2} \leq \frac{\|\Delta\bY^+\circ \Delta\mu_n(\bY,\bZ)\|_2}{\sigma_n \|\Delta\bY^+\|_2}=\mathrm{o}_P(1). 
    \]
    Finally, in the case $\beta_n > \sigma_n$, we have $\|\mu_0(\bZ)\|_\infty/\|\Delta\bY^+\|_2 =\mathrm{o}_P(\beta_n)$. Moreover, by the same argument as in~\eqref{eq:step2_firstbound_neighb}, and noting that $\|a\|_4\ge d^{-1/4}\,\|a\|_2$ for any $a\in \R^d$ by Cauchy--Schwarz, we obtain 
    \begin{align}
        \notag\|\Delta\bY^+\circ \Delta\mu_n(\bY,\bZ)\|_2 & \geq \beta_n \|\Delta\bY^+\|_4^2-\|\Delta\bY^+\circ \Delta\mu_0(\bZ)\|_2\\
        \notag&\ge \beta_n n^{-1/2} \|\Delta\bY^+\|_2^2-4 \|\Delta\mu_0(\bZ)\|_\infty \\
        &\geq  \beta_n\|\Delta\bY^+\|_2 \bigl( n^{-1/2} \|\Delta\bY^+\|_2-\mathrm{o}_P(1)\bigr).
    \end{align}
    Next, note that 
    \begin{align*}
        \|\Delta\bY^+\circ \Delta\mu_n(\bY,\bZ)\|_\infty\leq \beta_n \|\Delta\bY^+\|_\infty^2+\|\Delta\bY^+\circ \Delta\mu_0(\bZ)\|_\infty 
        &\leq 4\bigl(\beta_n+ \|\mu_0(\bZ)\|_\infty\bigr)\\
        & \leq 4\beta_n\bigl(1+ \|\Delta\bY^+\|_2\,\mathrm{o}_P(1)\bigr),
    \end{align*}
    and also that
    \[
    \|\Delta\bY^+\circ \Delta\mu_n(\bY,\bZ)\|_3\leq \|\Delta\bY^+\circ \Delta\mu_n(\bY,\bZ)\|_\infty^{1/3}\cdot \|\Delta\bY^+\circ \Delta\mu_n(\bY,\bZ)\|_2^{2/3}.
    \]
    Thus, by~\eqref{eq:deltaY-concen-neighbour} and~\eqref{eq:step2_firstbound_neighb}, if $\beta_n > \sigma_n$, then   
   \begin{align*}
    \frac{\|\Delta\bY^+\circ \Delta\mu_n(\bY,\bZ)\|_3}{\|\Delta\bY^+\circ \Delta\mu_n(\bY,\bZ)\|_2}&\le \frac{\|\Delta\bY^+\circ \Delta\mu_n(\bY,\bZ)\|_\infty^{1/3}}{\|\Delta\bY^+\circ \Delta\mu_n(\bY,\bZ)\|_2^{1/3}} \\&\le \frac{4^{1/3}\bigl( 1+\|\Delta\bY^+\|_2\,\mathrm{o}_P(1)\bigr)^{1/3}}{\|\Delta\bY^+\|_2^{1/3} \bigl(n^{-1/2} \|\Delta\bY^+\|_2-\mathrm{o}_P(1)\bigr)^{1/3}} \\
    &\leq \frac{4^{1/3} + 4^{1/3}\|\Delta\bY^+\|_2^{1/3}\,\mathrm{o}_P(1)}{\|\Delta\bY^+\|_2^{1/3}\bigl\{\mathbb{E}\bigl[\mathrm{Var}_{Y \sim P_{Y|Z}}(Y)\bigr]/\sqrt{2} -\mathrm{o}_P(1)\bigr\}^{1/3}}
    = \mathrm{o}_P(1).
    \end{align*}
    Therefore, combining both cases, the second condition of Assumption~\textbf{A2} is also satisfied by neighbor matching. This completes the proof. $\hfill{\square}$

\subsection{Power of cross-bin matching (proof of Theorem~\ref{thm:cross_bin_matching})}\label{app:power-cross-bin}
In this section, we study the power of the \textnormal{\texttt{PairSwap-ICI}} test for cross-bin matching. 
\begin{definition}
    For any probability distribution $P$ on $\R$, we define the deviance of $P$ as 
    \begin{equation}\label{defn:deviance}
    \textnormal{Dev}(P):=\Ep{q\sim \textnormal{Unif}(0,1)} {\bigl(F^{-1}(1-q)-F^{-1}(q)\bigr)^2},
\end{equation}
where $F^{-1}$ is the quantile function of $P$.
\end{definition}
We remark that deviance is a general measure of deviation: in fact, for any distribution, $\textnormal{Dev}(P)\geq 2\,\Vp{X\sim P}{X}$, as we will prove formally in Lemma~\ref{lem:deviance-and-variation}.

With this definition in place, we first state a general concentration result for the conditional power of cross-bin matching in the following theorem. 

\begin{theorem}\label{thm:cross_bin_matching_general}
     Under the same setting and assumptions as the first part of Theorem \ref{thm:cross_bin_matching}, the conditional power of cross-bin matching (\cref{alg:cross-bin-matching}) satisfies
    \[
   \PPst{p\leq\alpha}{\bY,\bZ}=\Phi\biggl(\Phi^{-1}(\alpha)+\frac{\sqrt{n}\beta_n}{2\sigma_n}\cdot\bigl(\EE{\textnormal{Dev}(P_{Y\mid Z})}\bigr)^{1/2}\biggr)+\mathrm{o}_P(1). 
    \]
\end{theorem}
This result will allow us to immediately verify Theorem \ref{thm:cross_bin_matching} (see below), but will also be useful for proving additional results later on.

\subsubsection{Proof of Theorem \ref{thm:cross_bin_matching}} \label{app:power-cross-bin-mainpaper}
By Lemma~\ref{lem:deviance-and-variation}, for any distribution $P$ on $\R$, it holds that $\textnormal{Dev}(P) \geq 2\VPst{X\sim P}{X}$. Therefore, $\bigl(\EE{\textnormal{Dev}(P_{Y\mid Z})}\bigr)^{1/2}\geq \sqrt{2}\,\bigl(\mathbb{E}\bigl[\VPst{Y\sim P_{Y\mid Z}}{Y}\bigr]\bigr)^{1/2}$, and so the first part of Theorem~\ref{thm:cross_bin_matching} follows as an immediate corollary of Theorem~\ref{thm:cross_bin_matching_general}. 

For the second part, we note that by Lemma~\ref{lem:deviance-and-variation}, for any distribution $P$ on $\R$ that is symmetric around its mean, it holds that $\textnormal{Dev}(P) = 4\VPst{X\sim P}{X}$. Applying this calculation to $P_{Y|Z}$, then, we have $\EE{\textnormal{Dev}(P_{Y|Z})} = 4\mathbb{E}\bigl[\VPst{Y\sim P_{Y\mid Z}}{Y}\bigr]$. Thus, the second part follows from Theorem~\ref{thm:cross_bin_matching_general}.
$\hfill{\square}$

\subsubsection{Proof of Theorem~\ref{thm:cross_bin_matching_general}}
    The proof for cross-bin matching follows the same key steps as in the proof for neighbor matching in Appendix~\ref{app:power-neighbour}. 
    We first show that
   \begin{equation}\label{eq:dom-term-conc-cross-bin}
       \biggl|\frac{{\Delta\bY^+}^\top \Delta \mu_n(\bY,\bZ)}{\sigma_n\|\Delta\bY^+\|_2}-\frac{\sqrt{n/2}\cdot\beta_n}{\sigma_n}\cdot \bigl\{\EE{\textnormal{Dev}(P_{Y\mid Z})}\bigr\}^{1/2} \biggr|=\mathrm{o}_P\Bigl(\frac{\beta_n\sqrt{n}}{\sigma_n}\vee 1\Bigr)
   \end{equation}
    and then establish that Assumptions~\textbf{A1} and~\textbf{A2} are satisfied by cross-bin matching.  Finally, we will apply Corollary~\ref{thm:general-power-asymptotic} with $\bw=\Delta \bY^+$ to conclude the proof.
    
    \paragraph{Step~1: establishing~\eqref{eq:dom-term-conc-cross-bin}.}
    Recall the decomposition~\eqref{eq:decomposition_delY_delmu}.  Since $\mu_0\in\mathcal{C}_\ISO$, the first term is non-positive, and we also have that $\bigl|{\Delta\bY^+}^\top\Delta\mu_0(\bZ)\bigr| \leq \|\Delta\bY^+\|_\infty\|\Delta\mu_0(\bZ)\|_1$. Next, recalling the notation from \cref{alg:cross-bin-matching}, and using the fact that $\mu_0\in\mathcal{C}_\ISO$, we have that 
   \begin{align}\label{eq:1-by-infty-crossbin}
       \|\Delta\mu_0(\bZ)\|_1&\leq \sum_{k=1}^{K_n-1}\sum_{s=1}^{\lfloor m/2\rfloor} \bigl(\mu_0(Z_{r_{k+1,s}})-\mu_0(Z_{r_{k,m+1-s}})\bigr)\nonumber\\
       &\leq \bigg\lfloor \frac{m}{2} \bigg\rfloor\sum_{k=1}^{K_n-1} \bigl(\mu_0(Z_{\pi((k+1)m)})-\mu_0(Z_{\pi((k-1)m+1)})\bigr)\nonumber\\
       &\leq \frac{n}{K_n} \bigl(\mu_0(Z_{\pi(n)})-\mu_0(Z_{\pi(1)})\bigr)
       \leq \frac{2n}{K_n} \|\mu_0(\bZ)\|_\infty.
   \end{align}
   Since $\mu_0(Z)$ is sub-Gaussian, $\|\mu_0(\bZ)\|_\infty=\mathrm{O}_P(\sqrt{\log n})$. Also, since $Y$ is bounded and $K_n\gg \sqrt{n\log n}/\sigma_n$, we have that 
   \[
   \bigl|{\Delta\bY^+}^\top\Delta\mu_0(\bZ)\bigr|=\mathrm{o}_P(\sigma_n\sqrt{n}).
   \]
    Moreover by Proposition \ref{lem:cross_bin_lemma},  
    \begin{equation}\label{eqn:concentration of cbin matching term}
        \frac{2}{n}\|\Delta\bY^+\|_2^2 \overset{P}{\to} \EE{\textnormal{Dev}(P_{Y\mid Z})}.
     \end{equation}
     Therefore, we finally have that
     \begin{align*}
         \biggl|&\frac{{\Delta\bY^+}^\top \Delta \mu_n(\bY,\bZ)}{\sigma_n\|\Delta\bY^+\|_2}-\frac{\sqrt{n/2}\cdot\beta_n}{\sigma_n}\cdot \bigl\{\EE{\textnormal{Dev}(P_{Y\mid Z})}\bigr\}^{1/2}\biggr|\\
         &\leq\biggl|\frac{{\Delta\bY^+}^\top \Delta \mu_n(\bY,\bZ)}{\sigma_n\|\Delta\bY^+\|_2}-\frac{\beta_n}{\sigma_n}\|\Delta\bY^+\|_2\biggr|+\biggl|\frac{\beta_n}{\sigma_n}\|\Delta\bY^+\|_2-\frac{\sqrt{n/2}\cdot\beta_n}{\sigma_n}\cdot \bigl\{\EE{\textnormal{Dev}(P_{Y\mid Z})}\bigr\}^{1/2}\biggr|\\
        &=\biggl|\frac{{\Delta\bY^+}^\top \Delta \mu_0(\bZ)}{\sigma_n\|\Delta\bY^+\|_2}\biggr|+\biggl|\frac{\beta_n}{\sigma_n}\|\Delta\bY^+\|_2-\frac{\sqrt{n/2}\cdot\beta_n}{\sigma_n}\cdot \bigl\{\EE{\textnormal{Dev}(P_{Y\mid Z})}\bigr\}^{1/2}\biggr|\\
         &= \frac{{\rm o}_P(\sigma_n\sqrt{n})}{\sigma_n\|\Delta\bY^+\|_2} + \mathrm{o}_P\Bigl(\frac{\beta_n\sqrt{n}}{\sigma_n}\Bigr)=\mathrm{o}_P\Bigl(\frac{\beta_n\sqrt{n}}{\sigma_n}\vee 1\Bigr).
     \end{align*}
     This establishes \eqref{eq:dom-term-conc-cross-bin}.
    \paragraph{Step~2: establishing Assumptions~\textbf{A1} and \textbf{A2}.}
    Since $Y$ is bounded, we have by~\eqref{eqn:concentration of cbin matching term} that Assumption~\textbf{A1} holds for the weights $\bw = \Delta\bY^+$ of cross-bin matching. Next, turning to Assumption~\textbf{A2}, we
    first observe that by \eqref{eq:1-by-infty-crossbin},
    \begin{align}\label{eq:cross_bin_DelY_Delmu_0_bd}
        \|\Delta\bY^+\circ \Delta\mu_n(\bY,\bZ)\|_2&\leq \|\Delta\bY^+\circ \Delta\mu_0(\bZ)\|_2+\beta_n\|\Delta\bY^+\circ \Delta\bY\|_2\notag\\
        &\leq \|\Delta\bY\|_\infty\bigl(\|\Delta\mu_0(\bZ)\|_2 +\beta_n\|\Delta\bY^+\|_2\bigr)\notag\\
        &\leq 2\bigl(\|\Delta\mu_0(\bZ)\|_1^{1/2}\|\Delta\mu_0(\bZ)\|_\infty^{1/2}+\beta_n\|\Delta\bY^+\|_2\bigr)\notag\\
        &\leq 2\biggl(\Bigl(\frac{4n}{K_n}\Bigr)^{1/2}\|\mu_0(\bZ)\|_\infty+\beta_n\|\Delta\bY^+\|_2\biggr)\notag\\
        &=\mathrm{O}_P\biggl(\sqrt{\frac{n \log n}{K_n}} + \beta_n\sqrt{n}\biggr) = \mathrm{o}_P(\sqrt{\sigma_n}\sqrt[4]{n\log n}) + \mathrm{O}_P(\beta_n\sqrt{n}).
    \end{align}
    Hence, in combination with \eqref{eq:dom-term-conc-cross-bin}, it follows that if $\liminf_{n\to\infty} \beta_n/\sigma_n>0$, then
    \begin{align*}
        \frac{\|\Delta\bY^+\circ \Delta\mu_n(\bY,\bZ)\|_2}{{\Delta\bY^+}^\top\Delta\mu_n(\bY,\bZ)}&\leq \frac{\mathrm{o}_P(\sigma_n^{-1/2}(n\log n)^{1/4})+\mathrm{O}_P\bigl( \frac{\beta_n\sqrt{n}}{\sigma_n}\bigr)}{\frac{\beta_n\sqrt{n/2}}{\sigma_n}\|\Delta\bY^+\|_2\bigl\{\EE{\textnormal{Dev}(P_{Y\mid Z})}\bigr\}^{1/2} - \|\Delta\bY^+\|_2\mathrm{o}_P\bigl( \frac{\beta_n\sqrt{n}}{\sigma_n}\bigr)} \\
        &= \frac{\mathrm{o}_P\bigl(\sigma_n^{-1/2}(n\log n)^{1/4}\bigr)+\mathrm{O}_P\bigl( \frac{\beta_n\sqrt{n}}{\sigma_n}\bigr)}{\mathrm{O}_P\bigl( \frac{\beta_n\,n}{\sigma_n}\bigr)} \\
        &= \mathrm{o}_P(\sigma_n^{-1/2}n^{-3/4} \log^{1/4} n) + \mathrm{O}_P(n^{-1/2}) = \mathrm{o}_P(1),
    \end{align*}
    since $\sigma_n\gg \sqrt{\frac{\log n}{n}}$, and $\EE{\textnormal{Dev}(P_{Y\mid Z})} > 0$. Otherwise, if $\beta_n\ll\sigma_n$, then
    \[
    \frac{\|\Delta\bY^+\circ \Delta\mu_n(\bY,\bZ)\|_2}{\sigma_n\|\Delta\bY^+\|_2}=\mathrm{o}_P\bigl(\sigma_n^{-1/2}n^{-1/4} \log^{1/4} n\bigr)+\mathrm{O}_P(\beta_n/\sigma_n)=\mathrm{o}_P(1).
    \]
    Thus, the first condition of Assumption~\textbf{A2} is satisfied by cross-bin matching. Moreover, if $\beta_n\ll\sigma_n$, then
    \[
    \frac{\|\Delta\bY^+\circ \Delta\mu_n(\bY,\bZ)\|_3}{\sigma_n\|\Delta\bY^+\|_2} \leq \frac{\|\Delta\bY^+\circ \Delta\mu_n(\bY,\bZ)\|_2}{\sigma_n\|\Delta\bY^+\|_2} =\mathrm{o}_P(1).
    \]
    Next, suppose $\liminf_{n\to\infty} \beta_n/\sigma_n>0$. We note that 
    \begin{align*}
        \|\Delta\bY^+\circ \Delta\mu_n(\bY,\bZ)\|_\infty&\leq \|\Delta\bY^+\circ \Delta\mu_0(\bZ)\|_\infty +\beta_n \|\Delta\bY^+\|_\infty^2\\
        &\leq 4\bigl( \|\mu_0(\bZ)\|_\infty+\beta_n\bigr) = \mathrm{O}_P(\sqrt{\log n}+\beta_n),
    \end{align*}
    and that 
    \[
    \|\Delta\bY^+\circ \Delta\mu_n(\bY,\bZ)\|_3 
    \leq \|\Delta\bY^+\circ \Delta\mu_n(\bY,\bZ)\|_\infty^{1/3}\,\|\Delta\bY^+\circ \Delta\mu_n(\bY,\bZ)\|_2^{2/3}.
    \]
    Since $\sigma_n\gg \sqrt{\frac{\log n}{n}}$, by the same argument as in~\eqref{eq:cross_bin_DelY_Delmu_0_bd}, we further have that
    \begin{align*}
        \|\Delta\bY^+\circ \Delta\mu_n(\bY,\bZ)\|_2&\geq \beta_n\|\Delta\bY^+\|_4^2-\|\Delta\bY^+\circ \Delta\mu_0(\bZ)\|_2\\
        &\ge \beta_n n^{-1/2} \|\Delta\bY^+\|_2^2-
        \mathrm{o}_P(\sqrt{\sigma_n}\sqrt[4]{n\log n})\\
        &\ge \beta_n n^{-1/2} \|\Delta\bY^+\|_2^2-
        \mathrm{o}_P(\sigma_n n^{1/2})\\
        &\ge \beta_n n^{-1/2} \|\Delta\bY^+\|_2^2\{1 - \mathrm{o}_P(1)\}.
    \end{align*}
    Hence, it holds that 
    \begin{align*}
    \frac{\|\Delta\bY^+\circ \Delta\mu_n(\bY,\bZ)\|_3}{\|\Delta\bY^+\circ \Delta\mu_n(\bY,\bZ)\|_2}&\le \frac{\|\Delta\bY^+\circ \Delta\mu_n(\bY,\bZ)\|_\infty^{1/3}}{\|\Delta\bY^+\circ \Delta\mu_n(\bY,\bZ)\|_2^{1/3}} \\&\le \frac{ \mathrm{O}_P\bigl(\log^{1/6} n +\beta_n^{1/3}\bigr)}{\beta_n^{1/3} n^{-1/6} \|\Delta\bY^+\|_2^{2/3}\{1 - \mathrm{o}_P(1)\}^{1/3}}
    = \mathrm{o}_P(1).
    \end{align*}
    Therefore, the second condition of Assumption~\textbf{A2} is also satisfied by cross-bin matching. 
    
    \paragraph{Step~3: completing the proof}

    By Corollary~\ref{thm:general-power-asymptotic},~\eqref{eq:dom-term-conc-cross-bin}, and Lemma~\ref{lem:gaussian-CDF-approximation}, we have
    \begin{align*}
&\biggl|\PPst{p\leq\alpha}{\bY,\bZ}-\Phi\biggl(\frac{\sqrt{n}\beta_n}{2\sigma_n}\cdot\bigl(\EE{\textnormal{Dev}(P_{Y\mid Z})}\bigr)^{1/2}-\bar\Phi^{-1}(\alpha)\biggr)\biggr| \\&\leq \biggl|\PPst{p\leq\alpha}{\bY,\bZ}- \Phi\biggl(\frac{{\Delta\bY^+}^\top \Delta \mu_n(\bY,\bZ)}{\sqrt{2}\sigma_n\|\Delta\bY^+\|_2}-\bar\Phi^{-1}(\alpha)\biggr)\biggr|
\\
&\hspace{0.4cm}+ \biggl|\Phi\biggl(\frac{{\Delta\bY^+}^\top \Delta \mu_n(\bY,\bZ)}{\sqrt{2}\sigma_n\|\Delta\bY^+\|_2}-\bar\Phi^{-1}(\alpha)\biggr) - \Phi\biggl(\frac{\sqrt{n}\beta_n}{2\sigma_n}\cdot\bigl(\EE{\textnormal{Dev}(P_{Y\mid Z})}\bigr)^{1/2}-\bar\Phi^{-1}(\alpha)\biggr)\biggr| \\
&= \mathrm{o}_P(1),
    \end{align*}
 as required. $\hfill{\square}$

\subsection{Proof of propositions, corollaries and lemmas from Section~\ref{sec:power}}\label{app:lemma-from-section-power}
\subsubsection{Proof of Proposition~\ref{thm: hardness result general}}
Fix a distribution $Q_{X,Y,Z}\in H_0^\textnormal{ICI}$. For any test function 
\[
\phi: (\Xcal\times\Ycal\times\Zcal)^n \to [0,1]\textnormal{ such that }\sup_{P\in H_0^{\textnormal{ICI}}}\Ep{P}{\phi(\bX,\bY,\bZ)} \leq \alpha,
\]
we have that
\begin{align*}
    \Ep{P_{X,Y,Z}}{\phi(\bX,\bY,\bZ)} &\leq \Ep{Q_{X,Y,Z}}{\phi(\bX,\bY,\bZ)} + \TV\left(P^n_{X,Y,Z},Q^n_{X,Y,Z}\right)\\
    & \leq \alpha + \TV\left(P^n_{X,Y,Z},Q^n_{X,Y,Z}\right).
\end{align*}
The result now follows since the last inequality holds for any $Q_{X,Y,Z}\in H_0^\textnormal{ICI}$.

\subsubsection{Proof of Lemma~\ref{lemma:beta-phase-transition}}
    Under the model class~\eqref{model:partial_linear_model}, since $\EE{Y}=0$, we have that
    \begin{align*}
        \sigma_n\ISNR_n 
        &\overset{(1)}{\leq} \Ep{P_{Y,Z}}{\|\mu_n(\bY,\bZ)-\mu_0(\bZ)\|_2} =\Ep{P_{Y,Z}}{\|\beta_n \bY\|_2}\\
        &\overset{(2)}{\leq} \beta_n\biggl(\mathbb{E}_{P_{Y,Z}}\biggl[\sum_{i=1}^n Y_i^2\biggr]\biggr)^{1/2}=\sqrt{n}\beta_n \bigl(\VV{Y}\bigr)^{1/2},
    \end{align*}
    where the inequality $(1)$ holds since $\mu_0\in \mathcal{C}_\ISO$, and  $(2)$ holds by Jensen's inequality.
    Next, again under the model class \eqref{model:partial_linear_model}, writing $\mathcal{G}$ for the set of Borel measurable functions from $\R$ to $\R$,
      \begin{align*}
\sigma_n\ISNR_n&=
\inf_{g\in\mathcal{C}_\ISO}\Ep{P_{Y,Z}}{\|\mu_n(\bY,\bZ) - g(\bZ)\|_2} \geq\inf_{g\in \mathcal{G}}\Ep{P_{Y,Z}}{\|\mu_n(\bY,\bZ) - g(\bZ)\|_2}\\
&=\Ep{P_{Y,Z}}{\bigl\|\mu_n(\bY,\bZ)-\EEst{\mu_n(\bY,\bZ)}{ \bZ}\bigr\|_2} = \Ep{P_{Y,Z}}{\bigl\|\beta_n~\bigl(\bY-\EEst{\bY}{\bZ}\bigr)\bigr\|_2}.
    \end{align*}
Finally, since $Y$ is a bounded random variable, we have by the weak law of large numbers that
\[
\frac{1}{n}\|\bY-\EEst{\bY}{\bZ}\bigr\|_2^2 =\frac{1}{n}\sum_{i=1}^n (Y_i-E[Y\mid Z_i])^2 = \mathbb{E}\bigl[\VPst{Y\sim P_{Y\mid Z}}{Y}\bigr]\,\bigl(1+\mathrm{o}_P(1)\bigr).
\]
Thus, $\sigma_n\ISNR_n\geq \sqrt{n}\beta_n~\bigl(\mathbb{E}\bigl[\VPst{Y\sim P_{Y\mid Z}}{Y}\bigr]\bigr)^{1/2} \,(1+\mathrm{o}_P(1))$. $\hfill{\square}$

\subsection{An oracle matching: isotonic median matching}\label{app:IMM}
A key step in Appendix~\ref{app:power-oracle-estimated} is establishing the inequalities in~\eqref{eq:key-property-oracle-matching} for oracle matching, which clarify its relationship to the isotonic SNR, $\widehat\ISNR_n$. In this section, we will prove these inequalities, restated here in the following theorem:
\begin{theorem}\label{thm:key-property-oracle}
For any measurable spaces $\mathcal{Y}$, $\mathcal{Z}$ with $\mathcal{Z}$ having a partial order, 
\[
\frac{\sup_{M\in \mathcal{M}_n(\bZ)}\|\Delta\mu_n^+(\bY,\bZ;M)\|_2}{\sigma_n} \leq \sqrt{2}\,\widehat{\ISNR}_n.
\]
Furthermore, when $\Zcal=\R$, it holds that
\[
\widehat{\ISNR}_n\leq \frac{\sup_{M\in \mathcal{M}_n(\bZ)}\|\Delta\mu_n^+(\bY,\bZ;M)\|_2}{\sigma_n} \leq \sqrt{2}\,\widehat{\ISNR}_n.
\]
\end{theorem}

In order to prove the second part of this result, we will give an explicit construction for an optimal matching, via a technique
that we refer to as isotonic median matching (IMM). In the remainder of this section, therefore we work with $\Zcal=\R$, and we will first develop the definition and properties of the IMM, and will then return to the proof of the theorem.

From this point on, assume that $Z_1 \leq \cdots\leq Z_n$ without loss of generality.  
Given $a=(a_1,\ldots,a_m) \in \RR^m$, we define the median of this vector by 
\[
\textnormal{Med}(a)=\begin{cases}
        a_{\left(k+1\right)} & \textnormal{if $m=2k+1$ for some $k \in \mathbb{N}_0$},\\
        \frac{a_{(k)}+a_{(k+1)}}{2} & \textnormal{if $m=2k$ for some $k \in \mathbb{N}$,}
    \end{cases}
    \]
where $a_{(1)} \leq \cdots \leq a_{(m)} $ denote the order statistics of $a_1,\ldots,a_m$. Given $\bZ$, we define the cone of vectors in $\RR^n$ that obey isotonic relation with $\bZ$ as 
\[
\widehat{\mathcal{C}}_\ISO(\bZ):=\{v\in \RR^n: \textnormal{for all } 1\leq i\leq j\leq n,\, (z_i-z_j)(v_i-v_j)\geq 0\}.
\]
Moreover, given $(\bY,\bZ)$ and any vector $u\in \RR^n$, we define $\tilde{u}\in \RR^n$ with $i$th component given by
\begin{equation}\label{eqn:empirical_isotonic_L1_projection}
    \tilde{u}_i := \max_{j\in [n]: Z_j\leq Z_i}\min_{k\in [n]: Z_k\geq Z_i} \textnormal{Med}\bigl((u_\ell)_{\ell:Z_j\leq Z_\ell\leq Z_k}\bigr).
\end{equation}
We claim that $\tilde{u}\in \widehat{\mathcal{C}}_\ISO(\bZ)$---the intuition is that we can think of $\tilde{u}$ as a projection of $u$ onto the cone $\widehat{\mathcal{C}}_\ISO(\bZ)$ (in fact, it is the projection with respect to the $\ell_1$ norm, although we will not establish this claim formally as it is not needed here).
To verify that $\tilde{u}\in \widehat{\mathcal{C}}_\ISO(\bZ)$, observe that for $1\leq i_1\leq i_2\leq n$, 
\begin{align*}
    \tilde{u}_{i_1}&=\max_{j\in [n]: Z_j\leq Z_{i_1}}\min_{k\in [n]: Z_k\geq Z_{i_1}} \textnormal{Med}\bigl((u_\ell)_{\ell:Z_j\leq Z_\ell\leq Z_k}\bigr)\\
    &\leq \max_{j\in [n]: Z_j\leq Z_{i_1}}\min_{k\in [n]: Z_k\geq Z_{i_2}} \textnormal{Med}\bigl((u_\ell)_{\ell:Z_j\leq Z_\ell\leq Z_k}\bigr)\\
    &\leq  \max_{j\in [n]: Z_j\leq Z_{i_2}}\min_{k\in [n]: Z_k\geq Z_{i_2}} \textnormal{Med}\bigl((u_\ell)_{\ell:Z_j\leq Z_\ell\leq Z_k}\bigr)=\tilde{u}_{i_2},
\end{align*}
which proves the claim. 
Next, since $\tilde{u}\in \widehat{\mathcal{C}}_\ISO(\bZ)$, there exist $m\in [n+1]$, integers $1=n_1<\ldots <n_m=n+1$ and real numbers $r_1<\ldots<r_{m-1}$ such that
 \begin{equation}\label{Eq:utildeiso}
     \tilde{u}_i = \sum_{t=1}^{m-1}  r_i\,\One{n_t\leq i<n_{t+1}},
 \end{equation}
 i.e., $\tilde{u}$ is a piecewise constant vector, with $m-1$ components. 
 
 The following lemma establishes some additional properties of this vector.
\begin{lemma}\label{lem:characterizing-isotonic-L1-projection_}
In the representation~\eqref{Eq:utildeiso} of the vector $\tilde{u}$, for each $t\in[m-1]$,     \begin{equation}\label{eq:cumulative-median-equals}
    r_t = \textnormal{Med}\bigl((u_\ell)_{n_t \leq \ell \leq n_{t+1}-1}\bigr) ,
    \end{equation}
    and moreover, for each $i \in \{n_t,n_t+1,\ldots,n_{t+1}-1\}$,
\begin{equation}\label{eq:cumulative-median-inequalities}
    \textnormal{Med}\bigl((u_\ell)_{\ell:Z_{n_t}\leq Z_\ell\leq Z_i}\bigr) 
        \geq r_t \geq  \textnormal{Med}\bigl((u_\ell)_{\ell:Z_i\leq Z_\ell\leq Z_{n_{t+1}-1}}\bigr).    
    \end{equation}
\end{lemma}
\begin{proof}
    First, since $\tilde{u}\in\widehat{\mathcal{C}}_\ISO(\bZ)$, and $\tilde{u}_{n_t-1}<\tilde{u}_{n_t}$ for each $t=2,\dots,m-1$, we have 
    $Z_{n_t-1} < Z_{n_t}$ for each $t=2,\dots,m-1$ and consequently
    \[n_t \leq \ell\leq n_{t+1}-1 \ \Longleftrightarrow \ Z_{n_t}\leq Z_\ell \leq Z_{n_{t+1}-1}\]
    for each $t\in[m-1]$.
    Therefore, establishing~\eqref{eq:cumulative-median-equals} is equivalent to proving that
    \[r_t = \textnormal{Med}\bigl((u_\ell)_{\ell:Z_{n_t} \leq Z_\ell \leq Z_{n_{t+1}-1}}\bigr),\]
    which is an immediate consequence of the inequalities in~\eqref{eq:cumulative-median-inequalities} (applied with $i=n_{t+1}-1$ on the left and with $i=n_t$ on the right). Therefore we now only need to prove these inequalities.

    \paragraph{Upper bound on $r_t$.}
    First we will prove that
    \begin{equation}\label{eq:cumulative-median-inequalities-1}
    r_t\leq  \textnormal{Med}\bigl((u_\ell)_{\ell:Z_{n_t}\leq Z_\ell\leq Z_i}\bigr)
\end{equation}
for all $t\in [m-1]$ and all $i\in\{n_t,\dots,n_{t+1}-1\}$.
By definition, we have 
\[r_t = \tilde{u}_{n_t} = \max_{j\in [n]: Z_j\leq Z_{n_t}}\min_{k\in [n]: Z_k\geq Z_{n_t}} \textnormal{Med}\bigl((u_\ell)_{\ell:Z_j\leq Z_\ell\leq Z_k}\bigr).\]
As noted above, $Z_j < Z_{n_t}$ for any $j<n_t$. 
Therefore we can write
\[
\{j\in[n]: Z_j \leq Z_{n_t}\} = \{j\in[n]: j< n_t\} \cup \{j\in[n]: Z_j = Z_{n_t}\},
\]
and hence,
\[
r_t = \max\Big\{\underbrace{\max_{j < n_t}\min_{k\in [n]: Z_k\geq Z_{n_t}} \textnormal{Med}\bigl((u_\ell)_{\ell:Z_j\leq Z_\ell\leq Z_k}\bigr)}_{=:r_{t,1}} , \ \underbrace{\min_{k\in [n]: Z_k\geq Z_{n_t}} \textnormal{Med}\bigl((u_\ell)_{\ell:Z_{n_t}\leq Z_\ell\leq Z_k}\bigr)}_{=:r_{t,2}}\Big\},
\]
where we take the convention that the maximum over an empty set is equal to $-\infty$ (to handle the case $t=1$, i.e., $\{j\in[n]:j<n_1\}=\emptyset$, so that $r_{t,1}=-\infty$ and $r_t = r_{t,2}$. 
If instead $t\geq 2$, then we claim that $r_{t,1}\leq r_{t,2}$, so that $r_t = r_{t,2}$ again. Suppose for a contradiction that $r_{t,1}>r_{t,2}$. Let $\hat{j}\in [n_{t}-1]$ and $\hat{k}\in [n]$ be such that $Z_{\hat{k}} \geq Z_{n_t}$ and
\begin{equation}\label{eq:expression-rt1-rt2}
    r_{t,1}=\min_{k\in [n]: Z_k\geq Z_{n_t}} \textnormal{Med}\bigl((u_\ell)_{\ell:Z_{\hat{j}}\leq Z_\ell\leq Z_k}\bigr)\, \text{ and }\,r_{t,2}=\textnormal{Med}\bigl((u_\ell)_{\ell:Z_{n_t}\leq Z_\ell\leq Z_{\hat{k}}}\bigr).
\end{equation}
Then, $\textnormal{Med}\bigl((u_\ell)_{\ell:Z_{\hat{j}}\leq Z_\ell\leq Z_{\hat{k}}}\bigr)\geq r_{t,1}>\textnormal{Med}\bigl((u_\ell)_{\ell:Z_{n_t}\leq Z_\ell\leq Z_{\hat{k}}}\bigr)$.
Since $\{j: j<n_t\} = \{j: Z_j\leq Z_{n_t-1}\}$ (due to $Z$ being a nondecreasing vector with $Z_{n_t-1}<Z_{n_t}$ as established above), it holds that $Z_{\hat{j}} \leq Z_{n_t-1} <Z_{n_t}$ and thus, by Lemma~\ref{lem:a-median-result},
\[
\textnormal{Med}\bigl((u_\ell)_{\ell:Z_{\hat{j}}\leq Z_\ell\leq Z_{\hat{k}}}\bigr)\leq \textnormal{Med}\bigl((u_\ell)_{\ell:Z_{\hat{j}}\leq Z_\ell\leq Z_{n_t-1}}\bigr).
\]
Now, we also have
\begin{equation}\label{eq:bounding-rt1-below}
    r_{t,1}=r_t>\tilde{u}_{n_t-1}\geq \min_{k\in [n]: Z_k\geq Z_{n_t-1}} \textnormal{Med}\bigl((u_\ell)_{\ell:Z_{\hat{j}}\leq Z_\ell\leq Z_k}\bigr)
\end{equation}
Moreover, by noting that $\{k : Z_k \geq Z_{n_t-1}\}=\{k : Z_k \geq Z_{n_t}\}\cup \{k:Z_k=Z_{n_t-1}\}$ and by~\eqref{eq:expression-rt1-rt2}, it follows that 
\begin{align*}
    &\min_{k\in [n]: Z_k\geq Z_{n_t-1}} \textnormal{Med}\bigl((u_\ell)_{\ell:Z_{\hat{j}}\leq Z_\ell\leq Z_k}\bigr)\\&\hspace{1.6cm}=\min\Bigl\{\min_{k\in [n]: Z_k\geq Z_{n_t}} \textnormal{Med}\bigl((u_\ell)_{\ell:Z_{\hat{j}}\leq Z_\ell\leq Z_k}\bigr), \textnormal{Med}\bigl((u_\ell)_{\ell:Z_{\hat{j}}\leq Z_\ell\leq Z_{n_t-1}}\bigr)\Bigr\}\\
    &\hspace{1.6cm}=\min\Bigl\{r_{t,1}, \textnormal{Med}\bigl((u_\ell)_{\ell:Z_{\hat{j}}\leq Z_\ell\leq Z_{n_t-1}}\bigr)\Bigr\}= \textnormal{Med}\bigl((u_\ell)_{\ell:Z_{\hat{j}}\leq Z_\ell\leq Z_{n_t-1}}\bigr),
\end{align*}
where the last step follows by~\eqref{eq:bounding-rt1-below}. Thus, we have that
\[\textnormal{Med}\bigl((u_\ell)_{\ell:Z_{\hat{j}}\leq Z_\ell\leq Z_{\hat{k}}}\bigr)\geq r_{t,1}> \textnormal{Med}\bigl((u_\ell)_{\ell:Z_{\hat{j}}\leq Z_\ell\leq Z_{n_t-1}}\bigr),\]
which is a contradiction.
Thus we again have $r_t=r_{t,2}$ in this case. Therefore,
for all $t\in [m-1]$,
\[r_t = r_{t,2} = \min_{k\in [n]: Z_k\geq Z_{n_t}} \textnormal{Med}\bigl((u_\ell)_{\ell:Z_{n_t}\leq Z_\ell\leq Z_k}\bigr) \leq \min_{i\in\{n_t,\dots,n_{t+1}-1\}}\textnormal{Med}\bigl((u_\ell)_{\ell:Z_{n_t}\leq Z_\ell\leq Z_i}\bigr).\]
This establishes the desired upper bound~\eqref{eq:cumulative-median-inequalities-1} on $r_t$.

    \paragraph{Lower bound on $r_t$.}
    Next we will prove that
    \begin{equation}\label{eq:cumulative-median-inequalities-2}
    r_t\geq  \textnormal{Med}\bigl((u_\ell)_{\ell:Z_i\leq Z_\ell\leq Z_{n_{t+1}-1}}\bigr)
\end{equation}
for all $t\in[m-1]$ and all $i\in\{n_t,\dots,n_{t+1}-1\}$. This proof is not simply a symmetric version of the proof of the upper bound, since the vector $\tilde{u}$ is defined by taking the maximum of the minimum (and not vice versa).

By definition, we have 
\[r_t = \tilde{u}_{n_{t+1}-1} = \max_{j\in [n]: Z_j\leq Z_{n_{t+1}-1}}\min_{k\in [n]: Z_k\geq Z_{n_{t+1}-1}} \textnormal{Med}\bigl((u_\ell)_{\ell:Z_j\leq Z_\ell\leq Z_k}\bigr).\]
If $t=m-1$ then $n_{t+1}-1=n$, and so $\{k\in[n] : Z_k \geq Z_{n_{t+1}-1}\} = \{k\in[n]: Z_k \geq Z_n\}$. Therefore, for $t=m-1$ we have
\begin{align*}
r_t = \max_{j\in [n]: Z_j\leq Z_n}\textnormal{Med}\bigl((u_\ell)_{\ell:Z_j\leq Z_\ell\leq Z_n}\bigr) &= \max_{j \in [n]} \textnormal{Med}\bigl((u_\ell)_{\ell:Z_j\leq Z_\ell\leq Z_n}\bigr) \\
&\geq \max_{j \in\{n_t,\dots,n\}}\textnormal{Med}\bigl((u_\ell)_{\ell:Z_j\leq Z_\ell\leq Z_{n}}\bigr).
\end{align*}

This proves the claim for the case $t=m-1$.

From this point on we assume $t\in[m-2]$. 
As noted above, $Z_k > Z_{n_{t+1}-1}$ for any $k\geq n_{t+1}$. 
Therefore we can write
\[
\{k\in[n]: Z_k \geq Z_{n_{t+1}-1}\} = \{k\in[n]: k\geq n_{t+1}\} \cup \{k\in[n]: Z_k = Z_{n_{t+1}-1}\},
\]
and hence,
\[
r_t = \max_{j\in[n]:Z_j \leq Z_{n_{t+1}-1}}\min\Big\{\underbrace{\min_{k\geq n_{t+1}}\textnormal{Med}\bigl((u_\ell)_{\ell:Z_j\leq Z_\ell\leq Z_k}\bigr)}_{=:r_{t,j,1}} , \ \underbrace{ \textnormal{Med}\bigl((u_\ell)_{\ell:Z_j\leq Z_\ell\leq Z_{n_{t+1}-1}}\bigr)}_{=:r_{t,j,2}}\Big\}.
\]
Fix any $j$ with $Z_j \leq Z_{n_{t+1}-1}$, and suppose that $r_{t,j,1} < r_{t,j,2}$. Then
\[
\textnormal{Med}\bigl((u_\ell)_{\ell:Z_j\leq Z_\ell\leq Z_k}\bigr) < \textnormal{Med}\bigl((u_\ell)_{\ell:Z_j\leq Z_\ell\leq Z_{n_{t+1}-1}}\bigr)
\]
for some $k\geq n_{t+1}$. Then by Lemma~\ref{lem:a-median-result}, we must have
 \[\textnormal{Med}\bigl((u_\ell)_{\ell:Z_j\leq Z_\ell\leq Z_k}\bigr) \geq \textnormal{Med}\bigl((u_\ell)_{\ell:Z_{n_{t+1}-1}< Z_\ell\leq Z_k}\bigr).\]
Now let $t' \geq t+1$ be such that $n_{t'}\leq k < n_{t'+1}$.  Since $Z_{n_{t+1}-1} < Z_{n_{t+1}}$ as explained above, we have by Lemma~\ref{lem:a-median-result} that 
\begin{multline*}
\textnormal{Med}\bigl((u_\ell)_{\ell:Z_{n_{t+1}-1}< Z_\ell\leq Z_k}\bigr) = \textnormal{Med}\bigl((u_\ell)_{\ell:Z_{n_{t+1}}\leq Z_\ell\leq Z_k}\bigr) \\\geq \min\left\{\textnormal{Med}\bigl((u_\ell)_{\ell:Z_{n_{t+1}}\leq Z_\ell < Z_{n_{t+2}}}\bigr), \dots,\textnormal{Med}\bigl((u_\ell)_{\ell:Z_{n_{t'-1}}\leq Z_\ell < Z_{n_{t'}}}\bigr),\textnormal{Med}\bigl((u_\ell)_{\ell:Z_{n_{t'}}\leq Z_\ell \leq Z_k}\bigr)\right\}\\\geq \min\{r_{t+1},\dots,r_{m-1}\} = r_{t+1},
\end{multline*}
where the second inequality applies the upper bound~\eqref{eq:cumulative-median-inequalities-1} on each $r_{t+1},\dots,r_{m-1}$. But we also have
\[
r_t \geq \min\{r_{t,j,1},r_{t,j,2}\} = r_{t,j,1} \geq r_{t+1},
\]
which is a contradiction.

Therefore we have $r_{t,j,1}\geq r_{t,j,2}$ for all $j$ such that $Z_j \leq Z_{n_{t+1}-1}$, and so
\begin{multline*}
r_t = \max_{j\in[n]:Z_j \leq Z_{n_{t+1}-1}}r_{t,j,2} = \max_{j\in[n]:Z_j \leq Z_{n_{t+1}-1}} \textnormal{Med}\bigl((u_\ell)_{\ell:Z_j\leq Z_\ell\leq Z_{n_{t+1}-1}}\bigr) \\\geq \max_{j\in\{n_t,\dots,n_{t+1}-1\}}\textnormal{Med}\bigl((u_\ell)_{\ell:Z_j\leq Z_\ell\leq Z_{n_{t+1}-1}}\bigr).
\end{multline*}
This proves the lower bound~\eqref{eq:cumulative-median-inequalities-2}, and thus completes the proof of the lemma.
\end{proof}

We are now in a position to define the isotonic median matching for any vector $u\in \RR^n$. Let $m$ be as defined in \eqref{Eq:utildeiso}. For $t\in [m-1]$, define
\begin{align*}
P_t &:= \bigl\{\ell \in \{n_t,n_t+1,\ldots,n_{t+1}-1\}: u_\ell> r_t=\tilde{u}_\ell\bigr\}, \\
N_t &:= \bigl\{\ell \in \{n_t,n_t+1,\ldots,n_{t+1}-1\}: u_\ell< r_t=\tilde{u}_\ell\bigr\}.
\end{align*}
Let us order the indices in $P_t$ as $\ell_{t,1}^+<\cdots <\ell_{t,n_{t}^+}^+$ and the indices in $N_i$ as $\ell_{t,1}^-<\cdots< \ell_{t,n_t^-}^-$. Defining $\bar{n}_t:=n_t^+\wedge n_t^-$, we consider the collection of ordered pairs
\[
\mathcal{C}_t:=\bigl\{(\ell_{t,1}^+,\ell_{t,1}^-),\cdots,(\ell_{t,\bar{n}_t}^+,\ell_{t,\bar{n}_t}^-)\bigr\}.
\]
Finally, we consider
$\widetilde{M}(u):=\bigcup_{t=1}^{m-1} \mathcal{C}_t$ and call it the isotonic median matching (IMM). In the lemma below, we prove that this is a valid matching.
\begin{lemma}\label{lem:validity-of-IMM}
    Suppose that the elements of $\{u_\ell:\ell \in [n]\}$ are all distinct. Then for any $t\in [m-1]$, we have $n_t^+=n_t^-$ and 
    $u_\ell=\tilde{u}_\ell$ for all $\ell \in \{n_t,n_t+1,\ldots,n_{t+1}-1\} \setminus \bigl(P_t\cup N_t\bigr)$.
    Moreover, $\widetilde{M}(u)\in \mathcal{M}_n(\bZ)$.
\end{lemma}
\begin{proof}
    Fix $t\in [m-1]$. When $n_{t+1}-n_t$ is even, we have by~\eqref{eq:cumulative-median-equals} (in Lemma~\ref{lem:characterizing-isotonic-L1-projection_}) that $n_t^+=n_t^-=(n_{t+1}-n_t)/2$ and $\{n_t,n_t+1,\ldots,n_{t+1}-1\}\setminus \bigl(P_t\cup N_t\bigr)=\emptyset$.  On the other hand, when $n_{t+1}-n_t$ is odd, we have $n_t^+=n_t^-=(n_{t+1}-n_t-1)/2$ and $\{n_t,n_t+1,\ldots,n_{t+1}-1\}\setminus \bigl(P_t\cup N_t\bigr)=:\{\ell_0\}$ 
    is a singleton set. Moreover, $u(\ell_0)=\tilde{u}(\ell_0)$ which proves the first two claims.

    In order to prove that $\widetilde{M}(u)$ is a valid matching, we need to show that $Z_{\ell_{t,j}^+}\leq Z_{\ell_{t,j}^-}$ for any $t\in [m-1]$ and $j\in [\bar{n}_t]$.  To see this, suppose for a contradiction that $Z_{\ell_{t,j_0}^+}>Z_{\ell_{t,j_0}^-}$ for some $t \in [m-1]$ and some minimal $j_0 \in [\bar{n}_t]$. Since $\tilde{u}\in \widehat{\mathcal{C}}_\ISO(\bZ)$, by \eqref{Eq:utildeiso}, we have $\ell\geq n_t$ whenever $Z_\ell\geq Z_{n_t}$.
    Consequently, this gives us
     \begin{align*}
         \bigl|\{\ell:Z_{n_t}\leq Z_\ell\leq Z_{\ell_{t,j_0}^-}~\text{and}~u_\ell< r_t\}\bigr|&\geq \bigl|\{\ell:n_t\leq \ell\leq \ell_{t,j_0}^-~\text{and}~u_\ell < r_t\}\bigr|=j_0\\
         \bigl|\{\ell:Z_{n_t}\leq Z_\ell\leq Z_{\ell_{t,j_0}^-}~\text{and}~u_\ell> r_t\}\bigr|&< \bigl|\{\ell:n_t\leq \ell\leq \ell_{t,j_0}^+~\text{and}~u_\ell> r_t\}\bigr|=j_0,
     \end{align*}
     where the second inequality follows by recalling that $\ell_{t,1}^+<\cdots <\ell_{t,n_{t}^+}^+$.
     Since the elements of $\{u_\ell:\ell \in [n]\}$ are all distinct, this further implies that
   \[
   \textnormal{Med}\bigl((u_\ell)_{\ell:Z_{n_t}\leq Z_\ell\leq Z_{\ell_{t,j_0}^-}}\bigr) 
        < r_t,
    \]
    which contradicts the result~\eqref{eq:cumulative-median-inequalities} proved in Lemma~\ref{lem:characterizing-isotonic-L1-projection_}.
\end{proof}
With the definition and properties of the IMM now in place, we are now ready to prove Theorem~\ref{thm:key-property-oracle}. This result establishes the claim~\eqref{eq:key-property-oracle-matching} in Appendix~\ref{app:power-oracle-estimated}, which was used for proving the power guarantees for both oracle and plug-in matching in Theorems~\ref{thm:oracle-matching-asymptote} and~\ref{thm:oracle-matching-estimated-mu}.

\begin{proof}[Proof of Theorem~\ref{thm:key-property-oracle}]
We start by proving the upper bound. Fix any matching $M\in\mathcal{M}_n(\bZ)$.
By definition, $\sigma_n\widehat{\ISNR}_n=\inf_{g\in\mathcal{C}_\ISO}\|\mu_n(\bY,\bZ)-g(\bZ)\|_2$.
    We claim that 
    \begin{equation}
    \label{Eq:Claim}
    \Bigl(\mu_n\bigl(Y_{i_\ell},Z_{i_\ell}\bigr)-\mu\bigl(Y_{j_\ell},Z_{j_\ell}\bigr)\Bigr)_+\leq \bigl|\mu_n\bigl(Y_{i_\ell},Z_{i_\ell}\bigr)-g\bigl(Z_{i_\ell}\bigr)\bigr|+\bigl|g\bigl(Z_{j_\ell}\bigr)-\mu\bigl(Y_{j_\ell},Z_{j_\ell}\bigr)\bigr| 
    \end{equation}
    for every $\ell\in [L]$ and every $g \in \mathcal{C}_\ISO$.  To see this, we may assume that  $\mu_n\bigl(Y_{i_\ell},Z_{i_\ell}\bigr)>\mu_n\bigl(Y_{j_\ell},Z_{j_\ell}\bigr)$ since otherwise the left-hand side is zero, and the inequality trivially holds. Then
    \begin{align*}
        &\hspace{-.5in}\bigl(\mu_n\bigl(Y_{i_\ell},Z_{i_\ell}\bigr)-\mu_n\bigl(Y_{j_\ell},Z_{j_\ell}\bigr)\bigr)_+
        =\mu_n\bigl(Y_{i_\ell},Z_{i_\ell}\bigr)-\mu_n\bigl(Y_{j_\ell},Z_{j_\ell}\bigr)\\
        &=\mu_n\bigl(Y_{i_\ell},Z_{i_\ell}\bigr)-g(Z_{i_\ell}) + g(Z_{j_\ell}) - \mu_n\bigl(Y_{j_\ell},Z_{j_\ell}\bigr) + g(Z_{i_\ell}) - g(Z_{j_\ell})\\
        &\leq \bigl|\mu_n\bigl(Y_{i_\ell},Z_{i_\ell}\bigr)-g\bigl(Z_{i_\ell}\bigr)\bigr|+\bigl|g\bigl(Z_{j_\ell}\bigr)-\mu_n\bigl(Y_{j_\ell},Z_{j_\ell}\bigr)\bigr| + \left(g(Z_{i_\ell}) - g(Z_{j_\ell})\right)\\
        &\leq \bigl|\mu_n\bigl(Y_{i_\ell},Z_{i_\ell}\bigr)-g\bigl(Z_{i_\ell}\bigr)\bigr|+\bigl|g\bigl(Z_{j_\ell}\bigr)-\mu_n\bigl(Y_{j_\ell},Z_{j_\ell}\bigr)\bigr|,
    \end{align*}
    where the last step holds since we must have $Z_{i_\ell}\leq Z_{j_\ell}$ (since $M\in\mathcal{M}_n(\bZ)$) and therefore $g(Z_{j_\ell})-g(Z_{j_\ell})\leq 0$ (since $g\in\mathcal{C}_\ISO$).
    This proves the claim~\eqref{Eq:Claim}, and therefore, it follows that
    \begin{align*}
    \|\Delta\mu_n^+(\bY,\bZ;M)\|_2^2 &=\sum_{\ell=1}^L  \Bigl(\mu_n\bigl(Y_{i_\ell},Z_{i_\ell}\bigr)-\mu_n\bigl(Y_{j_\ell},Z_{j_\ell}\bigr)\Bigr)_+^2\\
    &\leq \sum_{\ell=1}^L \left( \bigl|\mu_n\bigl(Y_{i_\ell},Z_{i_\ell}\bigr)-g\bigl(Z_{i_\ell}\bigr)\bigr|+\bigl|g\bigl(Z_{j_\ell}\bigr)-\mu_n\bigl(Y_{j_\ell},Z_{j_\ell}\bigr)\bigr|\right)^2\\
    &\leq 2\sum_{\ell=1}^L \bigl\{\mu_n\bigl(Y_{i_\ell},Z_{i_\ell}\bigr)-g\bigl(Z_{i_\ell}\bigr)\bigr\}^2 + 2\sum_{\ell=1}^L \bigl\{\mu_n\bigl(Y_{j_\ell},Z_{j_\ell}\bigr)-g\bigl(Z_{j_\ell}\bigr)\bigr\}^2 \\
    &\leq 2 \sum_{i=1}^n \bigl\{\mu_n(Y_i,Z_i)-g(Z_i)\bigr\}^2 = 2\|\mu_n(\bY,\bZ) - g(\bZ)\|^2_2 \leq 2\sigma_n^2\, \widehat{\ISNR}_n^2.
    \end{align*}
    This completes the proof of the first part.

    Next we turn to the lower bound. It is sufficient to show that, for any $\epsilon>0$, we can find some $M\in\mathcal{M}_n(\bZ)$ such that
    $\|\Delta\mu_n^+(\bY,\bZ;M)\|_2^2 \geq \sigma_n^2\widehat{\ISNR}_n^2 - \epsilon$.
    Fixing any $\epsilon>0$, first we note that we can find some $u\in \RR^n$ such that the coordinates of $u$ are all distinct and $\|\mu_n(\bY,\bZ)-u\|_2<\epsilon/(\sqrt{2}+1)$. Consider the IMM for this vector $u$,  $\widetilde{M}(u) := \{(i_\ell,j_\ell)\}_{\ell \in \tilde{L}}$, where $\widetilde{M}(u)\in \mathcal{M}_n(\bZ)$ by Lemma~\ref{lem:validity-of-IMM}.
Furthermore  by Lemma~\ref{lem:validity-of-IMM} and~\eqref{Eq:utildeiso}, we have for any $\ell\in [\tilde{L}]$ that
$u_{i_\ell} > \tilde{u}_{i_\ell} = \tilde{u}_{j_\ell} >u_{j_\ell}$. 
    Hence, 
    \begin{align*}
      \|\Delta u^+\|_2^2=\sum_{\ell=1}^{\tilde{L}}\bigl(u_{i_\ell}-u_{j_\ell}\bigr)^2
      & = \sum_{\ell=1}^{ \tilde{L}}\bigl(u_{i_\ell}-\tilde{u}_{i_\ell}+\tilde{u}_{j_\ell}-u_{j_\ell}\bigr)^2  \\
      &\geq \sum_{\ell=1}^{\tilde{L}}\bigl\{\bigl(u_{i_\ell}-\tilde{u}_{i_\ell}\bigr)^2 + \bigl(\tilde{u}_{j_\ell}-u_{j_\ell}\bigr)^2\bigr\} \\
      &=\sum_{i=1}^n \bigl(u_i-\tilde{u}_i\bigr)^2=\|u-\tilde{u}\|_2^2,
    \end{align*}
    where the penultimate equality holds because $u_i=\tilde{u}_i$ for $i \in [n] \setminus \{i_1,j_1,\ldots,i_{\tilde{L}},j_{\tilde{L}}\}$, by Lemma~\ref{lem:validity-of-IMM}.
    Next, we observe that 
    \begin{align*}
        \|\Delta u^+\|_2-\|\Delta\mu_n^+(\bY,\bZ;\widetilde{M}(u))\|_2&\leq \|\Delta u^+-\Delta\mu_n^+(\bY,\bZ;\widetilde{M}(u))\|_2\\
        &\leq \|\Delta u-\Delta\mu_n(\bY,\bZ;\widetilde{M}(u))\|_2\\
        &\leq \sqrt{2} \|u-\mu_n(\bY,\bZ)\|_2\leq\sqrt{2} \cdot \frac{\epsilon}{\sqrt{2}+1},
    \end{align*}
    and that by \eqref{eq: isotonic_signal_strength},
    \begin{align*}
        \sigma_n\widehat{\ISNR}_n &\leq \|\mu_n(\bY,\bZ)-\tilde u\|_2\leq \|u-\tilde{u}\|_2+\frac{\epsilon}{\sqrt{2}+1}.
    \end{align*}
    Therefore,
    \begin{multline*}
    \|\Delta\mu_n^+(\bY,\bZ;\widetilde{M}(u))\|_2 \geq \|\Delta u^+\|_2 - \frac{\sqrt{2}\epsilon}{\sqrt{2}+1}\geq \|u-\tilde{u}\|_2 - \frac{\sqrt{2}\epsilon}{\sqrt{2}+1}\\\geq \left(\sigma_n\widehat{\ISNR}_n -\frac{\epsilon}{\sqrt{2}+1}\right) - \frac{\sqrt{2}\epsilon}{\sqrt{2}+1} = \sigma_n\widehat{\ISNR}_n -\epsilon,
    \end{multline*}
    which completes the proof of the second part.
\end{proof}

\subsection{Proof of lemmas from Appendices~\ref{app:general_power}—\ref{app:power-cross-bin}}\label{app:lemmas-from-appendix-c1c4}
\subsubsection{Lemmas used to prove general power guarantees in Appendix~\ref{app:general_power}}
\begin{lemma}\label{lem:concentration_L2_L3_for_DeltaX}
    Under the setting and assumptions of Theorem~\ref{thm:general power analysis}, 
    for any $\delta>0$, we have that
    \begin{align*}
        \mathbb{P}\biggl\{\frac{\|\bw\circ \Delta\bX\|^2_2}{\left(2\sigma_n^2\|\bw\|^2_2 +   \|\bw\circ \Delta\mu_n(\bY,\bZ)\|^2_2\right)} &\in \bigl[(1 -\epsilon_*)\vee 0,1+\epsilon_*\bigr]~~\text{and}\\
        \|\bw\circ\Delta\bX\|^2_3 &\leq \left(2\sigma_n^2\|\bw\|^2_2 +   \|\bw\circ \Delta\mu_n(\bY,\bZ)\|^2_2\right)\cdot \epsilon_*'\biggm| \bY,\bZ\biggr\}\geq 1-\delta
    \end{align*}
    where
    \[\epsilon_* = \sqrt{\frac{3}{\delta}}\cdot\biggl(  2\rho^2_4 \epsilon_1 + 2(\epsilon_1^2\epsilon_2)^{1/3}\biggr), \quad \epsilon_*' = \frac{(24)^{2/3}\rho_4^2}{\delta^{2/3}}\cdot\epsilon_1^{2/3} + 4\epsilon_2^{2/3},\]
    and we write $\rho_4:=\Ep{P_\zeta}{\zeta^4}^{1/4}$.
    \end{lemma}
\begin{proof}
Recall the notation from the hypothesis in statement of Theorem~\ref{thm:general power analysis}:
\[
\epsilon_1 = \frac{\|\bw\|_\infty}{\|\bw\|_2}, \quad \epsilon_2 = \frac{\|\bw\circ\Delta\mu_n(\bY,\bZ)\|^3_3}{\sigma_n^3\|\bw\|^3_2 + \|\bw\circ\Delta\mu_n(\bY,\bZ)\|^3_2}.
\]
Under the model class~\eqref{model:general alternative}, $\Delta\bX = \Delta\mu_n(\bY,\bZ) + \sigma_n\Delta\bzeta$, and thus, we can write
\begin{equation}\label{eq:expand-2norm}
    \|\bw\circ \Delta\bX\|^2_2 = \sigma_n^2\|\bw\circ \Delta\bzeta\|^2_2 + 2\sigma_n (\bw\circ \Delta\bzeta)^\top \bigl(\bw\circ \Delta\mu_n(\bY,\bZ)\bigr) + \|\bw\circ \Delta\mu_n(\bY,\bZ)\|^2_2.
\end{equation}
For the first term, we have
\begin{align*}
    \EEst{\sigma_n^2\|\bw\circ \Delta\bzeta\|^2_2}{\bY,\bZ} &= 2\sigma_n^2\|\bw\|^2_2,\\
    \VVst{\sigma_n^2\|\bw\circ \Delta\bzeta\|^2_2}{\bY,\bZ} &\leq  \sum_{\ell=1}^L w_\ell^4\sigma_n^4 \Ep{\zeta,\zeta'\iidsim P_\zeta}{|\zeta - \zeta'|^4} \\
    &\leq 16\rho^4_4\sigma_n^4\|\bw\|^4_4 \leq 16\rho^4_4\sigma_n^4\|\bw\|^2_2\|\bw\|^2_\infty \leq  16\rho^4_4\sigma_n^4\|\bw\|^4_2\cdot \epsilon_1^2,
\end{align*}
and therefore by Chebyshev's inequality, for every $A > 0$,
\[\PPst{\left|\sigma_n^2\|\bw\circ \Delta\bzeta\|^2_2 - 2\sigma_n^2\|\bw\|^2_2 \right|\geq A}{\bY,\bZ} \leq\frac{16\rho^4_4\sigma_n^4\|\bw\|^4_2\cdot \epsilon_1^2}{A^2}.
\]
For the second term in \eqref{eq:expand-2norm}, we have that
\[
\EEst{2\sigma_n(\bw\circ \Delta\bzeta)^\top \bigl(\bw\circ \Delta\mu(\bY,\bZ)\bigr)}{\bY,\bZ} = 0
\]
and moreover, by H{\"o}lder's inequality,
\begin{align*}
&\VVst{2\sigma_n(\bw\circ \Delta\bzeta)^\top \bigl(\bw\circ \Delta\mu_n(\bY,\bZ)\bigr)}{\bY,\bZ} \\
&\hspace{5cm}=
4\sum_{\ell=1}^L w_\ell^2\sigma_n^2 \VVst{\Delta_\ell \bzeta}{\bY,\bZ} \cdot \bigl(w_\ell \Delta_\ell\mu_n(\bY,\bZ)\bigr)^2\\
&\hspace{5cm}=8\sigma_n^2  \sum_{\ell=1}^L w_\ell^2 \bigl(w_\ell \Delta_\ell\mu_n(\bY,\bZ)\bigr)^2 \\
&\hspace{5cm}\leq 8\sigma_n^2 \|\bw\|_6^{2} \|\bw\circ \Delta\mu_n(\bY,\bZ))\|_3^2\\
&\hspace{5cm}\leq  8\sigma_n^2 \left(\|\bw\|^2_2\|\bw\|^4_\infty\right)^{1/3}\|\bw\circ\Delta\mu_n(\bY,\bZ)\|_3^2 \\
&\hspace{5cm}=  8\sigma_n^2 \|\bw\|^2_2\left(\sigma_n^3\|\bw\|^3_2 + \|\bw\circ\Delta\mu_n(\bY,\bZ)\|^3_2\right)^{2/3} \cdot(\epsilon_1^2\epsilon_2)^{2/3}\\
&\hspace{5cm}\leq  8\sigma_n^2 \|\bw\|^2_2 \left(\sigma_n^2\|\bw\|^2_2 +  \|\bw\circ\Delta\mu_n(\bY,\bZ)\|^2_2\right)\cdot(\epsilon_1^2\epsilon_2)^{2/3},
\end{align*}
where the penultimate step uses the definitions of $\epsilon_1^2$ and $\epsilon_2$.  Another application of Chebyshev's inequality then gives us for every $B > 0$ that
\begin{align*}
&\PPst{\left|2\sigma_n(\bw\circ \Delta\bzeta)^\top (\bw\circ \Delta\mu_n(\bY,\bZ)) \right|\geq   B}{\bY,\bZ} \\
&\hspace{5cm}\leq \frac{ 8\sigma_n^2 \|\bw\|^2_2 \left(\sigma_n^2\|\bw\|^2_2 +  \|\bw\circ\Delta\mu_n(\bY,\bZ)\|^2_2\right)\cdot(\epsilon_1^2\epsilon_2)^{2/3} }{B^2}.
\end{align*}
Together, then,
\begin{multline*}
\PPst{\left|\|\bw\circ \Delta\bX\|^2_2 -  \left(2\sigma_n^2\|\bw\|^2_2 +   \|\bw\circ \Delta\mu_n(\bY,\bZ)\|^2_2\right)\right| \geq  A + B}{\bY,\bZ}\\ \leq \frac{16\sigma_n^4\rho^4_4\|\bw\|^4_2\cdot \epsilon_1^2}{A^2} +\frac{ 8\sigma_n^2 \|\bw\|^2_2 \left(\sigma_n^2\|\bw\|^2_2 +  \|\bw\circ\Delta\mu(\bY,\bZ)\|^2_2\right)\cdot(\epsilon_1^2\epsilon_2)^{2/3} }{B^2}. \end{multline*}
Next, we choose 
\begin{align*}
    A^2 &= \frac{3}{\delta}\cdot  16\rho^4_4\sigma_n^4\|\bw\|^4_2\cdot \epsilon_1^2,~~\text{and}\\
    B^2 &=\frac{3}{\delta} \cdot  8\sigma_n^2 \|\bw\|^2_2 \left(\sigma_n^2\|\bw\|^2_2 +  \|\bw\circ\Delta\mu_n(\bY,\bZ)\|^2_2\right)\cdot(\epsilon_1^2\epsilon_2)^{2/3}.
\end{align*}
Then 
\[
\PPst{\frac{\|\bw\circ \Delta\bX\|^2_2}{\left(2\sigma_n^2\|\bw\|^2_2 +   \|\bw\circ \Delta\mu_n(\bY,\bZ)\|^2_2\right)}\in [(1 -\epsilon_*) \vee 0,1+\epsilon_*]}{\bY,\bZ}\\ \geq 1-\frac{2\delta}{3}.
\]

Next we bound $\|\bw\circ\Delta\bX\|^3_3$. 
We start by writing
\begin{align*}
\|\bw\circ\Delta\bX\|_3 &\leq \sigma_n\|\bw\circ\Delta\bzeta\|_3 + \|\bw\circ\Delta\mu_n(\bY,\bZ)\|_3 \\
&= \sigma_n\|\bw\circ\Delta\bzeta\|_3 + \left(\sigma_n^3\|\bw\|^3_2 +\|\bw\circ\Delta\mu_n(\bY,\bZ)\|^3_2\right)^{1/3}\cdot \epsilon_2^{1/3}\\
&\leq \sigma_n\|\bw\circ\Delta\bzeta\|_3 + \left(\sigma_n\|\bw\|_2 +\|\bw\circ\Delta\mu_n(\bY,\bZ)\|_2\right)\cdot \epsilon_2^{1/3}.
\end{align*}
Further, we have that
\begin{align*}
    \EEst{\|\bw\circ\Delta\bzeta\|_3^3}{\bY,\bZ} &= \sum_{\ell=1}^L w_\ell^3 \EEst{|\Delta_\ell\bzeta|^3}{\bY,\bZ} \\&\leq \|\bw\|^3_3 \cdot\Ep{\zeta,\zeta'\iidsim P_\zeta}{|\zeta-\zeta'|^3} \leq  \|\bw\|^3_3 \cdot 8\rho^3_3
\end{align*}
where $\rho_3 = \Ep{\zeta\sim P_\zeta}{|\zeta|^3}^{1/3}$. 
Since $\rho_3\leq \rho_4$, we further have that
\[
\EEst{\|\bw\circ\Delta\bzeta\|_3^3}{\bY,\bZ} \leq 8\rho_4^3\|\bw\|^2_2 \cdot \|\bw\|_\infty \leq 8\rho_4^3\|\bw\|^3_2 \cdot\epsilon_1.
\]
By Markov's inequality, then,
\[
\PPst{\sigma_n^3\|\bw\circ\Delta\bzeta\|_3^3 \leq \frac{24\sigma_n^3\rho_4^3\|\bw\|^3_2 \cdot\epsilon_1}{\delta}}{\bY,\bZ} \geq 1 - \frac{\delta}{3}.
\]
Together, then, we have
\[
\PPst{\|\bw\circ\Delta\bX\|_3 \leq \frac{\sqrt[3]{24}\sigma_n\rho_4\|\bw\|_2 \epsilon_1^{1/3}}{\delta^{1/3}} + \bigl(\sigma_n\|\bw\|_2 +\|\bw\circ\Delta\mu_n(\bY,\bZ)\|_2\bigr) \epsilon_2^{1/3}}{\bY,\bZ} \geq 1 - \frac{\delta}{3},
\]
which we can relax to
\begin{align*}
&\PPst{\|\bw\circ\Delta\bX\|^2_3 \leq  \frac{2(24)^{2/3}\sigma_n^2\rho_4^2\|\bw\|^2_2 \epsilon_1^{2/3}}{\delta^{2/3}} + 4  \bigl(\sigma_n^2\|\bw\|^2_2 +\|\bw\circ\Delta\mu_n(\bY,\bZ)\|^2_2\bigr) \epsilon_2^{2/3}}{\bY,\bZ} \\
&\hspace{14.5cm}\geq 1 - \frac{\delta}{3}.
\end{align*}
Thus
\[
\PPst{\|\bw\circ\Delta\bX\|^2_3 \leq \left(2\sigma_n^2\|\bw\|^2_2 +   \|\bw\circ \Delta\mu_n(\bY,\bZ)\|^2_2\right)\cdot \epsilon_*'}{\bY,\bZ} \geq 1 - \frac{\delta}{3}
\]
as required.
\end{proof}
\begin{lemma}\label{lem:gaussian-CDF-approximation}
        Consider any real-valued sequence $(x_n)_{n\in \mathbb{N}}$, and $y\in \RR$. If $\epsilon_n\overset{P}{\to}0$ and $\epsilon_n^\prime\overset{P}{\to}0$ as $n\to\infty$, then 
        \[
        \bigl|\Phi\bigl(x_n(1+\epsilon_n)+y+\epsilon_n^\prime\bigr)-\Phi(x_n+y)\bigr|\overset{P}{\to} 0.
        \]
    \end{lemma}
    \begin{proof}
        First, since $\Phi$ is $\frac{1}{\sqrt 2\pi}$-Lipschitz, we have that 
        \begin{align*}
            \bigl|\Phi\bigl(x_n(1+\epsilon_n)+y+\epsilon_n^\prime\bigr)-\Phi\bigl((x_n+y)\cdot(1+\epsilon_n)\bigr)\bigr|&\leq \frac{1}{\sqrt{2\pi}}\,|\epsilon_n^\prime-y\epsilon_n|\\
            &\leq \frac{|\epsilon_n^\prime|+|y||\epsilon_n|}{\sqrt{2\pi}}=\mathrm{o}_P(1).
        \end{align*}
        Hence, it suffices to show that 
    \[
    \bigl|\Phi\bigl((x_n+y)\cdot(1+\epsilon_n)\bigr)-\Phi(x_n+y)\bigr|\overset{P}{\to} 0.
        \]
        In fact, without loss of generality, we may take $y=0$.
         By the mean value theorem, there exists a real-valued sequence of $(\bar{x}_n)_{n\in \mathbb{N}}$, where $\bar{x}_n$ lies between $x_n$ and $x_n(1+\epsilon_n)$, such that 
        \[
            \bigl|\Phi\bigl(x_n(1+\epsilon_n)\bigr)-\Phi(x_n)\bigr|=\phi(\bar{x}_n) \cdot |x_n\epsilon_n|=\phi\bigl(|\bar{x}_n|\bigr) \,|x_n|\,|\epsilon_n|.
        \]
        Define the event $A_n:=\bigl\{|\bar{x}_n|\geq |x_n|/2\bigr\}$
        and note that $\mathbb{P}(A_n) \geq \mathbb{P}(|\epsilon_n|\leq 1/2)$. Since $\epsilon_n\overset{P}{\to} 0$ as $n\to \infty$, we have $\mathbb{P}(A_n)\to 1$ as $n\to\infty$. Moreover, writing $K := \sup_{x \in \mathbb{R}} |x|\exp(-x^2/8) \in (0,\infty)$, we have on $A_n$ that,
        \[
        \phi(|\bar{x}_n|)\cdot |x_n|\leq \exp\bigl(-x_n^2/8\bigr)\cdot |x_n|\leq K.
        \]
        Finally fix $\delta>0$ and observe that
        \begin{align*}
            \mathbb{P}\bigl\{\bigl|\Phi\bigl(x_n(1+\epsilon_n)\bigr)-\Phi(x_n)\bigr| \geq \delta\bigr\} &\leq \mathbb{P}\Bigl[\bigl\{\bigl|\Phi\bigl(x_n(1+\epsilon_n)\bigr)-\Phi(x_n)\bigr|\geq \delta\bigr\}\cap A_n\Bigr] + \mathbb{P}(A_n^c) \\    
            &\leq \mathbb{P}\bigl[\bigl\{|\epsilon_n|\geq \delta/K\bigr\}\cap A_n\bigr] + \mathbb{P}(A_n^c)\\
            &\leq \mathbb{P}\bigl(|\epsilon_n|\geq \delta/K\bigr) + \mathbb{P}(A_n^c)\to\,0,
            \end{align*}
            as $n\to\infty$. Since $\delta > 0$ was arbitrary, the result follows.        
\end{proof}
\subsubsection{Concentration results for neighbor matching}\label{app:conc-neighbour}
We now prove a concentration result for the neighbor matching in the asymptotic regime from Section~\ref{sec:partial_linear_model}.
Let $(Y_1,Z_1),(Y_2,Z_2),\dots\iidsim P_{Y,Z}$, where $P_{(Y,Z)}$ is a distribution on $[-1,1]\times\R$, and let $Z_{(n,1)}\leq \dots \leq Z_{(n,n)}$ be the order statistics of $Z_1,\dots,Z_n$. Abusing notation, we will write $Y_{(n,i)}$ to denote the $Y$ value corresponding to $Z_{(n,i)}$---that is, 
\[(Y_{(n,1)},\dots,Y_{(n,n)})\mid (Z_{(n,1)},\dots,Z_{(n,n)})\sim P_{Y|Z}(\cdot\mid Z_{(n,1)})\times\dots \times P_{Y|Z}(\cdot\mid Z_{(n,n)}).\]
Then for neighbor matching at sample size $n$, we are interested in the quantity
\[
{\rm NM}_n  := \frac{2}{n}\sum_{i=1}^{\lfloor n/2\rfloor} (Y_{(n,2i-1)} - Y_{(n,2i)})^2_+.
\]
In the following result, we show that this quantity concentrates around the mean of the conditional variance of $P_{Y|Z}$:
\begin{proposition}\label{lem:neighbour_matching_conc}
Under the assumptions and definitions stated above,
\[
\bigl|{\rm NM}_n - \mathbb{E}\bigl[\VPst{Y\sim P_{Y\mid Z}}{Y}\bigr]\bigr| = {\rm o}_P(1).
\]
\end{proposition}

\begin{proof}
The proof is split into four key steps.
\paragraph{Step 1: a concentration step.}
First, we note that conditional on $Z_1,\dots,Z_n$, the sum $\sum_{i=1}^{\lfloor n/2\rfloor}(Y_{(n,2i-1)} - Y_{(n,2i)})^2_+$ is the sum of $\lfloor n/2\rfloor$ independent terms, with each summand lying in $[0,4]$. Hence, for any $\delta>0$, by Hoeffding's inequality we have for $n \geq 2$ that 
\[
\PPst{\bigl|{\rm NM}_n - \EEst{{\rm NM}_n}{Z_1,\dots,Z_n}\bigr| \geq 4 \cdot \frac{2}{n}\sqrt{\frac{\lfloor n/2\rfloor  \cdot \log(2/\delta)}{2}}}{\bZ} \leq \delta.
\]
In particular,
\begin{equation}\label{eq:concentration-of-NM}
    \bigl|{\rm NM}_n - \EEst{{\rm NM}_n}{Z_1,\dots,Z_n}\bigr| = {\rm O}_P(n^{-1/2}).
\end{equation}
\paragraph{Step 2: comparing to i.i.d.\ pairs.}
Since $(y_1,y_2)\mapsto (y_1-y_2)^2_+$ is $4$-Lipschitz with respect to the $\ell_1$ norm on $[-1,1]^2$, for any two distributions $Q$, $Q'$ on $[-1,1]^2$ we have
\[\left|\Ep{(Y,Y')\sim Q}{(Y-Y')^2_+} - \Ep{(Y,Y')\sim Q'}{(Y-Y')^2_+}\right|\leq 4{\rm d}_{\rm W}(Q,Q'),\]
where ${\rm d}_{\rm W}(\cdot,\cdot)$ is the $L_1$-Wasserstein distance between distributions, defined as
\[
\mathrm{d}_{\rm W}(P,Q):= \inf_{(U,V) \sim (P,Q)} \mathbb{E}\|U-V\|_2,
\]
where the infimum is taken over all couplings of $U$ and $V$ with $U \sim P$ and $V \sim Q$.  In particular, in the special case that $Q=Q_1\times Q_2$ and $Q'=Q_1'\times Q_2'$ are both product distributions, then we have
\[\left|\Ep{(Y,Y')\sim Q_1\times Q_2}{(Y-Y')^2_+} - \Ep{(Y,Y')\sim Q'_1\times Q'_2}{(Y-Y')^2_+}\right|\leq 4 \bigl({\rm d}_{\rm W}(Q_1,Q'_1)+{\rm d}_{\rm W}(Q_2,Q'_2)\bigr).\]

Now, fix any $i\in\{1,\dots,\lfloor n/2\rfloor\}$ and observe that
\[\big( Y_{(n,2i-1)}, Y_{(n,2i)}\big) \mid(Z_1,\dots,Z_n)\sim P_{Y|Z}(\cdot\mid Z_{(n,2i-1)})\times P_{Y|Z}(\cdot\mid Z_{(n,2i)}).\]
Hence, for any $z\in \R$, 
\begin{align*}
    &\left|\EEst{\left(Y_{(n,2i-1)} -Y_{(n,2i)}\right)^2_+}{Z_1,\dots,Z_n}
    - \textnormal{Var}_{Y \sim P_{Y|Z}(\cdot\mid z)}(Y)\right|\\
    &= \left|\EEst{\left(Y_{(n,2i-1)} -Y_{(n,2i)}\right)^2_+}{Z_1,\dots,Z_n}
    - \Ep{Y,Y'\iidsim P_{Y|Z}(\cdot\mid z)}{(Y-Y')^2_+}\right|\\
    &\leq 4\bigl\{{\rm d}_{\rm W}\left(P_{Y|Z}(\cdot\mid Z_{(n,2i-1)}),P_{Y|Z}(\cdot\mid z)\right)+{\rm d}_{\rm W}\left(P_{Y|Z}(\cdot\mid Z_{(n,2i)}),P_{Y|Z}(\cdot\mid z)\right)\bigr\}.
\end{align*}
Now let $Z'_1,Z'_2,\dots\iidsim P_Z$ be drawn independently of the data. Then, applying the calculations above with $z=Z'_{(\lfloor n/2\rfloor, i)}$,
\begin{align*}
    &\left|\EEst{\left(Y_{(n,2i-1)} -Y_{(n,2i)}\right)^2_+}{Z_1,\dots,Z_n}
    - \textnormal{Var}_{Y \sim P_{Y|Z}(\cdot\mid Z'_{(\lfloor n/2\rfloor, i)})}(Y)\right|\\
    &\leq 4\bigl\{{\rm d}_{\rm W}\left(P_{Y|Z}(\cdot\mid Z_{(n,2i-1)}),P_{Y|Z}(\cdot\mid Z'_{(\lfloor n/2\rfloor, i)})\right)+{\rm d}_{\rm W}\left(P_{Y|Z}(\cdot\mid Z_{(n,2i)}),P_{Y|Z}(\cdot\mid Z'_{(\lfloor n/2\rfloor, i)})\right)\bigr\}.
\end{align*}
For $n\geq 3$, summing the inequality above over $i=1,\dots,\lfloor n/2\rfloor$, we obtain
\begin{align}\label{eq:till-wasserstein-bound-neighbour-matching}
    & \biggl|\EEst{{\rm NM}_n}{Z_1,\dots,Z_n} - \frac{1}{\lfloor n/2 \rfloor} \sum_{i=1}^{\lfloor n/2\rfloor} \textnormal{Var}_{Y \sim P_{Y|Z}(\cdot\mid Z'_{(\lfloor n/2\rfloor, i)})}(Y)\biggr|\nonumber\\
    &\leq  \frac{1}{\lfloor n/2 \rfloor} \sum_{i=1}^{\lfloor n/2\rfloor}\left|\EEst{\left(Y_{(n,2i-1)} -Y_{(n,2i)}\right)^2_+}{Z_1,\dots,Z_n}
    - \textnormal{Var}_{Y \sim P_{Y|Z}(\cdot\mid Z'_{(\lfloor n/2\rfloor, i)})}(Y)\right|
    +\frac{4}{n}\nonumber\\
    &\leq \frac{8}{n-2}\sum_{i=1}^{\lfloor n/2\rfloor} \sum_{j = 2i-1}^{2i} {\rm d}_{\rm W}\big(P_{Y|Z}(\cdot\mid Z_{(n,j)}), P_{Y|Z}(\cdot\mid Z'_{(\lfloor n/2\rfloor, i)})\big)+\frac{4}{n}.
\end{align}

\paragraph{Step 3: concentration around the mean of conditional variance.}
Note that
\[
\frac{1}{\lfloor n/2 \rfloor} \sum_{i=1}^{\lfloor n/2\rfloor} \textnormal{Var}_{Y \sim P_{Y|Z}(\cdot\mid Z'_{(\lfloor n/2\rfloor, i)})}(Y)
= \frac{1}{\lfloor n/2\rfloor} \sum_{i=1}^{\lfloor n/2\rfloor} \textnormal{Var}_{Y \sim P_{Y|Z}(\cdot\mid Z'_{i})}(Y),
\]
since by definition, $Z'_{(\lfloor n/2\rfloor, 1)},\dots,Z'_{(\lfloor n/2\rfloor, \lfloor n/2\rfloor)}$ is simply a permutation of $Z'_1,\dots,Z'_{\lfloor n/2\rfloor}$. 
Moreover, since $Z'_1,Z'_2,\ldots \iidsim P_Z$, the random variables $\textnormal{Var}_{Y \sim P_{Y|Z}(\cdot\mid Z'_{i})}(Y)$ are independent and identically distributed, with mean $\EE{\textnormal{Var}(P_{Y|Z})}$ and variance at most 1. Therefore, by the central limit theorem, we have
\[
\biggl| \frac{1}{\lfloor n/2\rfloor } \sum_{i=1}^{\lfloor n/2\rfloor} \textnormal{Var}_{Y \sim P_{Y|Z}(\cdot\mid Z'_{i})}(Y) -  \mathbb{E}\bigl[\VPst{Y\sim P_{Y\mid Z}}{Y}\bigr]\biggr| =  {\rm O}_P(n^{-1/2}).
\]
Combining this with \eqref{eq:concentration-of-NM} and \eqref{eq:till-wasserstein-bound-neighbour-matching}, we have
\begin{multline*}
    \left|{\rm NM}_n - \EE{\textnormal{Var}(P_{Y|Z})}\right|\\\leq 
     \frac{8}{n-2}\sum_{i=1}^{\lfloor n/2\rfloor} \sum_{j = 2i-1}^{2i} {\rm d}_{\rm W}\big(P_{Y|Z}(\cdot\mid Z_{(n,j)}), P_{Y|Z}(\cdot\mid Z'_{(\lfloor n/2\rfloor, i)})\big)  + {\rm O}_P(n^{-1/2}).
\end{multline*}

\paragraph{Step 4: the Wasserstein distance.}
Finally, by Lemma~\ref{lem:chatterjee_lemma_implies_dw} (applied with $K_n = \lfloor n/2\rfloor$, and $B_{n,k} = \{2k-1,2k\}$ for $k=1,\dots,K_n$), we have
\[
\frac{8}{n-2}\sum_{i=1}^{\lfloor n/2\rfloor} \sum_{j = 2i-1}^{2i} {\rm d}_{\rm W}\big(P_{Y|Z}(\cdot\mid Z_{(n,j)}), P_{Y|Z}(\cdot\mid Z'_{(\lfloor n/2\rfloor, i)})\big)={\rm o}_P(1).
\]
This completes the proof.
\end{proof}
\subsubsection{Concentration results for cross-bin matching}

The aim of this subsection is to prove a concentration result for cross-bin matching under the asymptotic framework of Section~\ref{sec:partial_linear_model}, analogous to Proposition~\ref{lem:neighbour_matching_conc} for neighbor matching in Section~\ref{app:conc-neighbour}. To state and prove the result, we first introduce some more notation for cross-bin matching at sample size $n$ with $K_n$ bins. Let $m_n=\lfloor n/K_n\rfloor$, and let
\[
A_{n,1} = \big\{1,\dots,m_n\big\}, \dots, A_{n,K_n} = \big\{(K_n-1)m_n+ 1, \dots, K_nm_n\big\}
\]
be an uniform partition of $\{1,\ldots,K_nm_n\}$ into $K_n$ bins.  Given $y,y'\in[-1,1]^{m_n}$, let $y_{(1)}\leq \dots \leq y_{(m_n)}$ and $y'_{(1)}\leq \dots \leq y'_{(m_n)}$ denote the sorted values of $y$ and $y'$, and define $f_n: [-1,1]^{m_n}\times [-1,1]^{m_n} \to [0,\infty)$ by
\begin{equation}
\label{Eq:fn}
f_n(y,y') = \sum_{i=1}^{\lfloor m_n/2\rfloor}\big(y_{(m_n+1-i)} - y'_{(i)}\big)^2_+.
\end{equation}
Building on the notation defined at the beginning of Section~\ref{app:conc-neighbour}, for any $k\in [K_n]$ we define 
$Y_{(n,A_{n,k})}:=(Y_{(n,(k-1)m_n+1)},\ldots, Y_{(n,km_n)})$.
Our quantity of interest can then be written as
\[
{\rm CBM}_n := \frac{2}{n}\|\Delta\bY^+\|^2_2 = \frac{2}{n}\sum_{k=1}^{K_n-1} f_n(Y_{(n,A_{n,k})},Y_{(n,A_{n,k+1})}).
\]

Similar to neighbor matching, the following proposition establishes concentration of the aforementioned quantity around the mean conditional deviance:
\begin{proposition}\label{lem:cross_bin_lemma}
Under the assumptions and definitions stated above, if $K_n\to\infty$ and $n/K_n\to\infty$, then
\[\left|{\rm CBM}_n - \EE{\textnormal{Dev}(P_{Y|Z})}\right| = {\rm o}_P(1).\]
\end{proposition}

\begin{proof}[Proof of Proposition~\ref{lem:cross_bin_lemma}]
Similar to the concentration result for neighbor matching, this proof is also split into four key steps.
\paragraph{Step 1: a concentration step.}
Our first step is a concentration result: we claim that for any $\delta \in (0,1]$, we have
\begin{equation}\label{eqn:conc-CBM_n}
    \mathbb{P}\biggl\{\biggl|{\rm CBM}_n - \EEst{{\rm CBM}_n}{\bZ}\biggr| \geq \sqrt{\frac{4\log(4/\delta)}{K_n}} \biggm| \bZ\biggr\} \leq \delta.
\end{equation}
To prove this, we start by writing
\[
{\rm CBM}_n = \frac{2}{n}\sum_{k=1}^{\lfloor K_n/2\rfloor} f_n(Y_{(n,A_{n,2k-1})},Y_{(n,A_{n,2k})}) + \frac{2}{n}\sum_{k=1}^{\lfloor (K_n-1)/2\rfloor} f_n(Y_{(n,A_{n,2k})},Y_{(n,A_{n,2k+1})}).
\]
After conditioning on $\bZ$, the first term is a sum of $\lfloor K_n/2\rfloor$ independent terms because the sets $\{A_{n,k}:k \in [K_n]\}$ are pairwise disjoint. Furthermore, each term $2f_n(Y_{(n,A_{n,2k-1})},Y_{(n,A_{n,2k})})/n$ takes values in $[0,8\cdot \lfloor m_n/2\rfloor/n] \subseteq [0,4/K_n]$. Therefore by Hoeffding's inequality, 
\begin{align*}
\mathbb{P}\Biggl\{\biggl|&\frac{2}{n}\sum_{k=1}^{\lfloor K_n/2\rfloor} f_n(Y_{(n,A_{n,2k-1})},Y_{(n,A_{n,2k})}) \\&\hspace{1cm}- \frac{2}{n}\sum_{k=1}^{\lfloor K_n/2\rfloor} \EEst{f_n(Y_{(n,A_{n,2k-1})},Y_{(n,A_{n,2k})})}{\bZ}\biggr|\geq \frac{4}{K_n}\sqrt{\frac{\lfloor K_n/2\rfloor \cdot \log(4/\delta)}{2}}\Biggm| \bZ\Biggr\}\leq \frac{\delta}{2}.
\end{align*}
We can prove a similar concentration for the second term, and together with a union bound this establishes~\eqref{eqn:conc-CBM_n}.  It follows that
\[
\bigl|{\rm CBM}_n - \EEst{{\rm CBM}_n}{\bZ}\bigr| = {\rm O}_P(K_n^{-1/2}).
\]

\paragraph{Step 2: comparing to i.i.d.\ terms.}
Since $(y,y')\mapsto f_n(y,y')$ is $4$-Lipschitz with respect to the $\ell_1$ norm on $[-1,1]^{m_n} \times [-1,1]^{m_n}$ (see Lemma~\ref{lem:f_Lipschitz}), we have for any distributions $Q,Q'$ on $[-1,1]^{m_n}\times [-1,1]^{m_n}$ that
\[
\bigl|\Ep{(Y,Y')\sim Q}{f_{n}(Y,Y')}
- \mathbb{E}_{(\tilde{Y},\tilde{Y}')\sim Q'}\bigl[f_{n}(\tilde{Y},\tilde{Y}')\bigr]\bigr|\leq 4{\rm d}_{\rm W}(Q,Q'),
\]
where ${\rm d}_{\rm W}(\cdot,\cdot)$ is the $L_1$-Wasserstein distance between distributions.
In particular, in the special case where $Q=Q_1\times\dots\times Q_{2m_n}$ and $Q'=Q_1'\times \dots \times  Q_{2m_n}'$, we have
\[
\bigl|\Ep{(Y,Y')\sim Q}{f_{n}(Y,Y')}
- \mathbb{E}_{(\tilde{Y},\tilde{Y}')\sim Q'}\bigl[f_{n}(\tilde{Y},\tilde{Y}')\bigr]\bigr|\leq 4\sum_{j=1}^{2m_n} {\rm d}_{\rm W}(Q_j,Q'_j).
\]
Next, fix any $k\in [K_n-1]$. By construction, we have
\[
\big( Y_{(n,A_{n,k})}, Y_{(n,A_{n,k+1})}\big) \mid \bZ \sim P_{Y|Z}(\cdot\mid Z_{(n,A_{n,k})})\times P_{Y|Z}(\cdot\mid Z_{(n,A_{n,k+1})}),
\]
where, abusing notation, we write 
\[
P_{Y|Z}(\cdot\mid Z_{(n,A_{n,k})}) = P_{Y|Z}(\cdot\mid Z_{(n,(k-1)m_n+1)}) \times \dots \times P_{Y|Z}(\cdot\mid Z_{(n,km_n)}),
\]
Now let $Z'_1,Z'_2,\dots\iidsim P_Z$ be drawn independently of the data. Then
\begin{align*}
    \biggl|\EEst{{\rm CBM}_n}{\bZ} &- \frac{2}{n} \sum_{k=1}^{K_n-1}\Ep{Y,Y'\iidsim (P_{Y|Z}(\cdot\mid Z'_{(K_n,k)}))^{m_n}}{f_n(Y,Y')}\biggr|\\
    &\leq  \frac{2}{n} \sum_{k=1}^{K_n-1}\Bigl|\EEst{f_n\left(Y_{(n,A_{n,k})}, Y_{(n,A_{n,k+1})}\right)}{\bZ}
    \!-\!  \Ep{Y,Y'\iidsim (P_{Y|Z}(\cdot\mid Z'_{(K_n,k)}))^{m_n}}{f_n(Y,Y')}\Bigr| \\
    &\leq \frac{8}{n}\sum_{k=1}^{K_n-1}
    \sum_{i\in A_{n,k}\cup A_{n,k+1}}{\rm d}_{\rm W}\left(P_{Y|Z}(\cdot\mid Z_{(n,i)}),P_{Y|Z}(\cdot\mid Z'_{(K_n,k)})\right).
\end{align*}

\paragraph{Step 3: concentration around the mean of conditional deviance.}

By Lemma~\ref{lem:f_div}, there exists a universal constant $C > 0$ such that
\begin{align*}
&\biggl|\frac{2}{n} \sum_{k=1}^{K_n-1}\Ep{Y,Y'\iidsim (P_{Y|Z}(\cdot\mid Z'_{(K_n,k)}))^{m_n}}{f_n(Y,Y')} - \frac{m_n}{n}\sum_{k=1}^{K_n}\textnormal{Dev}\left(P_{Y|Z}(\cdot\mid Z'_k)\right)\biggr| \\
&= \biggl|\frac{2}{n} \sum_{k=1}^{K_n-1}\Ep{Y,Y'\iidsim (P_{Y|Z}(\cdot\mid Z'_{(K_n,k)}))^{m_n}}{f_n(Y,Y')} - \frac{m_n}{n}\sum_{k=1}^{K_n}\textnormal{Dev}\left(P_{Y|Z}(\cdot\mid Z'_{(K_n,k)})\right)\biggr| \\ 
&\leq \sum_{i=1}^{K_n-1} \biggl|\frac{2}{n} \Ep{Y,Y'\iidsim (P_{Y|Z}(\cdot\mid Z'_{(K_n,k)}))^{m_n}}{f_n(Y,Y')} - \frac{m_n}{n}\textnormal{Dev}\left(P_{Y|Z}(\cdot\mid Z'_{(K_n,k)})\right)\biggr| \\
&\hspace{10.5cm}+\frac{m_n}{n}\bigl|\textnormal{Dev}\left(P_{Y|Z}(\cdot\mid Z'_{(K_n,K_n)})\right)\bigr|\\
&\leq (K_n-1) \cdot \frac{m_n}{n} \cdot \frac{C}{\sqrt{m_n}} +\frac{4m_n}{n} \leq \frac{C}{\sqrt{m_n}}+\frac{4m_n}{n}.
\end{align*}
Moreover, since $1/(K_n+1) \leq m_n/n \leq 1/K_n$, since $Z'_1,\dots,Z'_{K_n}\iidsim P_Z$, we have by the central limit theorem that
\begin{align*}
\biggl|\frac{m_n}{n}&\sum_{k=1}^{K_n}\textnormal{Dev}\left(P_{Y|Z}(\cdot\mid Z'_k)\right) - \EE{\textnormal{Dev}(P_{Y|Z})}\biggr|\\
&\leq \biggl(\frac{1}{K_n} - \frac{m_n}{n}\biggr)4K_n + \biggl|\frac{1}{K_n}\sum_{k=1}^{K_n}\textnormal{Dev}\left(P_{Y|Z}(\cdot\mid Z'_k)\right) - \EE{\textnormal{Dev}(P_{Y|Z})}\biggr|= {\rm O}_P(K_n^{-1/2}).
\end{align*}
Combining these bounds with our conclusions from Steps~1 and~2, then,
\begin{multline*}
    \left|{\rm CBM}_n - \EE{\textnormal{Dev}(P_{Y|Z})}\right|
    \leq \frac{8}{n}\sum_{k=1}^{K_n-1}
    \sum_{i\in A_{n,k}}{\rm d}_{\rm W}\left(P_{Y|Z}(\cdot\mid Z_{(n,i)}),P_{Y|Z}(\cdot\mid Z'_{(K_n,k)})\right)  \\
    + \frac{8}{n}\sum_{k=1}^{K_n-1}
    \sum_{i\in A_{n,k+1}}{\rm d}_{\rm W}\left(P_{Y|Z}(\cdot\mid Z_{(n,i)}),P_{Y|Z}(\cdot\mid Z'_{(K_n,k)})\right) + {\rm O}_P(K_n^{-1/2}) + {\rm O}(m_n^{-1/2}).
\end{multline*}

\paragraph{Step 4: the Wasserstein distance.}
Finally, we bound the remaining terms. First, 
\[
\frac{8}{n}\sum_{k=1}^{K_n-1}
    \sum_{i\in A_{n,k}}{\rm d}_{\rm W}\left(P_{Y|Z}(\cdot\mid Z_{(n,i)}),P_{Y|Z}(\cdot\mid Z'_{(K_n,k)})\right)={\rm o}_P(1), 
    \]
  by Lemma~\ref{lem:chatterjee_lemma_implies_dw} below (applied with sets $B_{n,k} = A_{n,k}$ for $k \in [K_n-1]$, and $B_{n,K_n} = \emptyset$). Similarly,
 \[
 \frac{8}{n}\sum_{k=1}^{K_n-1}
    \sum_{i\in A_{n,k+1}}{\rm d}_{\rm W}\left(P_{Y|Z}(\cdot\mid Z_{(n,i)}),P_{Y|Z}(\cdot\mid Z'_{(K_n,k)})\right) ={\rm o}_P(1),
    \]
  by Lemma~\ref{lem:chatterjee_lemma_implies_dw} again (applied with sets $B_{n,k} = A_{n,k+1}$ for $k \in [K_n-1]$, and $B_{n,K_n} = \emptyset$).  This completes the proof.
\end{proof}

\section{Proof of Theorem~\ref{thm:main-general}}\label{app:proof_validity_extension}
The argument closely follows that of Theorem~\ref{thm:main}, given in Section~\ref{sec:validity}. The only modifications are to accommodate the multivariate notions of monotonic functions and stochastic monotonicity introduced in Section~\ref{sec:extension}.

\paragraph{Step 1: deterministic properties of the $p$-value.}
Fix $\alpha \in [0,1]$ and define $p:\R^n\to[0,1]$ by
\begin{align*}
    p(\bx) = \frac{1}{2^L}\sum_{\bs\in\{\pm 1\}^L}\One{T(\bx^{\bs})\geq T(\bx)}.
\end{align*}
Then the $p$-value in~\eqref{eq:def-test} can be written as $p=p(\bX)$, and satisfies the deterministic inequality~\eqref{eqn:pvalue_deterministic_quantile}.
Moreover, by the same argument as in the proof of Theorem~\ref{thm:main}, $p(\bx)$ is nonincreasing in each $x_{i_\ell}$ and nondecreasing in each $x_{j_\ell}$, where monotonicity of functions is interpreted coordinatewise as in  Section~\ref{sec:extension}, and $(\psi_\ell)_{\ell=1}^L$ are anti-monotonic as in~\eqref{defn:monotonicity-of-psi-general}.

\paragraph{Step 2: comparison to the sharp null.}
We define the sharp null analogously to the proof of Theorem~\ref{thm:main}. Writing $P_i = P_{X|Z}(\cdot\mid Z_i)$ for $i\in [n]$, set 
\[
\bar{P}_\ell = \tfrac{1}{2}P_{i_\ell}+ \tfrac{1}{2}P_{j_\ell}
\]
for $\ell\in[L]$.  Under Assumption~\ref{asm:st_gen}, conditional on $\bY,\bZ$, for any nondecreasing function $f:\Xcal \to \R$,
\begin{align}\label{eqn:gen_stochastic_order_for_mix}
    \Ep{X\sim P_{X\mid Z_{i_\ell}}}{f(X)}\le\Ep{X\sim \bar{P}_\ell}{f(X)}, \qquad \ell\in[L],
    \end{align}
    and
 \begin{align}   
 \label{eqn:gen_stochastic_order_for_mix2}
    \Ep{X\sim \bar{P}_\ell}{f(X)} \le \Ep{X\sim P_{X\mid Z_{j_\ell}}}{f(X)}, \qquad \ell\in[L].
\end{align}

Further, conditional on $\bY,\bZ$,
\[
\bX \sim  P_1\times \dots \times P_n,
\]
and we sample
\[
\bX_{\sharp} \sim (P_\sharp)_1 \times \dots\times (P_\sharp)_n,
\]
where $(P_\sharp)_{i_\ell} = (P_\sharp)_{j_\ell} = \bar{P}_\ell$ for $\ell\in[L]$, and $(P_\sharp)_i=P_i$ for indices not in any matched pair. As in the proof of Theorem~\ref{thm:main}, for any $\bs\in\{\pm 1\}^L$,
\[
(\bX_\sharp)^\bs \eqd \bX_\sharp
\]
conditional on $\bY,\bZ$.

Combining the monotonicity of $p(\bx)$ from Step~1 with the stochastic ordering relations~\eqref{eqn:gen_stochastic_order_for_mix} and~\eqref{eqn:gen_stochastic_order_for_mix2} , and using the fact that $t\mapsto\One{t\le \alpha}$ is nonincreasing, we obtain 
\begin{align*}
   \PPst{p(\bX)\leq \alpha}{\bY,\bZ}
   &=\Ep{\bX\sim \otimes_{i=1}^n P_i}{\One{p(\bX)\le \alpha}}\\
   &\leq \Ep{\bX_\sharp\sim \otimes_{i=1}^n (P_\sharp)_i}{\One{p(\bX_\sharp)\le \alpha}}
   =\PPst{p(\bX_\sharp)\leq \alpha}{\bY,\bZ}.
\end{align*}

\paragraph{Step 3: validity under the sharp null.}
This step is identical to its counterpart in the proof of Theorem~\ref{thm:main}.

\section{Auxiliary lemmas}\label{app:lemmas}
\begin{lemma}\label{lem:deterministic_control_on_perm_pvalue}
    It holds that deterministically, for any $\alpha\in [0,1]$,
    \[
    \frac{1}{2^L}\sum_{\bs\in\{\pm 1\}^L} \One{p(\bx^{\bs}) \leq \alpha}\leq \alpha,
    \]
    where $p(\cdot)$ is as defined in \eqref{eq:pvalue-as-a-function-of-data}.
\end{lemma}
\begin{proof}
    Consider a bijection $\sigma:[2^L] \rightarrow \{-1,1\}^L$ such that $T(\bx^{\sigma(1)}) \geq \ldots \geq T(\bx^{\sigma(2^L)})$, and let $r \in \{0,1,\ldots,2^L\}$ be such that $\alpha \in \bigl[r/2^L,(r+1)/2^L\bigr)$.  Then, since we have $\sum_{k=1}^{2^L} \One{T(\bx^{\sigma(k)}) \geq T(\bx^{\sigma(j)})} \in [2^L]$  for each $j \in [2^L]$, we have deterministically that
    \begin{align*}
    \frac{1}{2^L}\sum_{\bs\in\{\pm 1\}^L} \One{p(\bx^{\bs}) \leq \alpha} 
    &=\frac{1}{2^L}\sum_{j=1}^{2^L}\One{p(\bx^{\sigma(j)})\leq \alpha}\\
    &= \frac{1}{2^L}\sum_{j=1}^{2^L} \One{\frac{1}{2^L}\sum_{k=1}^{2^L} \One{T(\bx^{\sigma(k)}) \geq T(\bx^{\sigma(j)})} \leq \alpha} \nonumber \\
    &= \frac{1}{2^L}\sum_{j=1}^{2^L} \One{\sum_{k=1}^{2^L} \One{T(\bx^{\sigma(k)}) \geq T(\bx^{\sigma(j)})} \leq r} \nonumber \\
    &\leq \frac{1}{2^L}\sum_{j=1}^{2^L} \One{\sum_{k=1}^{2^L} \One{k \leq j} \leq r} = \frac{1}{2^L}\sum_{j=1}^{2^L} \One{j \leq r} = \frac{r}{2^L} \leq \alpha,
    \end{align*}
as required.
\end{proof}

\begin{lemma}\label{lem:hellinger_lemma}
For any $p,q\in (0,1)$, we have
    $$\mathrm{H}^2\bigl(\textnormal{Ber}(p),\textnormal{Ber}(q)\bigr)\leq \frac{2(p-q)^2}{(p+q)(2-p-q)}\leq \min\biggl\{\frac{(p-q)^2}{p(1-p)},\frac{(p-q)^2}{q(1-q)}\biggr\}.$$
\end{lemma}
\begin{proof}
For any $p,q\in (0,1)$,
\begin{align*}
    \mathrm{H}^2\bigl(\textnormal{Ber}(p),\textnormal{Ber}(q)\bigr) &= (\sqrt{p}-\sqrt{q})^2+(\sqrt{1-p}-\sqrt{1-q})^2\\
    &=\frac{(p-q)^2}{(\sqrt{p}+\sqrt{q})^2}+\frac{(p-q)^2}{(\sqrt{1-p}+\sqrt{1-q})^2}\\
    &\leq \frac{(p-q)^2}{p}+\frac{(p-q)^2}{(1-p)}= \frac{(p-q)^2}{p(1-p)}.
\end{align*}
By symmetry, the bound $\frac{2(p-q)^2}{(p+q)(2-p-q)} \leq \frac{(p-q)^2}{q(1-q)}$ also holds.
\end{proof}
\begin{lemma}\label{lem:deviance-and-variation}
    For any distribution $P$ on $\R$, 
    \[2\,\Vp{X\sim P}{X}\leq \textnormal{Dev}(P)\leq 4\,\Vp{X\sim P}{X},\] where $\textnormal{Dev}(\cdot)$ is the deviance of a distribution, defined in~\eqref{defn:deviance}. Moreover, if $P$ is symmetric around its mean, then $\textnormal{Dev}(P)=4\,\Vp{X\sim P}{X}$.
\end{lemma}
\begin{proof} Let $F^{-1}$ denote the quantile function of $P$ (i.e., the generalized inverse of its CDF).
    For any $q \in (0,1)$, by monotonicity of $F$, we have
\[
\bigl\{F^{-1}(1-q)-F^{-1}(1/2)\bigr\}\cdot \bigl\{F^{-1}(1/2)-F^{-1}(q)\bigr\} \geq 0.
\]
Therefore,
    \begin{align*}
        \textnormal{Dev}(P)
        &=\Ep{q\sim \textnormal{Unif}(0,1)}{\bigl(F^{-1}(1-q)-F^{-1}(1/2)+F^{-1}(1/2)-F^{-1}(q)\bigr)^2}\\
        &\geq \Ep{q\sim \textnormal{Unif}(0,1)}{\bigl(F^{-1}(1-q)-F^{-1}(1/2)\bigr)^2+\bigl(F^{-1}(1/2)-F^{-1}(q)\bigr)^2}\\
        &=2\,\Ep{q\sim \textnormal{Unif}(0,1)}{\bigl(F^{-1}(q)-F^{-1}(1/2)\bigr)^2} \geq 2\,\Vp{q\sim\textnormal{Unif}(0,1)}{F^{-1}(q)} = 2\,\Vp{X\sim P}{X},
    \end{align*}
    since when $q\sim\textnormal{Unif}(0,1)$ we have $F^{-1}(q)\sim P$ by construction. For the upper bound, writing $\mu = \Ep{X\sim P}{X}$, 
        \begin{align}
        \textnormal{Dev}(P)
        \notag&=\Ep{q\sim \textnormal{Unif}(0,1)}{\bigl(F^{-1}(1-q)-\mu+\mu-F^{-1}(q)\bigr)^2}\\
        \label{eqn:ineq_4var}&\leq \Ep{q\sim \textnormal{Unif}(0,1)}{2\bigl(F^{-1}(1-q)-\mu\bigr)^2 + 2\bigl(F^{-1}(q) - \mu\bigr)^2}\\
        \notag&=4\,\Ep{q\sim \textnormal{Unif}(0,1)}{\bigl(F^{-1}(q)-\mu\bigr)^2} = 4\,\Vp{q\sim\textnormal{Unif}(0,1)}{F^{-1}(q)} = 4\,\Vp{X\sim P}{X},
    \end{align}
    
    Next we turn to the special case that $P$ is symmetric around its mean. In this case, symmetry of the distribution around its mean ensures that $F^{-1}(1-q)-\mu = \mu-F^{-1}(q)$ holds for almost every $q\in[0,1]$ (more precisely, this can fail only at discontinuities of the CDF of $P$). Therefore,
    \[\bigl(F^{-1}(1-q)-\mu+\mu-F^{-1}(q)\bigr)^2 = 2\bigl(F^{-1}(1-q)-\mu\bigr)^2 + 2\bigl(F^{-1}(q) - \mu\bigr)^2\]
    for almost every $q\in[0,1]$, and so the inequality in the step~\eqref{eqn:ineq_4var} above is in fact an equality. This proves $\textnormal{Dev}(P)=4\,\Vp{X\sim P}{X}$, as desired.
\end{proof}
\begin{lemma}\label{lem:a-median-result}
Suppose $a\in\R^m$ and $b\in\R^n$. Let $c=(a,b)\in\R^{m+n}$ be the concatenation of the two vectors. If $\textnormal{Med}(c)< \textnormal{Med}(a)$, then $\textnormal{Med}(c)\geq\textnormal{Med}(b)$; similarly, if $\textnormal{Med}(c)>\textnormal{Med}(a)$, then $\textnormal{Med}(c)\leq\textnormal{Med}(b)$.
\end{lemma}
\begin{proof}
    We will assume that $\textnormal{Med}(c)< \textnormal{Med}(a)$ and $\textnormal{Med}(c) < \textnormal{Med}(b)$, and will show that this leads to a contradiction. (The second claim is proved analogously and so we omit the proof.) 

We first calculate
    \[\bigl|\{i \in [m]:a_i\leq  \textnormal{Med}(c)\}\bigr|\leq\bigl|\{i \in [m]:a_i < \textnormal{Med}(a)\}\bigr|\leq \begin{cases} \frac{m}{2}, &\textnormal{ if $m$ is even,}\\
    \frac{m-1}{2}, & \textnormal{ if $m$ is odd},\end{cases}\]
    where the first step holds since we have assumed $\textnormal{Med}(c)< \textnormal{Med}(a)$, and the second step uses the definition of the median. Similarly,
    \[\bigl|\{j \in [n]:b_j \leq  \textnormal{Med}(c)\}\bigr|\leq \bigl|\{j \in [n]:b_j < \textnormal{Med}(b)\}\bigr|\leq \begin{cases} \frac{n}{2}, &\textnormal{ if $n$ is even,}\\
    \frac{n-1}{2}, & \textnormal{ if $n$ is odd}.\end{cases}\]
    On the other hand, again by definition of the median,
    \begin{multline*}\bigl|\{i \in [m]:a_i\leq  \textnormal{Med}(c)\}\bigr| + \bigl|\{j \in [n]:b_j \leq  \textnormal{Med}(c)\}\bigr| \\= \bigl|\{k \in [m+n]:c_k \leq \textnormal{Med}(c)\}\bigr| \geq \begin{cases}\frac{m+n}{2}, & \textnormal{ if $m+n$ is even}, \\ \frac{m+n+1}{2}, & \textnormal{ if $m+n$ is odd}.\end{cases}\end{multline*}

So far, then, we have shown that
\[\frac{m-\one{\textnormal{$m$ is odd}}}{2}+ \frac{n-\one{\textnormal{$n$ is odd}}}{2} \geq \frac{m+n+\one{\textnormal{$m+n$ is odd}}}{2}.\]
If either $m$ or $n$ is odd (or both), this immediately yields a contradiction. The remaining case is that $m$ and $n$ are both even, and
\[\bigl|\{i \in [m]:a_i\leq  \textnormal{Med}(c)\}\bigr|=\bigl|\{i \in [m]:a_i < \textnormal{Med}(a)\}\bigr| = \frac{m}{2}\]
and
\[\bigl|\{j \in [n]:b_j\leq  \textnormal{Med}(c)\}\bigr|=\bigl|\{j \in [n]:b_j < \textnormal{Med}(b)\}\bigr| = \frac{n}{2}.\]
Letting $a_{(1)}\leq \dots \leq a_{(m)}$ denote the order statistics of $a$, and similarly letting $b_{(1)}\leq \dots\leq b_{(n)}$ denote the order statistics of $b$, we therefore have 
\[a_{(m/2)} \leq \textnormal{Med}(c) < a_{(m/2+1)}\textnormal{ and }
b_{(n/2)} \leq \textnormal{Med}(c) < b_{(n/2+1)}.\]
In particular, writing $c_{(1)}\leq \dots\leq c_{(n+m)}$ for the order statistics of $c$, this means that
\[c_{((m+n)/2)} = \max\big\{ a_{(m/2)},b_{(n/2)}\big\}\textnormal{ and }
c_{((m+n)/2+1)} =\min\{ a_{(m/2+1)},b_{(n/2+1)}\big\}.\]
By definition of the median, then,
\begin{multline*}\textnormal{Med}(c) = \frac{c_{((m+n)/2)} + c_{((m+n)/2+1)}}{2} = \frac{\max\big\{ a_{(m/2)},b_{(n/2)}\big\} + \min\{ a_{(m/2+1)},b_{(n/2+1)}\big\}}{2}\\
\geq \min\left\{\frac{a_{(m/2)} + a_{(m/2+1)}}{2}, \frac{b_{(n/2)} + b_{(n/2+1)}}{2}\right\}
= \min\left\{\textnormal{Med}(a),\textnormal{Med}(b)\right\}.
\end{multline*}
This contradicts our assumption, and completes the proof.
    \end{proof}
\begin{lemma}\label{lem:chatterjee_lemma_implies_dw} Let
$P_{(Y,Z)}$ be a distribution on $[-1,1]\times\R$, and let $Z_1,Z_1',Z_2,Z_2',\ldots$ be independent draws from the corresponding marginal $P_Z$ on $\R$. Let $(K_n)$ be a sequence of positive integers with $K_n\to\infty$ as $n\to \infty$, and for each $n$, let $B_{n,1},\dots,B_{n,K_n}\subseteq [n]$ be a collection of disjoint sets such that\footnote{If $B_{n,k} = \emptyset$ for some $k \in [K_n]$, then we interpret $\max_{i\in B_{n,k}} \bigl|i/n - k/K_n\bigr| = 0$.} 
\[
\max_{k \in [K_n]} \max_{i\in B_{n,k}} \left|\frac{i}{n} - \frac{k}{K_n}\right| \to 0.
\] Then, as $n\to\infty$
\[
\mathbb{E}\Biggl[\frac{1}{n}\sum_{k=1}^{K_n}\sum_{i\in B_{n,k}}{\rm d}_{\rm W}\bigl(P_{Y|Z}(\cdot\mid Z_{(n,i)}) , P_{Y|Z}(\cdot\mid Z'_{(K_n,k)})\bigr)\Biggr] \to 0.
\]
\end{lemma}
\begin{remark}
The intuition is that, at each value of $n$, we should have $\frac{i}{n}\approx \frac{k}{K_n}$ and therefore $Z_{(n,i)}\approx Z'_{(K_n,k)}$, for all $k$ and all $i\in B_{n,k}$.
\end{remark}
\begin{proof}
Let $F_Z^{-1}$ be the quantile function corresponding to $P_Z$, so that if $U\sim\textnormal{Unif}[0,1]$ then $F_Z^{-1}(U)\sim P_Z$. Now define $\tilde{P}_{Y|U}(\cdot\mid u) = P_{Y|Z}\bigl(\cdot\mid F_Z^{-1}(u)\bigr)$, i.e., the conditional distribution of~$Y$ given $U$ where $Z = F_Z^{-1}(U)$ and we sample $Y$ conditional on $Z$ from $P_{Y\mid Z}$. Since $F_Z^{-1}$ is monotone, it suffices to prove that
\[
\mathbb{E}\Biggl[\frac{1}{n}\sum_{k=1}^{K_n}\sum_{i\in B_{n,k}}{\rm d}_{\rm W}\bigl(\tilde{P}_{Y|U}(\cdot\mid U_{(n,i)}) , \tilde{P}_{Y|U}(\cdot\mid V_{(K_n,k)})\bigr)\Biggr] \to 0,\]
where $U_1,V_1,U_2,V_2,\dots$ are independent and identically distributed Unif$[0,1]$ random variables. 
Writing $F_{Y|U}$ for the distribution function corresponding to the conditional distribution $\tilde{P}_{Y|U}$, we have
\begin{align}\label{eqn:wasserstein-breakdown}
    &\mathbb{E}\Biggl[\frac{1}{n}\sum_{k=1}^{K_n}\sum_{i\in B_{n,k}}{\rm d}_{\rm W}\bigl(\tilde{P}_{Y|U}(\cdot\mid U_{(n,i)}) , \tilde{P}_{Y|U}(\cdot\mid V_{(K_n,k)})\bigr)\Biggr] \nonumber\\
    &\hspace{1.7cm}= \mathbb{E}\Biggl[\frac{1}{n}\sum_{k=1}^{K_n}\sum_{i\in B_{n,k}}\int_{-1}^1 \bigl|F_{Y|U}(t\mid U_{(n,i)}) - F_{Y|U}(t\mid V_{(K_n,k)})\bigr|\;\mathsf{d}t\Biggr] \nonumber\\
    &\hspace{1.7cm}= \int_{-1}^1 \mathbb{E}\Biggl[\frac{1}{n}\sum_{k=1}^{K_n}\sum_{i\in B_{n,k}}\bigl|F_{Y|U}(t\mid U_{(n,i)}) - F_{Y|U}(t\mid V_{(K_n,k)})\bigr|\Biggr]\;\mathsf{d}t ,
\end{align}
where the last step follows by Fubini's theorem.

Now, we define some additional notation: for $n\in \mathbb{N}$ and $j\in [n]$, let
\begin{equation}\label{defn:V-same-order-as-U}
    V_{n,j}=\begin{cases}
       V_{(K_n,k)} & \text{if } U_j = U_{(n,i)}~\text{for some } i\in B_{n,k},\\
       U_j & \text{if } U_j = U_{(n,i)}~\text{for some } i\in[n] \setminus (B_{n,1}\cup\dots\cup B_{n,K_n}).
    \end{cases}
\end{equation}
Then for every $t \in [0,1]$ and $n \in \mathbb{N}$, 
\[
\sum_{k=1}^{K_n} \sum_{i\in B_{n,k}} \bigl|F_{Y|U}(t\mid U_{(n,i)}) - F_{Y|U}(t\mid V_{(K_n,k)})\bigr| = \sum_{j=1}^n |F_{Y|U}(t\mid U_j) - F_{Y|U}(t\mid V_{n,j})|.
\]
Moreover, $|F_{Y|U}(t\mid U_j) - F_{Y|U}(t\mid V_{n,j})|$ is identically distributed for each $j=1,\dots,n$ (when $n$ is fixed). Hence, 
\begin{align*}
    \mathbb{E}\Biggl[\sum_{k=1}^{K_n} \sum_{i\in B_{n,k}} \bigl|F_{Y|U}(t\mid U_{(n,i)}) - F_{Y|U}(t\mid V_{(K_n,k)})\bigr|\Biggr]&=\EE{\frac{1}{n}\sum_{j=1}^n \bigl|F_{Y|U}(t\mid U_j) - F_{Y|U}(t\mid V_{n,j})\bigr|}\\ &= \mathbb{E}\bigl[\bigl|F_{Y|U}(t\mid U_1) - F_{Y|U}(t\mid V_{n,1})\bigr|\bigr].
\end{align*}
By Lemma~\ref{lem:chatterjee_lemma_extension} below,
\[
\bigl|F_{Y|U}(t\mid U_{1}) - F_{Y|U}(t\mid V_{n,1})\bigr| = \mathrm{o}_P(1).
\]
Since $F_{Y\mid U}(t\mid \cdot)$ is a bounded function, we deduce that
\[
\mathbb{E}\Biggl[\frac{1}{n}\sum_{k=1}^{K_n}\sum_{i\in B_{n,k}}\bigl|F_{Y|U}(t\mid U_{(n,i)}) - F_{Y|U}(t\mid V_{(K_n,k)})\bigr|\Biggr] \to 0
\]
for every $t \in [-1,1]$.  The result now follows by \eqref{eqn:wasserstein-breakdown} and the dominated convergence theorem.
\end{proof}

\begin{lemma}\label{lem:chatterjee_lemma_extension}
 Let $U_1,V_1,U_2,V_2,\dots\iidsim\textnormal{Unif}[0,1]$, and let $(K_n)$ be a sequence of positive integers with $K_n\to\infty$ as $n\to\infty$. Furthermore, for each $n$, let $B_{n,1},\dots,B_{n,K_n}\subseteq \{1,\dots,n\}$ be a collection of disjoint sets such that 
\[
\max_{k \in [K_n]} \max_{i\in B_{n,k}} \biggl|\frac{i}{n} - \frac{k}{K_n}\biggr| \to 0.
\]
Then, for any Borel measurable $g:[0,1]\to\R$, we have $|g(U_1) - g(V_{n,1})| = \mathrm{o}_P(1)$,
where $V_{n,1}$ is defined as in \eqref{defn:V-same-order-as-U}.
\end{lemma}

\begin{proof}
Let $\epsilon > 0$.  For any Borel measurable $h:[0,1]\to\R$, we can write
\begin{align}\label{eq:decomposition-lusin}
&\PP{|g(U_1) - g(V_{n,1})|>\epsilon}\nonumber\\
&\hspace{2cm}\leq \PP{|h(U_1) - h(V_{n,1})|>\epsilon}
+ \PP{g(U_1)\neq h(U_1)} + \PP{g(V_{n,1})\neq h(V_{n,1})}.\end{align}
By Lusin's theorem (restated for convenience in Theorem~\ref{lem:lusin-theorem}), we can find a continuous function $h:[0,1]\to\R$ such that
\[
\Pp{U\sim {\rm Unif}[0,1]}{g(U)\neq h(U)}\leq \frac{\epsilon}{4}.
\]
Next, for each $k=1,\dots,K_n$, let $\mathcal{E}_{n,k}$ denote the event that $U_1 = U_{(n,i)}$ for some $i\in B_{n,k}$, which has probability $\frac{|B_{n,k}|}{n}$.  Further, define $\mathcal{E}_{n,*}$ to be the event that $U_1 = U_{(n,i)}$ for some $i\in[n]\setminus (B_{n,1}\cup\dots\cup B_{n,K_n})$. By construction of $V_{n,1}$, on the event $\mathcal{E}_{n,k}$ we have $V_{n,1} = V_{(K_n,k)}$, and on $\mathcal{E}_{n,*}$ we have $V_{n,1} = U_1$. 
Therefore, letting $\pi_n$ denote a permutation of $[K_n]$ such that $V_{(K_n,k)} = V_{\pi_n(k)}$ for $k \in [K_n]$, we have
\begin{align*}
    \PP{g(V_{n,1})\neq h(V_{n,1})}
    &=  \sum_{k=1}^{K_n}\PP{\mathcal{E}_{n,k}, \  g(V_{(K_n,k)})\neq h(V_{(K_n,k)}} + \PP{\mathcal{E}_{n,*}, g(U_1)\neq h(U_1)}\\
    &\stackrel{\rm (*)}{=} \sum_{k=1}^{K_n}\PP{\mathcal{E}_{n,k}}\cdot \PP{g(V_{(K_n,k)})\neq h(V_{(K_n,k)}} + \PP{\mathcal{E}_{n,*},g(U_1)\neq h(U_1)}\\
    &\leq \sum_{k=1}^{K_n} \frac{|B_{n,k}|}{n}\PP{g(V_{(K_n,k)})\neq h(V_{(K_n,k)}} + \PP{g(U_1)\neq h(U_1)}\\
    &=\sum_{k=1}^{K_n} \frac{|B_{n,k}|}{n}\PP{g(V_{\pi_n(k)})\neq h(V_{\pi_n(k)})} + \PP{g(U_1)\neq h(U_1)}\\
    &=\sum_{k=1}^{K_n} \frac{|B_{n,\pi_n^{-1}(k)}|}{n}\PP{g(V_k)\neq h(V_k)} + \PP{g(U_1)\neq h(U_1)}\\
    &\leq \sum_{k=1}^{K_n} \frac{|B_{n,k}|}{n} \cdot \frac{\epsilon}{4} + \frac{\epsilon}{4} \leq \frac{\epsilon}{2},
\end{align*}
where the step marked $\rm (*)$ holds because $\mathcal{E}_{n,k}$ depends only on $(U_1,\dots,U_n)$, while $V_{(K_n,k)}$ depends only on $(V_1,\dots,V_{K_n})$, and therefore $\mathbbm{1}_{\mathcal{E}_{n,k}} \independent V_{(K_n,k)}$, while the penultimate step holds since $U_1\sim\textnormal{Unif}[0,1]$, and $V_k\sim\textnormal{Unif}[0,1]$ for each $k$.
Finally, let $\widehat{F}_{U,n} = \frac{1}{n}\sum_{i=1}^n\delta_{U_i}$ and $\widehat{F}_{V,K_n} = \frac{1}{n}\sum_{i=1}^{K_n}\delta_{V_i}$ be the empirical distribution of $U_1,\ldots,U_n$ and $V_1,\ldots,V_{K_n}$ respectively. We may assume that $U_1,\dots,U_n$ are distinct, since this holds almost surely, so if $i \in [n]$ is such that $U_1 = U_{(n,i)}$, then we have 
\[
\widehat{F}_{U,n}(U_1) = \frac{i}{n}.
\]
Let $F_{\textnormal{Unif}}$ denote the distribution function of the $\textnormal{Unif}[0,1]$ distribution, i.e., $F_{\textnormal{Unif}}(t)=t$ for $t\in[0,1]$.  If $i$ further satisfies $i\in B_{n,k}$ for some $k\in [K_n]$, then $V_{n,1} = V_{(K_n,k)}$, so
\begin{align*}
|U_1 - V_{n,1}| &\leq \biggl|U_1 - \frac{i}{n}\biggr| + \biggl|V_{n,1} - \frac{k}{K_n}\biggr| + \biggl|\frac{i}{n} - \frac{k}{K_n}\biggr| \\
&\leq \bigl\|\widehat{F}_{U,n} - F_{\textnormal{Unif}}\bigr\|_\infty + \bigl\|\widehat{F}_{V,K_n} - F_{\textnormal{Unif}}\bigr\|_\infty + \max_{k\in[K_n]}\max_{i\in B_{n,k}}\biggl|\frac{i}{n} - \frac{k}{K_n}\biggr|.
\end{align*}
If instead we have $i \in[n]\setminus (B_{n,1}\cup\dots\cup B_{n,K_n})$, then this bound holds trivially, since in that case $V_{n,1} = U_1$ by definition.  Since the first two terms are each $\mathrm{o}_P(1)$ by the Glivenko--Cantelli theorem, while the last term $\max_{k\in[K_n]}\max_{i\in B_{n,k}}\biggl|\frac{i}{n} - \frac{k}{K_n}\biggr|$ is vanishing by assumption, it follows that $V_{n,1} \stackrel{P}{\to} U_1$, so by the continuous mapping theorem, $\PP{|h(U_1) - h(V_{n,1})|>\epsilon} \leq \epsilon/4$ for sufficiently large $n$.  The result therefore follows from~\eqref{eq:decomposition-lusin}.
\end{proof}
\begin{lemma}\label{lem:f_Lipschitz}
    The function $f_n$ from~\eqref{Eq:fn} satisfies the Lipschitz property
    \[
    \bigl|f_n(y,y') - f_n(\tilde{y},\tilde{y}')\bigr|\leq 4\|y-\tilde{y}\|_1 + 4\|y'-\tilde{y}'\|_1 
    \]
    for $(y,y'),(\tilde{y},\tilde{y}') \in [-1,1]^{m_n} \times [-1,1]^{m_n}$.
\end{lemma}
\begin{proof}[Proof of Lemma~\ref{lem:f_Lipschitz}]
    Since $t\mapsto t^2_+$ is $4$-Lipschitz on $[-2,2]$, we have for each $i=1,\ldots,\lfloor m/2 \rfloor$ that
    \[\left|(y_{(m_n+1-i)} - y'_{(i)})^2_+ - (\tilde{y}_{(m_n+1-i)} - \tilde{y}'_{(i)})^2_+\right|\leq 4|y_{(m_n+1-i)} - \tilde{y}_{(m_n+1-i)}|
    + 4|y'_{(i)} - \tilde{y}'_{(i)}|.\]
    Therefore, 
    \[
    \left|f_n(y,y') - f_n(\tilde{y},\tilde{y}')\right|
    \leq 4\sum_{i=1}^{\lfloor m_n/2\rfloor}|y_{(m_n+1-i)} - \tilde{y}_{(m_n+1-i)}| + 4\sum_{i=1}^{\lfloor m_n/2\rfloor} |y'_{(i)} - \tilde{y}'_{(i)}|.
    \]
    Writing $y_{()} = (y_{(1)},\dots,y_{(m_n)})$ for the sorted vector of $y$, and same for $y'_{()}$, $\tilde{y}_{()}$, and $\tilde{y}'_{()}$, we deduce that
    \[
    \left|f_n(y,y') - f_n(\tilde{y},\tilde{y}')\right|
    \leq 4\|y_{()} - \tilde{y}_{()}\|_1 + 4\|y'_{()} - \tilde{y}'_{()}\|_1.
    \]
    Finally, we write $\hat{P}_{y}:=\sum_{i=1}^{m_n} \delta_{y_i}$ and $\hat{P}_{\tilde{y}}:=\sum_{i=1}^{m_n} \delta_{\tilde{y}_i}$. Then, from the definition of the $L_1$-Wasserstein distance ${\rm d}_{\rm W}(\cdot,\cdot)$, it holds that
    \[
    \|y_{()} - \tilde{y}_{()}\|_1 = {\rm d}_{\rm W}(\hat{P}_{y},\hat{P}_{\tilde{y}})\leq \|y - \tilde{y}\|_1,\]
     and a similar inequality also holds for $y',\tilde{y}'$. This completes the proof.
\end{proof}
\begin{lemma}\label{lem:f_div}
    Let $Q$ be any distribution on $[-1,1]$, and let $A=(A_1,\dots,A_{m_n})\sim Q^{m_n}$ and $B=(B_1,\dots,B_{m_n})\sim Q^{m_n}$. Then, with $f_n$ as defined in~\eqref{Eq:fn}, there exists a univeral constant $C > 0$ such that
    \[
    \biggl|\frac{2}{m_n}\EE{f_n(A,B)} - \textnormal{Dev}(Q)\biggr| \leq \frac{C}{\sqrt{m_n}}.
    \]
\end{lemma}
\begin{proof}[Proof of Lemma~\ref{lem:f_div}]
First let $A_{(1)} \leq \dots \leq A_{(m_n)}$ and $B_{(1)}\leq \dots \leq B_{(m_n)}$ be the order statistics of $A$ and $B$, and write $A_{()} = (A_{(1)},\dots,A_{(m_n)})$ and $B_{()} = (B_{(1)},\dots,B_{(m_n)})$ for the corresponding sorted vectors. Then
\[
f_n(A,B) = f_n(A_{()},B_{()})
\]
since $(y,y') \mapsto f_n(y,y')$ is invariant to permutations of $y$, and invariant to permutations of $y'$. 

Next, let $F_Q^{-1}$ denote the quantile function corresponding to $Q$. Let $U_1,\dots,U_{m_n}\iidsim \textnormal{Unif}[0,1]$, so we can assume without loss of generality that $A_i = F_Q^{-1}(U_i)$ for $i \in [m_n]$. Define $A' = (A_1',\ldots,A_{m_n}')$ by $A'_i = F_Q^{-1}(1-U_i)$ for $i \in [m_n]$, and let $A'_{()}$ denote the corresponding sorted vector. Then by Lemma~\ref{lem:f_Lipschitz},
\begin{align}\label{eq:f_n-split-1}
    \bigl|\EE{f_n(A,B)} - \EE{f_n(A,A')}\bigr|&\leq \EE{\left|f_n(A,B) - f_n(A,A')\right|}\nonumber\\
    &= \EE{\left|f_n(A_{()},B_{()}) - f_n(A_{()},A'_{()})\right|} \leq 4\EE{\|A'_{()}-B_{()}\|_1}.
\end{align}
Next, let $U_{(1)}\leq \dots \leq U_{(m_n)}$ be the order statistics of $U_1,\ldots,U_{m_n}$. Then $A_{(i)} = F_Q^{-1}(U_{(i)})$ and $A'_{(i)} = F_Q^{-1}(1-U_{(m_n+1-i)})$ for each $i \in [m_n]$. Now, for $i \in [m_n]$, define the event
\[
\mathcal{E}_i:=\begin{cases}
    \{U_{(i)} < 1/2\} & \text{if } i \leq \lfloor m_n/2 \rfloor,\\
    \emptyset &\text{if } i=\lceil m_n/2\rceil\, \text{when } m_n/2\notin \mathbb{N},\\
    \{U_{(i)} \geq 1/2\} &\text{if } i\geq \lceil m_n/2\rceil+1.
\end{cases}
\]
Thus, it follows that 
\begin{align*}
    f_n(A',A) = \sum_{i=1}^{\lfloor m_n/2\rfloor} (A'_{(m_n+1-i)} - A_{(i)})^2_+ = \sum_{i=1}^{\lfloor m_n/2\rfloor}
\bigl\{F_Q^{-1}(1-U_{(i)}) - F_Q^{-1}(U_{(i)})\bigr\}^2\one{\mathcal{E}_i}
\end{align*}
and
\begin{align*}
f_n(A,A') &= \sum_{i=1}^{\lfloor m_n/2\rfloor} (A_{(m_n+1-i)} - A'_{(i)})^2_+ \\
&= \sum_{i=1}^{\lfloor m_n/2\rfloor}
\bigl\{F_Q^{-1}(U_{(m_n+1-i)}) - F_Q^{-1}(1-U_{(m_n+1-i)})\bigr\}^2\one{\mathcal{E}_{m_n+1-i}}\\
&=\sum_{i=\lceil m_n/2\rceil+1}^{m_n}
\bigl\{F_Q^{-1}(1-U_{(i)}) - F_Q^{-1}(U_{(i)})\bigr\}^2\one{\mathcal{E}_i}.
\end{align*}
Moreover, since $(A,A') \stackrel{d}{=} (A',A)$, we have $f_n(A,A') \stackrel{d}{=} f_n(A',A)$.
This further implies that
\begin{align*}
    \EE{f_n(A,A')}&=\EE{\frac{f_n(A',A)+f_n(A,A')}{2}}
  =\frac{1}{2}\, \mathbb{E}\biggl[ \sum_{i=1}^ {m_n} \bigl\{F_Q^{-1}(1-U_{(i)}) - F_Q^{-1}(U_{(i)})\bigr\}^2\cdot\one{\mathcal{E}_i}\biggr].
\end{align*}
By definition, $\mathbb{E}\bigl[\bigl\{F_Q^{-1}(U_i) - F_Q^{-1}(1-U_i)\bigr\}^2\bigr] = \textnormal{Dev}(Q)$ for each $i\in [m_n]$. Hence,
\begin{align}\label{eqn:f_n-split-2}
    \biggl|\EE{f_n(A,A')}& -\frac{m_n}{2}\textnormal{Dev}(Q)\biggr|\nonumber\\
    &=\biggl| \EE{f_n(A,A')} -  \frac{1}{2} \mathbb{E}\biggl[\sum_{i=1}^{m_n} \bigl\{F_Q^{-1}(1-U_i) - F_Q^{-1}(U_i)\bigr\}^2\biggr]\biggr|\nonumber\\
    &=\biggl| \EE{f_n(A,A')} -  \frac{1}{2} \mathbb{E}\biggl[\sum_{i=1}^{m_n} \bigl\{F_Q^{-1}(1-U_{(i)}) - F_Q^{-1}(U_{(i)})\bigr\}^2\biggr]\biggr| \leq  2\, \mathbb{E}\biggl[\sum_{i=1}^{m_n}\one{\mathcal{E}_i^c}\biggr],
\end{align}
where the last step follows by recalling that $Q$ is supported on $[-1,1]$.
Next, we define $I_1 := \sum_{i=1}^{m_n}\one{\{U_i<1/2\}}$.
Then by definition of the events $\mathcal{E}_i$, we have
\[\sum_{i=1}^{\lfloor m_n/2\rfloor}\one{\mathcal{E}_i^c} = \sum_{i=1}^{\lfloor m_n/2\rfloor}\one{U_{(i)} \geq 1/2} = \lfloor m_n/2\rfloor - \sum_{i=1}^{\lfloor m_n/2\rfloor}\one{U_{(i)} < 1/2} = \max\{\lfloor m_n/2\rfloor - I_1, 0\},\]
and similarly
\[\sum_{i=\lceil m_n/2\rceil + 1}^{m_n}\one{\mathcal{E}_i^c} = \sum_{i=\lceil m_n/2\rceil + 1}^{m_n}\one{U_{(i)} < 1/2} = \max\left\{I_1 - \lceil m_n/2\rceil, 0\right\}.\]
Therefore,
\begin{multline*}
\sum_{i=1}^{m_n}\one{\mathcal{E}_i^c}\leq 1 + \sum_{i=1}^{\lfloor m_n/2\rfloor}\one{\mathcal{E}_i^c} + \sum_{i=\lceil m_n/2\rceil+1}^{m_n}\one{\mathcal{E}_i^c}\\= 1+\max\{\lfloor m_n/2\rfloor - I_1, 0\}  + \max\left\{I_1 - \lceil m_n/2\rceil, 0\right\} \leq 1 + \left|I_1 - \frac{m_n}{2}\right|.
\end{multline*}
Since $I_1 = \sum_{i=1}^{m_n}\one{\{U_i<1/2\}}\sim \textnormal{Bin}(m_n,1/2)$ and $m_n\in \mathbb{N}$, by the Cauchy--Schwarz inequality we further have that
\begin{align}\label{eqn:f_n-split-3}
    \mathbb{E}\biggl[\sum_{i=1}^{m_n}\one{\mathcal{E}_i^c}\biggr] \leq 1+ \EE{\biggl|I_1-\frac{m_n}{2}\biggr|}\leq 1+\sqrt{\VV{I_1}} = 1 + \frac{\sqrt{m_n}}{2}\leq \frac{3\sqrt{m_n}}{2}.
\end{align}
Next we turn to bounding $\|A'_{()}-B_{()}\|_1$.
Write $\widehat{P}_{A'}$ and $\widehat{P}_B$ for the empirical distributions of $A'$ and $B$, respectively. Then
\[\|A'_{()}-B_{()}\|_1 = m_n{\rm d}_{\rm W}(\widehat{P}_{A'},\widehat{P}_B) \leq m_n{\rm d}_{\rm W}(\widehat{P}_{A'},Q) + m_n{\rm d}_{\rm W}(\widehat{P}_B,Q).\]
By bounds on the empirical $1$-Wasserstein distance \citep[Theorem~3.1]{lei2020}, we have
\[
\EE{{\rm d}_{\rm W}(\widehat{P}_{A'},Q)}\leq C'm_n^{-1/2}
\]
for a universal constant $C' > 0$, and same for $\widehat{P}_B$.  We conclude from this together with~\eqref{eq:f_n-split-1}, \eqref{eqn:f_n-split-2} and~\eqref{eqn:f_n-split-3} that
\[
\Bigl|\EE{f_n(A,B)} - \frac{m_n}{2}\textnormal{Dev}(Q)\Bigr| \leq 8C'\sqrt{m_n} + 3\sqrt{m_n},
\]
as required.
\end{proof}
\begin{lemma}[A version of Lusin's theorem]\label{lem:lusin-theorem}
    Let $f:\RR\to\RR$ be a measurable function and let $\nu$ be a probability measure on $\RR$. Then, given any $\epsilon>0$, there is a compactly supported continuous function $g:\RR\to\RR$ such that $\nu\{x:f(x)\neq g(x)\}<\epsilon$.
\end{lemma}
\begin{proof}
    The proof follows by combining Theorem~2.18 and Theorem~2.24 of \citet{rudin1987real}.
\end{proof}

\section{Additional experiments}\label{app:additional_experiments}

In this section, we evaluate the proposed \texttt{PairSwap-ICI} testing procedure alongside several well-established benchmarks from the literature in the following simulation settings.  Let $Y_1,Z_1,\ldots,Y_n,Z_n\iidsim\mathcal{N}(0,1)$. Writing $\bY=(Y_1,\ldots,Y_n)$ and $\bZ=(Z_1,\ldots,Z_n)$, we assume that $\bX = (X_1,\ldots,X_n)\mid(\bY,\bZ)$ follows one of the following models:
\begin{enumerate}[(1)]
    \item $\bX = \beta_n\,\bY + \bZ + \beps$,
    \item $\bX = \beta_n\,\bY^2 + \bZ + \beps$,
    \item $\bX = \beta_n\,\bY^2 + \bZ + \bzeta$,
    \item $\bX = \beta_n\,\bY + \bZ + 0.4\cdot \sin (2\pi\bZ)+ \beps$,
\end{enumerate}
where $\epsilon_1,\ldots,\epsilon_n\iidsim\mathcal{N}(0,1)$ and $\zeta_1,\ldots,\zeta_n\iidsim t_4$ are generated independently of $(\bY,\bZ)$, and we write $\beps=(\epsilon_1,\ldots,\epsilon_n)$ and $\bzeta=(\zeta_1,\ldots,\zeta_n)$. In all cases, $\beta_n=0$ corresponds to the isotonic CI null $H_0^{\mathrm{ICI}}$, while $\beta_n>0$ represents a deviation from the null.

We consider the following tests of conditional independence:
\begin{itemize}
    \item The \texttt{PairSwap-ICI} test with the plug-in matching $\hat{M}$ and the weight matrix $\hat{W}$, as defined in \eqref{def:plug-in-matching}.  With the linear kernel $\psi(x,x')=x-x'$, the weight matrix $\hat{W}= (\hat{W}_{ij})_{i,j \in [n]}$ is computed using estimates $\hat{E}_{ij}$ of $\EEst{\psi(X_i, X_j)}{\bY, \bZ}$, constructed via sample splitting.
    
    \item The \textbf{Generalized Covariance Measure} (GCM) test, introduced by \citet{shah2020hardness}. Define
    \[
    R(X,Y,Z):=\bigl(X-\hat{\mathbb{E}}(X\mid Z)\bigr)\bigl(Y-\hat{\mathbb{E}}(Y\mid Z)\bigr),
    \]
    where $\hat{\mathbb{E}}(X\mid Z)$ and $\hat{\mathbb{E}}(Y\mid Z)$ are estimates of the corresponding conditional means obtained via sample splitting. We then compute the $p$-value
    \[
    p_{\rm GCM}
    =
    2-2\cdot\Phi\!\Biggl(
    \frac{\sqrt{n}\cdot \frac{1}{n}\bigl|\sum_{i=1}^n R(X_i,Y_i,Z_i)\bigr|}
    {\bigl\{\frac{1}{n}\sum_{i=1}^n R^2(X_i,Y_i,Z_i)
    -\bigl(\frac{1}{n}\sum_{i=1}^n R(X_i,Y_i,Z_i)\bigr)^2\bigr\}^{1/2}}
    \Biggr).
    \]

    \item The \textbf{Projected Covariance Measure} (PCM) test, introduced by \citet{lundborg2022projected}. At a high level, PCM follows a structure similar to GCM, but operates on
    \[
    L(X,Y,Z)
    =
    \bigl\{f(X,Z)-\hat{\mathbb{E}}\bigl(f(X,Z)\mid Z\bigr)\bigr\}
    \bigl(Y-\hat{\mathbb{E}}(Y\mid Z)\bigr),
    \]
    where $f$ is a data-adaptive function selected via sample splitting to capture departures from conditional mean independence. As in GCM, the $p$-value is computed as
    \[
    p_{\rm PCM}
    =
    2-2\cdot
    \Phi\!\Biggl(
    \frac{\sqrt{n}\cdot \frac{1}{n}\bigl|\sum_{i=1}^n L(X_i,Y_i,Z_i)\bigr|}
    {\bigl\{\frac{1}{n}\sum_{i=1}^n L^2(X_i,Y_i,Z_i)
    -\bigl(\frac{1}{n}\sum_{i=1}^n L(X_i,Y_i,Z_i)\bigr)^2\bigr\}^{1/2}}
    \Biggr).
    \]

    \item The \textbf{Conditional Randomization Test} (CRT), proposed by \citet{candes2018panning}. Given a test statistic
    $T:\cup_{n=1}^\infty(\mathcal{X}^n\times\mathcal{Y}^n\times\mathcal{Z}^n)\to\mathbb{R}$
    and a sampler $\hat{P}_{X\mid Z}$ approximating the true $P_{X\mid Z}$, we draw
    $\bX^{(1)},\ldots,\bX^{(M)}\iidsim\hat{P}_{X\mid Z}$ and compute the $p$-value
    \[
    p_{\rm CRT}
    =
    \frac{1+\sum_{\ell=1}^M
    \One{T(\bX^{(\ell)},\bY,\bZ)\ge T(\bX,\bY,\bZ)}}{M+1}.
    \]
    In our experiments, we take $T(\bX,\bY,\bZ)=\bX^\top\bY$, and use isotonic distributional regression \citep{henzi2021isotonic} to estimate $\hat{P}_{X\mid Z}$ in order to preserve stochastic monotonicity.
\end{itemize}

\paragraph{Type~I error control.}
We evaluate the Type~I error as the empirical rejection probability when $\beta_n=0$, averaged over $500$ Monte Carlo repetitions for each method.  Figure~\ref{fig:validity_general_expt} displays the Type~I error curves for each method as a function of the sample size. Consistent with Theorem~\ref{thm:main}, \texttt{PairSwap-ICI} maintains finite-sample Type~I error control across all sample sizes. Although Model~(4) exhibits mild violations of Assumption~\ref{asm:st} and thus falls outside the scope of Theorem~\ref{thm:main}, the empirical results indicate that the proposed \texttt{PairSwap-ICI} framework continues to perform reliably in practice. This suggests a degree of robustness to such deviations and its broader applicability beyond such setting.  By contrast, PCM and CRT provide only asymptotic validity guarantees, and we observe noticeable deviations from the nominal level~$\alpha$ in finite samples. GCM is also asymptotically valid; however, in these simulation settings, its Type~I error also remains below the nominal level across the sample sizes considered. 

\begin{figure}[!ht]
    \centering
    \fbox{\includegraphics[scale=0.52]{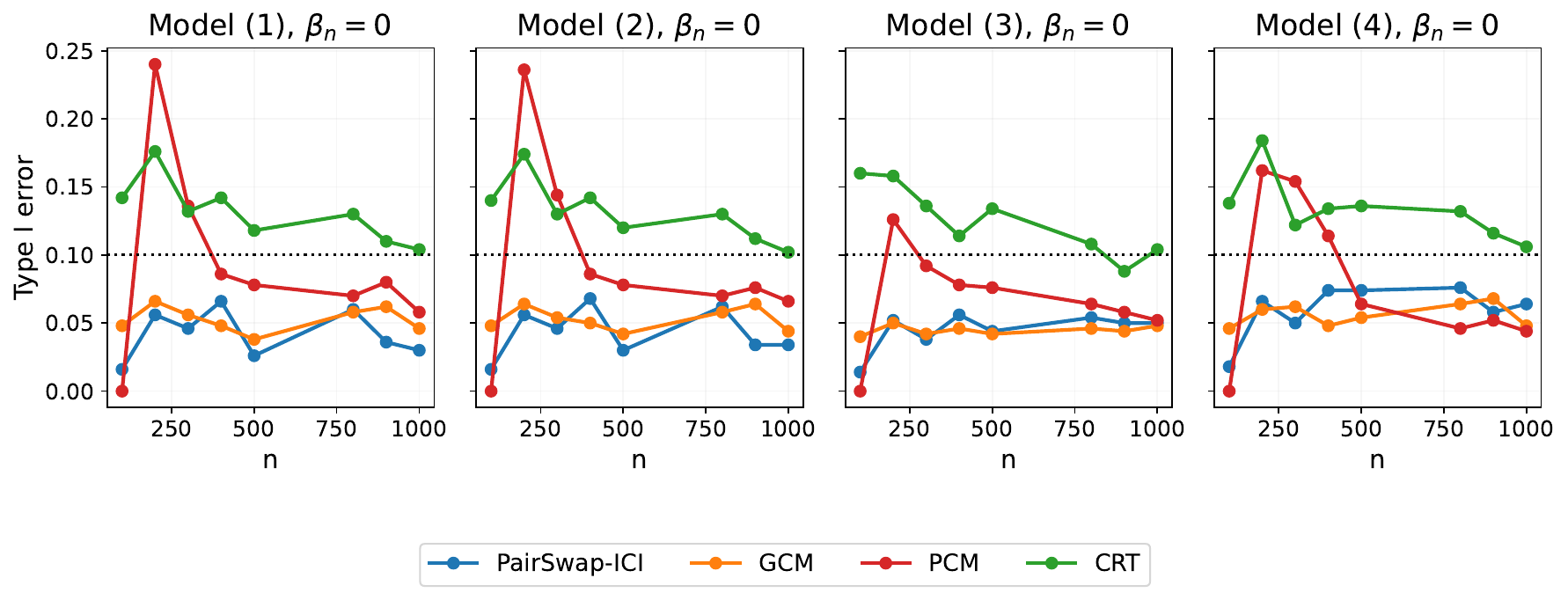}}
    \caption{Type~I error as a function of sample size under the simulation setting of Appendix~\ref{app:additional_experiments}. The nominal level $\alpha=0.1$ is indicated by the dotted black line.}
    \label{fig:validity_general_expt}
\end{figure}

\paragraph{Power.}

\begin{figure}[!ht]
    \centering
    \fbox{\includegraphics[scale=0.52]{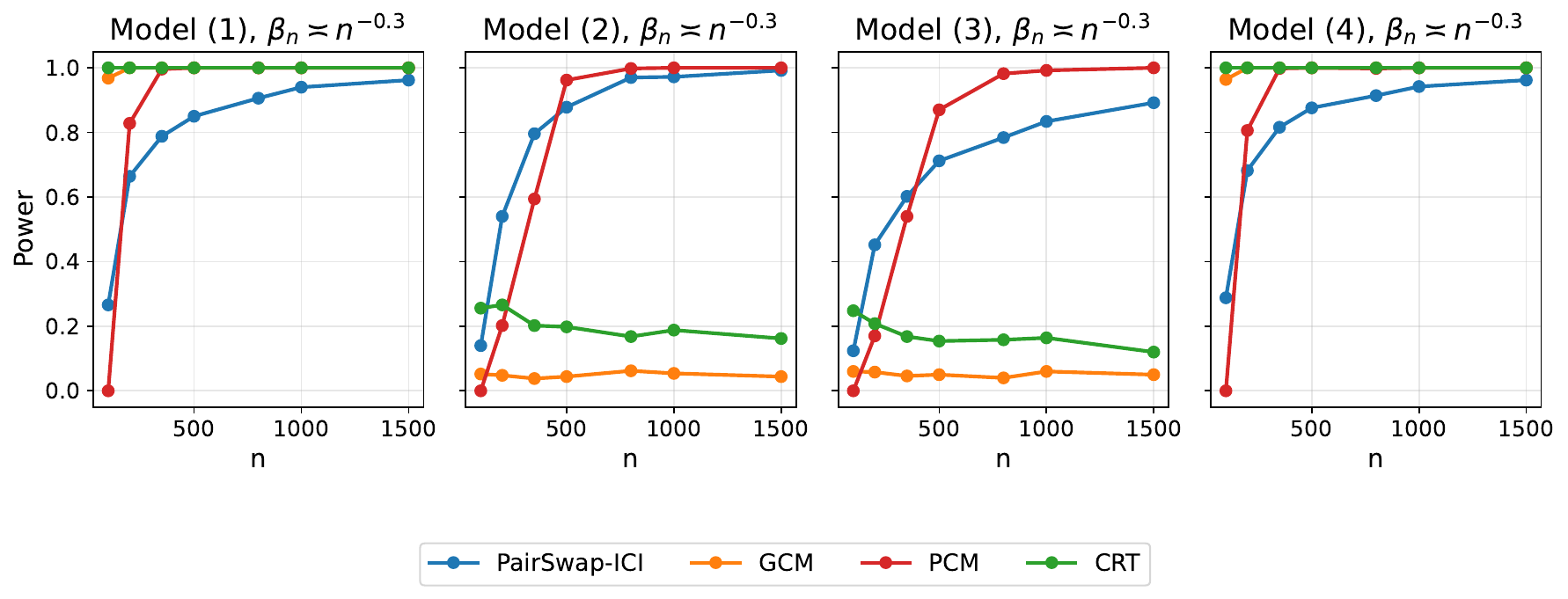}}
    \caption{Power as a function of sample size under the simulation models of Appendix~\ref{app:additional_experiments}, with $\beta_n=3\,n^{-0.3}$. Columns correspond to different data-generating models.}
    \label{fig:power-sim_general_expt}
\end{figure}

We next examine the power of the four methods under the alternatives described above, with $\beta_n=3\,n^{-0.3}$. Figure~\ref{fig:power-sim_general_expt} reports the power, evaluated as empirical rejection probability at different sample sizes, averaged over $500$ independent repetitions.

Model~(1) corresponds to the linear Gaussian setting analyzed in Section~\ref{sec:partial_linear_model}.  By Theorem~\ref{thm:neighbour_matching} and~\ref{thm:cross_bin_matching}, when $\beta_n \gg n^{-1/2}$, the power of \texttt{PairSwap-ICI} converges to $1$, as confirmed by the left panel of Figure~\ref{fig:power-sim_general_expt}. Although competing methods attain power close to $1$ at smaller sample sizes in this model, it is important to note several aspects: first, the CRT method does not control Type~I error at these smaller sizes (see the left panel of Figure~\ref{fig:validity_general_expt}); second, in model~(1), $\beta_n = \EE{\mathrm{Cov}(X,Y\mid Z)}$, which is the target of the GCM test statistic.  In this sense, the GCM method is tailored to this setting.  The drawback of testing the narrower null hypothesis  $H_0^{\mathrm{GCM}}: \EE{\mathrm{Cov}(X,Y\mid Z)}=0$ is illustrated in models~(2) and~(3), where $H_0^{\mathrm{GCM}}$ continues to hold despite the fact that conditional independence is violated.  Thus, the GCM has no power in these models, as can be seen in the middle two panels of Figure~\ref{fig:power-sim_general_expt}.  The CRT suffers a similar drawback, while the PCM has strong power in both models~(2) and~(3), since the departures from conditional independence also represent violations of conditional mean independence, and are therefore favorable settings for the PCM.  Nevertheless, the \texttt{PairSwap-ICI} procedure has competitive power with the PCM in both models~(2) and~(3).  Lastly, model~(4) closely mimics model~(1), and therefore the performances of all the methods are similar under models~(1) and (4).

In summary, the \texttt{PairSwap-ICI} procedure is the only one of the four methods studied that both retains finite-sample Type~I error control in all examples where the stochastic monotonicity holds true, and has competitive power in all settings considered.

\end{document}